\documentclass[a4paper,UKenglish,cleveref, autoref, thm-restate]{lipics-v2021}

\usepackage{amsmath}
\usepackage{mathpartir}
\usepackage{amsfonts}
\usepackage{amsthm}
\usepackage[dvipsnames]{xcolor}
\usepackage{tikz}
\usepackage{tikz-cd}
\usepackage[capitalise]{cleveref}
\usepackage{cancel}
\usetikzlibrary{automata,
arrows,
positioning,
shapes.geometric,
shapes,
fit,
calc,
positioning,
commutative-diagrams,
backgrounds}
\usepackage{mathpartir}
\usepackage{relsize}
\usepackage{float}
\usepackage[noend]{algorithm2e}
\usepackage{epsdice}
\usepackage{scalerel}
\usepackage{mathrsfs}
\newcommand{\tikzxmark}{%
\tikz[scale=0.17] {
    \draw[line width=0.7,line cap=round] (0,0) to [bend left=6] (1,1);
    \draw[line width=0.7,line cap=round] (0.2,0.95) to [bend right=3] (0.8,0.05);
}}
\DeclareMathOperator*{\bbigplus}{\scalerel*{+}{\bigoplus}}
\newcommand{\V}{\mathsf{Out}}
\newcommand{\T}{\mathsf{T}}
\newcommand{\Act}{\mathsf{Act}}
\newcommand{\At}{\mathsf{At}}
\newcommand{\D}{\mathcal{D}_\omega}
\newcommand{\G}{\mathcal{G}}
\newcommand{\Exp}{\mathtt{Exp}}
\newcommand{\Pexp}{\mathtt{PExp}}
\newcommand{\Bexp}{\mathtt{BExp}}
\newcommand{\accept}{\checkmark}
\newcommand{\reject}{\tikzxmark}
\newcommand{\expr}[1]{ \mathsf{#1}}
\newcommand{\action}[1]{\mathsf{#1}}
\newcommand{\test}[1]{{\color{MidnightBlue} \mathsf{#1}}}
\newcommand{\prob}[1]{{{#1}}}
\newcommand{\seq}{\mathbin{;}}
\newcommand{\zero}{\mathsf{0}}
\newcommand{\one}{\mathsf{1}}
\newcommand{\acro}[1]{\(\mathsf{#1}\)}
\newcommand{\trmt}[1]{\mathsf{E}\left(#1\right)}
\newcommand{\var}[1]{ {{#1} }}
\newcommand{\beq}{\equiv_{\mathsf{BA}}}
\newcommand{\bleq}{\leq_{\mathsf{BA}}}
\newcommand{\eqc}[2]{[#1]_{#2}}
\newcommand{\Set}{\mathtt{Set}}
\newcommand{\B}{\mathcal{B}}
\newcommand{\coalg}[1]{\mathsf{Coalg}_{#1}}
\newcommand{\supp}{\mathtt{supp}}
\newcommand{\cogen}[2]{\langle#1\rangle_{#2}}
\newcommand{\2}{\mathsf{2}}
\newcommand{\Id}{\mathsf{Id}}
\newcommand{\beh}{!}
\newcommand{\semseq}{\lhd}
\newcommand{\gfp}{\mathsf{gfp}}
\newcommand{\id}{\mathtt{id}}
\newcommand{\bisim}{\sim}
\newcommand{\bN}{\mathbb{N}}
\newcommand{\bR}{\mathbb{R}}
\newcommand{\eR}{\bR^{+}_{\infty}}
\newcommand{\bigsum}[1]{\bbigplus\limits_{#1}}

\newcommand{\bigplus}[1]{\bigoplus_{#1}}
\newcommand{\ex}{\mathtt{exp}}
\newcommand{\sys}{\mathtt{sys}}
\newcommand{\uaequiv}{\mathrel{\dot\equiv}}
\newcommand{\boplus}{\; \mathlarger{{\oplus}} \;}
\newcommand{\bsum}{\; \mathlarger{{+}} \;}
\newcommand{\customlabel}[2]{%
   \protected@write \@auxout {}{\string \newlabel {#1}{{#2}{\thepage}{#2}{#1}{}} }
   \hypertarget{#1}{#2}%
}

\usepackage{comment}

\newcommand{\mset}[1]{\ensuremath{\{\!|#1|\!\}}}

\SetKw{kwTrue}{true}
\SetKw{kwReturn}{return}
\SetKw{kwSkip}{skip}
\SetKwProg{Def}{def}{:}{}

\pdfoutput=1 
\hideLIPIcs  

\title{Probabilistic Guarded KAT Modulo Bisimilarity: Completeness and Complexity} 

\author{Wojciech Różowski}{Department of Computer Science, University College London, United Kingdom \and \url{https://wkrozowski.github.io} }{w.rozowski@cs.ucl.ac.uk}{https://orcid.org/0000-0002-8241-7277}{}
\author{Tobias Kapp{\'e}}{Open Universiteit, Heerlen, The Netherlands \and ILLC, University of Amsterdam, The Netherlands \and \url{https://tobias.kap.pe}}{tobias.kappe@ou.nl}{http://orcid.org/0000-0002-6068-880X}{}

\author{Dexter Kozen}{Department of Computer Science, Cornell University, Ithaca, NY, USA \and \url{https://www.cs.cornell.edu/~kozen/}}{kozen@cs.cornell.edu}{https://orcid.org/0000-0002-8007-4725}{}

\author{Todd Schmid}{Department of Computer Science, University College London, United Kingdom \and \url{https://toddtoddtodd.net} }{todd.schmid.19@ucl.ac.uk}{https://orcid.org/0000-0002-9838-2363}{}

\author{Alexandra Silva}{Department of Computer Science, Cornell University, Ithaca, NY, USA \and \url{https://alexandrasilva.org}}{alexandra.silva@cornell.edu}{https://orcid.org/0000-0001-5014-9784}{}

\authorrunning{W. Różowski, T. Kapp{\'e}, D. Kozen, T. Schmid, A. Silva} 

\Copyright{Wojciech Różowski, Tobias Kapp{\'e}, Dexter Kozen, Todd Schmid, Alexandra Silva} 



\category{Track B:\@ Automata, Logic, Semantics, and Theory of Programming} 

\funding{This work was partially supported by ERC grant Autoprobe (no. 101002697; Różowski, Schmid and Silva), the EU's Horizon 2020 research and innovation program under Marie Skłodowska-Curie grant VERLAN (no. 101027412; Kapp\'e), and NSF grant CCF-2008083 (Kozen).}

\acknowledgements{}

\nolinenumbers

\EventEditors{Kousha Etessami, Uriel Feige, and Gabriele Puppis}
\EventNoEds{3}
\EventLongTitle{50th International Colloquium on Automata, Languages, and Programming (ICALP 2023)}
\EventShortTitle{ICALP 2023}
\EventAcronym{ICALP}
\EventYear{2023}
\EventDate{July 10--14, 2023}
\EventLocation{Paderborn, Germany}
\EventLogo{}
\SeriesVolume{261}
\ArticleNo{113}

\begin{document}

\maketitle

\begin{abstract}
We introduce Probabilistic Guarded Kleene Algebra with Tests (\acro{ProbGKAT}), an extension of \acro{GKAT} that allows reasoning about uninterpreted imperative programs with probabilistic branching.
We give its operational semantics in terms of special class of probabilistic automata.
We give a sound and complete Salomaa-style axiomatisation of bisimilarity of \acro{ProbGKAT} expressions.
Finally, we show that bisimilarity of \acro{ProbGKAT} expressions can be decided in \(O(n^3 \log n)\) time via a generic partition refinement algorithm.
\end{abstract}

\section{Introduction}
Randomisation is an important feature in the design of efficient algorithms, cryptographic protocols, and stochastic simulation~\cite{Barthe:2020:Foundations}.
For a simple example of randomisation, imagine simulating a three-sided die~\cite{threesideddie}.
There are at least two ways to do this:
\begin{itemize}
    \item
    A reference implementation could use a fair coin and a biased coin with probability $\frac{1}{3}$ of landing on \texttt{heads}:
    Toss the biased coin first.
    If it lands on \texttt{heads}, return $\epsdice{1}$, and otherwise toss the fair coin and return $\epsdice{2}$ if it lands on \texttt{heads} or $\epsdice{3}$ otherwise.
    \item
    Another way to do this is with two consecutive tosses of a fair coin: if the outcome is \texttt{heads-heads}, then return $\epsdice{1}$; if it is \texttt{heads-tails}, return $\epsdice{2}$; if it is \texttt{tails-heads}, return $\epsdice{3}$; and if it is \texttt{tails-tails}, repeat the process~\cite{Knuth:1976:Complexity}.
\end{itemize}
These programs can be written using a function $\texttt{flip}(p)$ that returns $\texttt{true}$ (\texttt{heads}) with probability $\prob{p}$, and $\texttt{false}$ (\texttt{tails}) with probability $\prob{1-p}$, see \cref{fig:3die}.
If we can prove that those programs are equivalent, then we can be certain they implement the same distribution.

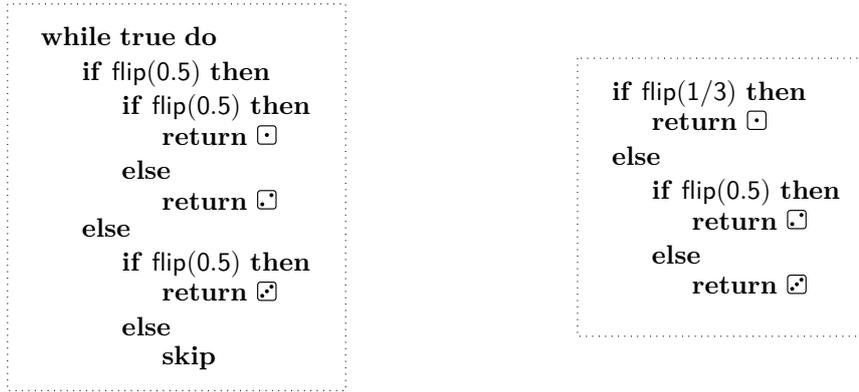
\begin{figure}[t]
    \begin{center}
        \begin{tikzpicture}
            \node(p1){
                \begin{minipage}{60mm}
                    \begin{algorithm}[H]
                        \SetAlgoNoLine
                        \While{\kwTrue}{
                            \eIf{$\mathsf{flip(0.5)}$}
                            {
                                \eIf{$\mathsf{flip(0.5)}$}
                                {\kwReturn{$\epsdice{1}$}}
                                {\kwReturn{$\epsdice{2}$}}
                            }
                            {
                                \eIf{$\mathsf{flip(0.5)}$}
                                {\kwReturn{$\epsdice{3}$}}
                                {\kwSkip}
                            }
                        }
                    \end{algorithm}
                \end{minipage}
            };
            \node[right=7cm of p1.center,anchor=center] (p2) {
                \begin{minipage}{50mm}
                    \begin{algorithm}[H]
                        \SetAlgoNoLine
                        \eIf{$\mathsf{flip(1/3)}$}
                        {\kwReturn{$\epsdice{1}$}}
                        {\eIf{$\mathsf{flip(0.5)}$}
                        {\kwReturn{$\epsdice{2}$}}
                        {\kwReturn{$\epsdice{3}$}}
                        }
                    \end{algorithm}
                    \vspace{0.8em}
                \end{minipage}
            };
            \draw[dotted] ($(p1.north west) + (0.2, 0)$) rectangle ($(p1.south east) + (-1.6, -0.05)$);
            \draw[dotted] ($(p2.north west) + (0.2, 0)$) rectangle ($(p2.south east) + (-1.2, -0.05)$);
        \end{tikzpicture}  
    \end{center}

    \caption{On the left, the Knuth-Yao process to simulate a three-sided die with two throws of a fair coin. On the right, direct implementation of the three-sided die (note that the probability of the first else-branch is $\frac{2}{3}$, hence both $\epsdice{2}$ and $\epsdice{2}$ are returned with probability $\frac{1}{3}$).}%
    \label{fig:3die}
\end{figure}
In this paper, we introduce Probabilistic \acro{GKAT} (\acro{ProbGKAT}), a language based on Guarded Kleene Algebra with Tests (\acro{GKAT})~\cite{Kozen:2008:Bohm,Smolka:2020:Guarded,Schmid:2021:Guarded} augmented with extra  constructs for reasoning about such randomised programs.
The laws of \acro{GKAT} allow reasoning about the equivalence of uninterpreted programs with deterministic control flow in the form of Boolean branching (\texttt{if-then-else}) and looping (\texttt{while-do}) constructs.
\acro{GKAT} comes equipped with an automata-theoretic operational semantics, a \emph{nearly linear} decision procedure, and complete axiomatic systems for reasoning about trace equivalence~\cite{Smolka:2020:Guarded} and bisimilarity~\cite{Schmid:2021:Guarded} of expressions, both inspired by Salomaa's axiomatisation of Kleene Algebra~\cite{Salomaa:1966:Two}.

\acro{ProbGKAT} extends \acro{GKAT} with three new syntactic constructs:
(1)~a probabilistic choice operator, representing branching based on a (possibly biased) coin flip;
(2)~a probabilistic loop operator, representing a generalised Bernoulli process; and
(3)~return values, which allow a limited form of non-local control flow akin to \texttt{return} statements in imperative programming.

The main focus of this paper is the problem of axiomatising bisimilarity of \acro{ProbGKAT} expressions.
We build on an inference system for reasoning about bisimilarity of \acro{GKAT} expressions~\cite{Schmid:2021:Guarded}, which includes a generalisation of Salomaa's axiomatisation of the Kleene star~\cite{Salomaa:1966:Two} called the Uniqueness of Solutions axiom (\acro{UA}), also known in the process algebra community as Recursive Specification Principle (\acro{RSP})~\cite{Bergstra:1985:Verification}.
In the presence of both Boolean guarded and probabilistic branching, axiomatisation becomes significantly more involved.
Besides adding intuitive rules governing the behaviour of probabilistic choice and loops, we add axioms capturing the interaction of both kinds of branching when combined with looping constructs.
Moreover, in the case of \acro{ProbGKAT}, showing the soundness of \acro{UA} becomes highly nontrivial.
We do so by exploiting the topological structure of the operational model, namely the \emph{behavioural pseudometric} associated with bisimilarity.
Despite the jump in difficulty, our completeness proof follows a similar strategy as the one for \acro{GKAT} modulo bisimilarity~\cite{Schmid:2021:Guarded}.

Our main contributions are as follows.
\begin{itemize}
    \item We provide an operational semantics of \acro{ProbGKAT} programs, which relies on a type of automata that have both Boolean guarded and probabilistic transitions (\cref{sec:opsem}).
    \item We concretely characterise bisimulations for our automata using an adaption of the flow network characterisation of bisimilarity of Markov chains~\cite{Jones:1990:Probabilistic,Bartels:2003:Hierarchy} (\cref{sec:bisim}).
    \item We give a sound and complete Salomaa-style axiomatisation of bisimulation equivalence of \acro{ProbGKAT} expressions (\cref{sec: axiomatisation,sec: completeness}).
    \item We show that, when the number of tests is fixed, bisimilarity of \acro{ProbGKAT} expressions can be efficiently decided in \(\mathcal{O}(n^3 \log n)\) time, where \(n\) is the total size of programs under comparison, by using \emph{coalgebraic partition refinement}~\cite{Deifel:2019:Generic, Wissmann:2020:Efficient} (\cref{sec:decidability}).
\end{itemize}
In \cref{sec: syntax} we define the syntax of \acro{ProbGKAT}.
We survey related work in \cref{sec:related}; conclusions and further work appear in \cref{sec:future}.
Proofs appear in the appendix.

\section{Syntax}\label{sec: syntax}

\acro{ProbGKAT} has a two-sorted syntax consisting of a set of \emph{expressions} \(\Exp\) that contains a set \(\Bexp\) of \emph{Boolean assertions} or \emph{tests}.
For a fixed finite set \(\T\) of \emph{primitive tests}, the syntax for tests is denoted \(\Bexp\) and generated by the grammar
\[
    \test{b}, \test{c} \in \Bexp ::= \test \zero \mid \test \one \mid t \in \T \mid \test{b} + \test{c} \mid \test{b}\test{c} \mid \Bar{\test{b}}
\]
Here, \(0\) and \(1\) respectively denote false and true, \(\bar{\test{\cdot}}\) denotes negation, \(+\) is disjunction, and juxtaposition is conjunction.
Let \(\beq\) denote Boolean equivalence in \(\Bexp\).  Entailment is a preorder on \(\Bexp\) given by \(\test{b} \bleq \test{c} \iff \test{b} + \test{c} \beq \test{c}\).  The quotient of \(\Bexp\) by \(\beq\) is the free Boolean algebra on the set of generators \(\T\), in which entailment---\(\eqc{\test{b}}{\beq} \leq \eqc{\test{c}}{\beq} \iff \test{b} \bleq \test{c}\)---is a partial order, with bottom and top elements being the equivalence classes of \(\test{\zero}\) and \(\test{\one}\) respectively. The minimal non-zero elements of that partial order are called \emph{atoms}, and we will use \(\At\) to denote the set of atoms.
For fixed sets \(\Act\) of \emph{atomic actions} and \(\V\) of \emph{return values}, the set \(\Exp\) of \acro{ProbGKAT} expressions is defined by the grammar in \Cref{fig:probgkat}.
\begin{figure}[t]
\begin{align*}
    \expr{e}, \expr{f} \in \Exp ::= &\hspace{0.83em}\action{p} \in \Act && \mathsf{do}~\action{p}\\
    &\mid \test{b} \in \Bexp && \mathsf{assert}~\test{b}\\
    &\mid \expr{e} +_{\test{b}} \expr{f} && \mathsf{if}~\test{b}~\mathsf{then}~\expr{e}~\mathsf{else}~\expr{f}\\
    &\mid \expr{e} \seq \expr{f} \\
    &\mid \expr{e}^{(\test{b})} && \mathsf{while}~\test{b}~\mathsf{do}~\expr{e}\\
    &\mid \var{v} \in \V && \mathsf{return}~\var{v}\\
    &\mid \expr{e} \oplus_{\prob{r}} \expr{f} && \mathsf{if}~\mathsf{flip(}\prob{r}\mathsf{)}~\mathsf{then}~\expr{e}~\mathsf{else}~\expr{f}\\ 
    &\mid \expr{e}^{[\prob{r}]} && \mathsf{while}~\mathsf{flip(}\prob{r}\mathsf{)}~\mathsf{do}~\expr{e} 
   \end{align*}
\caption{Syntax of \acro{ProbGKAT}.}\label{fig:probgkat}
\end{figure}

The syntax of \acro{GKAT} is captured by the first five cases in \cref{fig:probgkat}, and so is a proper fragment of \acro{ProbGKAT}.
There are three new constructs: \emph{return values}, \emph{probabilistic choices}, and \emph{probabilistic loops}.
Return values behave like return statements in imperative programs, introducing a form of non-local control flow.
The probabilistic choice \(\expr{e} \oplus_\prob{r} \expr{f}\) flips a biased coin with real bias $\prob r\in[0,1]$ and depending on the outcome runs \(\expr{e}\) with probability \(\prob{r}\) and \(\expr{f}\) with probability \(\prob{1-r}\).
The probabilistic loop \(\expr{e}^{[\prob{r}]}\) also begins with a biased coin flip, and depending on the outcome it either executes \(\expr{e}\) and starts again (probability $r$) or terminates (probability $1 -r$).
A probabilistic loop can be regarded as a generalised Bernoulli processes.

\begin{example}%
    \label{ex:3die}
    Recall the two programs from the introduction (\cref{fig:3die}): one directly implementing a 3-sided die and the other simulating a 3-sided die with a fair coin.
    We can express these programs using three output values, $\epsdice{1}$, $\epsdice{2}$, and $\epsdice{3}$, to model the possible outcomes of the three-sided die.
    The first program is the infinite while loop $\expr{f} = \expr{g}^{(\test{\one})}$, where the loop body is given by \(\expr{g} = (\epsdice{1} \oplus_{\prob{\frac{1}{2}}} \epsdice{2}) \oplus_{\prob{\frac{1}{2}}} (\epsdice{3} \oplus_{\prob{\frac{1}{2}}} \test{\one})\).
    In \(\expr{f}\), \(\epsdice{1}\) represents \texttt{heads-heads}, \(\epsdice{2}\) is \texttt{heads-tails}, \(\epsdice{3}\) is \texttt{tails-heads}, and \texttt{tails-tails} prompts a rethrow.
    The second program encodes the \acro{ProbGKAT} expression \(\expr{e} = \epsdice{1} \oplus_{\prob{\frac{1}{3}}}(\epsdice{2} \oplus_{\prob{\frac{1}{2}}} \epsdice{3})\).
\end{example}

\section{Operational semantics}\label{sec:opsem}
In this section, we formally introduce \acro{ProbGKAT} automata, the operational models of \acro{ProbGKAT} expressions.
We associate a \acro{ProbGKAT} automaton with each expression via a small-step semantics inspired by Brzozowski derivatives~\cite{Brzozowski:1964:Derivatives}.
As we will see, the biggest hurdle is the semantics of the probabilistic loop.
Before we provide our small-step semantics, we introduce the notation and operations on probability distributions that we will need.

\subparagraph*{Preliminary definitions. }
A function \(\nu: X \to [0,1]\) is called a \emph{(probability) distribution on \(X\)} if it satisfies \(\sum_{x \in X} \nu(x) = 1\).
In case \(\sum_{x \in X} \nu(x) \le 1\) we call \(\nu\) a \emph{subprobability distribution}, or \emph{subdistribution}.
Every (sub)distribution \(\nu\) in this paper is \emph{finitely supported}, which means that the set \(\supp(\nu) = \{x \in X \mid \nu(x) > 0 \}\) is finite.
Given \(A \subseteq X\), we define \(\nu[A]=\sum_{x \in A} \nu(x)\).
This sum is well-defined because only finitely many summands have non-zero probability.

We use \(\D(X)\) to denote the set of finitely supported probability distributions on the set \(X\).
A function \(f : X \to Y\) can be \emph{lifted} to a map \(\D(f) : \D(X) \to \D(Y)\) between distributions by setting \(\D(f)(\nu) = \nu[f^{-1}(y)]\).
Given \(x \in X\), its \emph{Dirac delta} is the distribution \(\delta_x\); here \(\delta_x(y)\) is equal to $1$ when $x=y$, and $0$ otherwise.
Given \(f : X \to \D (Y)\), there is a unique map \(\bar{f} \colon \D (X) \to \D (Y)\) such that \(f = \bar{f} \circ \delta\), called the \emph{convex extension of \(f\)}, and explicitly given by \(\bar{f}(\nu)(y) = \sum_{x \in X} \nu(x) f(x)(y) \).

When $\nu, \mu: X \to [0,1]$ are probability distributions and $r \in [0,1]$, we write $r\nu + (1-r)\mu$ for the \emph{convex combination} of $\nu$ and $\mu$, which is the probability distribution given by $(r\nu + (1-r)\mu)(x) = r\nu(x) + (1-r)\mu(x)$; this operation preserves finite support.

\subparagraph*{Operational model. }
Operationally, \acro{ProbGKAT} expressions denote states in a transition system called a \emph{\acro{ProbGKAT} automaton}.
Below, we write \(\2=\{\reject, \accept\}\) for a two element set of symbols denoting rejection and acceptance respectively.

\begin{definition}\label{def:transition_system}
    A \emph{\acro{ProbGKAT} automaton} is a pair \((X, \beta)\) consisting of set of \emph{states} \(X\) and a \emph{transition function} \(\beta : X \times \At \to \D (\2 + \V + \Act \times X)\).
\end{definition}

A state in a \acro{ProbGKAT} automaton associates each Boolean atom $\alpha \in \At$ (capturing the global state of the Boolean variables) with a finitely supported probability distribution over several possible outcomes.
One possible outcome is \emph{termination}, which ends execution and either signals success (\(\accept\)) or failure (\(\reject\)), or returns an output value ($\var{v} \in \V$).
The other possible outcome is \emph{progression}, performing an action ($\action{p} \in \Act$) and transitioning to a state.

\begin{example}\label{ex:automaton}
    Let \(X = \{x_1, x_2\}\), \(\At = \{\alpha, \beta\}\), \(\Act = \{\action{p}, \action{q}\}\) and \(\V = \{v\}\).
    On the right, there is a definition of a transition function \(\tau: X \times \At \to \D(\2 + \V + \Act \times X)\), while on the left there is a visual representation of \((X, \tau)\).
    Given a state \(x \in X\) and an atom \(\alpha \in \At\), we write \(\tau(x)_\alpha\) rather than \(\tau(x)(\alpha)\).
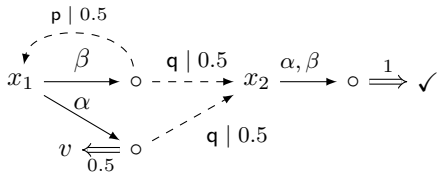
\[
        \begin{tikzpicture}
        \useasboundingbox (0,-1) rectangle (10.2,1);
            \node(0) {$x_1$};
            \node (1) [right=1cm of 0]{$\circ$};
            \node (2) [below=.5cm of 1]{$\circ$};
            \draw (0) edge[-latex] node[above] {\(\beta\)} (1);
            \draw (0) edge[-latex] node[above] {\(\alpha\)} (2);
            \path[-latex,dashed] (1) edge[bend right=80]  node[above] {\scriptsize{\(\action{p} \mid \prob{0.5}\)}} (0);
            \node(3) [right=2.5cm of 0]{\(x_2\)};
            \path[-latex,dashed] (2) edge  node[below right, pos=0.55] {\footnotesize{\(\action{q} \mid \prob{0.5}\)}} (3);
            \path[-latex,dashed] (1) edge  node[above, pos=0.55] {\footnotesize{\(\action{q} \mid \prob{0.5}\)}} (3);
            \node (4) [right=.75cm of 3]{$\circ$};
            \draw (3) edge[-latex] node[above, pos=0.35] {\footnotesize\(\alpha, \beta\)} (4);
            \node(5) [right=.5cm of 4]{$\accept$};
            \draw (4) edge[-implies, double, double distance=0.5mm] node[above] {\scriptsize\(\prob{1}\)} (5);
            \node(6) [left=.5cm of 2] {$v$};
            \draw (2) edge[-implies, double, double distance=0.5mm] node[below] {\scriptsize\(\prob{0.5}\)} (6);
            \node[right=.1cm of 5,align=left,yshift=2mm] {
                \begin{minipage}{0.65\textwidth}
                \small
                \begin{align*}
                    &\tau(x_1)_\alpha = \prob{\frac{1}{2}}\delta_{(\action{q}, x_2)} + \prob{\frac{1}{2}}\delta_{\var{v}}\\[1ex]
                    & \tau(x_1)_\beta = \prob{\frac{1}{2}}\delta_{(\action{p}, x_1)} + \prob{\frac{1}{2}}\delta_{(\action{q}, x_2)}\\[1ex]
                    &\tau(x_2)_\alpha=\tau(x_2)_\beta=\delta_{\accept}
                \end{align*}
                \end{minipage}
            };
        \end{tikzpicture}
\]
    We use solid lines annotated with (sets of) atoms to denote Boolean guarded branching, dashed lines annotated with atomic actions and probabilities to denote probabilistic labelled transitions to a next state, and double bar arrows pointing at elements of \(\2 + \V\) annotated with probabilities to denote probabilistic transitions that result in termination or output.
\end{example}

The following notions of homomorphism and bisimulation describe structure-preserving maps and relations between \acro{ProbGKAT} automata.

\begin{definition}\label{def:homomorphism_of_transition_systems}
    A \emph{homomorphism} between \acro{ProbGKAT} automata \((X, \beta)\) and \((Y, \gamma)\) is a function \(f : X \to Y\) satisfying for all \(x \in X\) and \(\alpha \in \At\)
    \begin{enumerate}
        \item For any \(o \in \2 + \V\), \(\gamma(f(x))_\alpha(o) = \beta(x)_\alpha(o)\)
        \item For any \((\action{p}, y) \in \Act \times Y\), \(\gamma(f(x))_\alpha(\action{p}, y) = \beta(x)_\alpha[\{\action{p}\} \times f^{-1}(y)]\)
    \end{enumerate}
\end{definition}

\begin{definition}\label{def:bisimulation_of_transition_systems}
    Let \((X, \beta)\) and \((Y, \gamma)\) be \acro{ProbGKAT} automata and let \(R \subseteq X \times Y\) be a relation. \(R\) is a \emph{bisimulation} if there exists a transition function \(\rho: R \times \At \to \D(\2 + \V + \Act \times R)\) such that projection maps \(\pi_1 : R \to X\) and \(\pi_2 : R \to Y\) given by \(\pi_1(x,y)=x\) and \(\pi_2(x,y)=y\) are homomorphisms from \((R, \rho)\) to \((X, \beta)\) and \((Y, \gamma)\) respectively.
\end{definition}

\begin{remark}
    \cref{def:homomorphism_of_transition_systems,def:bisimulation_of_transition_systems} are direct translations from the coalgebraic theory of \acro{ProbGKAT} automata (see \cref{apx:coalgebra}).
    Coalgebra plays a central role in our proofs, but for purposes of exposition it does not appear in the body of the present paper.
\end{remark}

\subparagraph*{Brzozowski construction. }

 \acro{ProbGKAT} expressions can be endowed with an operational semantics in the form of a \acro{ProbGKAT} automaton \(\partial: \Exp \times \At \to \D(\2 + \V + \Act \times \Exp)\), which we refer to as the \emph{Brzozowski derivative}, as it is reminiscent of the analogous construction for regular expressions and deterministic finite automata due to Brzozowski~\cite{Brzozowski:1964:Derivatives}.

Given \(\alpha \in \At\), \(\expr{e}, \expr{f} \in \Exp\), \(\test{b} \in \Bexp\), \(\var{v} \in \V\), \(\prob{r} \in [0,1]\), and \(\action{p} \in \Act\), we define
\begin{mathpar}
    \partial(\test{b})_\alpha = \begin{cases} \delta_\accept & \alpha \bleq \test{b}\\\delta_{\reject} & \alpha\bleq\test{\bar{b}}\end{cases}
    \and
    \partial(\expr{e} +_\test{b} \expr{f})_\alpha = \begin{cases} \partial(\expr{e})_\alpha & \alpha \bleq \test{b}\\\partial(\expr{f})_\alpha & \alpha\bleq\test{\bar{b}}\end{cases}
    \and
    \partial(\var{v})_\alpha = \delta_{\var{v}} \quad\quad\quad\quad \partial(\action{p})_\alpha = \delta_{(\action{p}, \test{\one})}
    \and

    \partial(\expr{e} \oplus_{\prob{r}} \expr{f})_\alpha=\prob{r}\partial(\expr{e})_\alpha + \prob{(1-r)}\partial(\expr{f})_\alpha
\end{mathpar}
The derivatives of sequential composition and loops are defined below.
The outgoing transitions of \(\test{b} \in \Bexp\) depend on whether or not the input atom \(\alpha \in \At\) satisfies \(\test{b}\), either outputting \(\accept\) (success) or \(\reject\) (abort) with probability \(\prob{1}\).
The outgoing transitions of a guarded choice \(\expr{e} +_\test{b}\expr{f}\) consist of the outgoing transitions of \(\expr{e}\) labelled by atoms satisfying \(\test{b}\) and the outgoing transitions of \(\expr{f}\) labelled by atoms satisfying \(\test{\bar b}\) (as in \acro{GKAT}).
The output value \(v \in \V\) returns the value \(v\) with probability \(\prob{1}\) given any input atom.
The atomic action \(\action{p} \in \Act\) emits \(\action{p}\) given any input atom and transitions to the expression \(\test{\one}\).
The outgoing transitions of the probabilistic choice \(\expr e \oplus_{\prob r} \expr f\) consist of the outgoing transitions of \(\expr{e}\) with probabilities scaled by \(r\) and the outgoing transitions of \(\expr{f}\) scaled by \(1 - r\).

\smallskip
The behaviour of the sequential composition $\expr{e} \seq \expr{f}$ is more complicated.
We need to factor in the possibility that $\expr{e}$ may accept with some probability \(\prob{t}\) given an input atom \(\alpha\), in which case the \(\alpha\)-labelled outgoing transitions of \(f\) contribute to the outgoing transitions of \(\expr e \seq \expr f\).
Formally, we write \(\partial(\expr{e} \seq \expr{f})_\alpha = \partial(\expr{e})_\alpha \semseq_\alpha \expr{f}\), where given \(\alpha \in \At\) and \(\expr{f} \in \Exp\) we define \(( - \semseq_\alpha \expr{f}) : \D(\2 + \V + \Act \times \Exp) \to \D(\2 + \V + \Act \times \Exp)\) to be the convex extension of \(c_{\alpha, \expr{f}} : \2 + \V + \Act \times \Exp \to \D(\2 + \V + \Act \times \Exp)\) given below on the left.
\begin{gather*}
    c_{\alpha, \expr{f}}(x) =
    \begin{cases}
        \delta_x & x \in \{\reject\} \cup \V \\
        \partial(\expr{f})_\alpha & x = \accept\\
        \delta_{(\action{p}, \expr{e}'\seq\expr{f})} & x = (\action{p}, \expr{e'})\\
    \end{cases}
    \quad\quad\quad\quad
    \begin{tikzpicture}[baseline=-5ex]
        \node (0) {$\expr{e}\seq\expr{f}$};
        \node (1) [below=.75cm of 0] {$\circ$};
        \draw (0) edge[-latex] node[left, pos=0.5] {\(\alpha\)} (1);
        \node (4) [right=1.2cm of 1] {$\expr{e}'{\color{Maroon} {} \seq \expr{f}}$};
        \node (5) [left=1.2cm of 1] {{\color{gray} \bcancel\checkmark}};
                \node (5a) [left=-.2cm of 5] {{\color{Maroon} \(\partial(\expr{f})_\alpha\)}};
        \draw (1) edge[-implies, double, double distance=0.5mm] node[above] {\(\prob{t}\)} (5);
        \draw (1) edge[-latex, dashed] node[above] {\footnotesize\( \action{p} \mid \prob{s}\)} (4);
    \end{tikzpicture}
\end{gather*}

Intuitively, $c_{\alpha,\expr{f}}$ reroutes the transitions coming out of $\expr{e}$: acceptance (the second case) is replaced by the behaviour of $\expr{f}$, and the probability mass of transitioning to $\expr{e}'$ (the third case) is reassigned to $\expr e\seq\expr{f}$.
The branches that output the elements of \(\{\reject\} + \V\) are unchanged by this operation.
A pictorial representation of the effect on the derivatives of $\expr{e} \seq \expr{f}$ is given above on the right.
Here, we assume that \(\partial(\expr{e})_\alpha\) can perform a \(\action{p}\)-transition to \(\expr{e'}\) with probability \(\prob{s}\); we make the same assumption in the informal descriptions of derivatives for loops, below.

\smallskip
For guarded loops, we consider three cases when defining \(\partial\left(\expr{e}^{(\test{b})}\right)_\alpha\).
If \(\alpha \bleq \test{\bar{b}}\), then the current state does not satisfy the loop guard and can be skipped: \(\partial\left(\expr{e}^{(\test{b})}\right)_\alpha=\delta_{\accept}\).
If \(\alpha \bleq \test{b} \) and \(\partial(\expr{e})_\alpha(\accept)=\prob{1}\), then the loop body is called, but the inner program $\expr{e}$ does not perform actions.
 We identify divergent loops with rejection %
and so in this case we set \(\partial\left(\expr{e}^{(\test{b})}\right)_\alpha=\delta_{\reject}\).
If  \(\alpha \bleq \test{{b}}\) and \(\partial(\expr{e})_\alpha(\accept)<\prob{1}\), the program executes the loop body and starts again, having to redistribute the probability mass of immediate acceptance \(\partial(\expr{e})_\alpha(\accept)\) through each execution.
So,  for \(\alpha \bleq \test{{b}}\) and \(\partial(\expr{e})_\alpha(\accept)<\prob{1}\), the definition of  $\partial(\expr{e}^{(\test{b})})_\alpha$ is given below on the left: it rejects or returns when $\expr{e}$ does, and transitions to $\expr{e'} \seq \expr{e}^{(b)}$ when $\expr{e}$ transitions to $\expr{e'}$.
\begin{gather*}
    \partial(\expr{e}^{(\test{b})})_\alpha(x)=\begin{cases} \frac{\partial(\expr{e})_\alpha(x)}{\prob{1}-\partial(\expr{e})_\alpha(\accept)} & x \in \{\reject\}\cup\V \\ \frac{\partial(\expr{e})_\alpha(\action{p}, \expr{e'})}{\prob{1}-\partial(\expr{e})_\alpha(\accept)} & x = \left(\action{p},\left( \expr{e'}\seq\expr{e}^{(\test{b})}\right)\right)\\
        \prob{0} &\text{otherwise} \\\end{cases}
    \quad\quad
    \begin{tikzpicture}[baseline=-5ex]
        \node (0) {$\expr{e}^{(\test{b})}$};
        \node (1) [below=.7cm of 0] {$\circ$};
        \node (2) [right=2cm of 0] {$\circ$};
        \node (3) [right = 0.65cm of 2] {\checkmark};
        \draw (2) edge[-implies, double, double distance=0.5mm] node[above] {\(1\)} (3);
        \draw (0) edge[-latex] node[left, pos=0.5] {\(\test{b}\)} (1);
        \draw (0) edge[-latex] node[above] {\(\test{\bar b}\)} (2);
        \node (4) [right=2cm of 1] {$\expr{e}'{\color{Maroon}{}\seq\expr{e}^{(\mathsf{b})}}$};
        \node (5) [left=.8cm of 1] {{\color{gray} \(\bcancel{\checkmark}\)}};
        \draw (1) edge[-implies, double, double distance=0.5mm, gray] node[above, gray] {\(\prob{t}\)} (5);
        \draw (1) edge[-latex, dashed] node[above] {\footnotesize\( \action{p} \mid \prob{s}/{\color{Maroon} (1-t)}\)} (4);
    \end{tikzpicture}
\end{gather*}
The reweighing of probabilities used in the definition of the loops comes from defining loops as least fixpoints w.r.t.\ to an order on distributions, similarly to Stark and Smolka~\cite{Stark:2000:Complete}.

\smallskip
Finally, we specify the behaviour of the probabilistic loop.
In the special case where \(\partial(\expr{e})_\alpha(\accept) = \prob{1}\) and \(\prob{r} = \prob{1}\), the loop will not terminate; hence we set \(\partial\left(\expr{e}^{[\prob{r}]}\right)_\alpha=\delta_{\reject}\).
In all other cases, we look at $\partial(\expr{e})_\alpha$ to build $\partial(\expr{e}^{[\prob{r}]})_\alpha$ for each $\alpha \in \At$.
First, we make sure that the loop may be skipped with probability $1-r$.
Next, we account for the possibility that $\expr{e}$ may reject or return a value, and we modify the productive branches by adding $\expr{e}^{[\prob{r}]}$ to be executed next, as was done for the guarded loop.
The remaining mass is $r\partial(e)_\alpha(\accept)$, the probability that we will enter the loop with an atom that can skip over the loop body.
As was the case for the guarded loop, we discard this possibility and redistribute it among the remaining branches.
The resulting definition of $\partial(\expr{e}^{[\prob{r}]})_\alpha$ is given below on the left.
\begin{gather*}
    \partial\left(\expr{e}^{[\prob{r}]}\right)_\alpha(x) = \begin{cases}
        \frac{\prob{1-r}}{\prob{1}-\prob{r}\partial(\expr{e})_\alpha(\accept)} & x = \accept\\
        \frac{\prob{r}\partial(\expr{e})_\alpha(x)}{\prob{1}-\prob{r}\partial(\expr{e})_\alpha(\accept)} & x \in \{\reject\} \cup \V \\[.7ex]
        \frac{\prob{r}\partial(\expr{e})_\alpha(\action{p},\expr{e'})}{\prob{1}-\prob{r}\partial(\expr{e})_\alpha(\accept)} & x = \left(\action{p}, \left(\expr{e'}\seq\expr{e}^{[\prob{r}]}\right)\right)\\[.7ex]
        \prob{0} &\text{otherwise}
        \end{cases}
    \quad
    \begin{tikzpicture}[baseline=-7ex]
        \node (0) {$\expr{e}^{[\prob{r}]}$};
        \node (1) [below=.5cm of 0] {$\circ$};
        \draw (0) edge[-latex] node[left, pos=0.5] {\(\alpha\)} (1);
        \node (6) [left=1.2cm of 1, Maroon] {\checkmark};
        \node (4) [right=2.2cm of 1] {$\expr{e}'{\color{Maroon}{} \seq \expr{e}^{[\prob{r}]}}$};
        \node (5) [below=.5cm of 1] {{\color{gray} \(\bcancel{\checkmark}\)}};
        \draw (1) edge[-implies, double, double distance=0.5mm, gray, Maroon] node[above] {\footnotesize \(\frac{1-\prob{r}}{1-\prob{rt}}\)} (6);
        \draw (1) edge[-implies, double, double distance=0.5mm, gray, pos=0.3] node[left, gray] {\(\prob{rt}\)} (5);
        \draw (1) edge[-latex, dashed] node[above] {\footnotesize\( \action{p} \mid {\textcolor{Maroon}{\prob{r}}}\prob{s}/{\color{Maroon} (1-\prob{rt})}\)} (4);
    \end{tikzpicture}
\end{gather*}
As before, we provide an informal visual depiction of the probabilistic loop semantics above on the right, using the same conventions.

\subparagraph*{Reachable states. }
For any \acro{ProbGKAT} automaton \((X, \beta)\) and any \(x \in X\), we denote by \(\cogen{x}{\beta}\) the set of states reachable from \(x\) via \(\beta\). Clearly, \((\cogen{x}{\beta}, \beta)\) is a \acro{ProbGKAT} automaton and is the smallest subautomaton of \((X, \beta)\) containing \(x\). The canonical inclusion map \((\cogen{x}{\beta}, \beta)\to(X, \beta)\) is a \acro{ProbGKAT} automaton homomorphism. In particular, \((\cogen{\expr e}{\partial}, \partial)\) is the smallest subautomaton of \((\Exp,\partial)\) containing $\expr e$. We will refer to this subautomaton as the small-step semantics of \(\expr{e}\). We will often abuse notation and write \(\cogen{\expr e}{\partial}\) for \((\cogen{\expr e}{\partial}, \partial)\).

The following lemma says that every \acro{ProbGKAT} expression generates a finite automaton.
\begin{restatable}{lemma}{locallyfinite}\label{lem:locally_finite}
For all \(\expr{e} \in \Exp\), \(\cogen{\expr{e}}{\partial}\) is finite. In fact, the number of states is bounded above by \(\#(\expr{e}): \Exp \to \bN\), where \(\#(-)\) is defined recursively by
    {
        \small
    \[\#(\test{b})=1 \quad \#(\var{v})=1 \quad \#(\action{p})=2 \quad \#(\expr{e} +_\test{b} \expr{f}) = \#(\expr{e}) + \#(\expr{f}) \quad \#(\expr{e} \oplus_{\prob{r}} \expr{f}) = \#(\expr{e}) + \#(\expr{f}) \]
    \[\#(\expr{e}\seq\expr{f}) = \#(\expr{e}) + \#(\expr{f}) \quad \#\left(\expr{e}^{(\test{b})}\right)=\#(\expr{e}) \quad \#\left(\expr{e}^{[\prob{r}]}\right)=\#(\expr{e}) \]
    }
\end{restatable}

\section{Bisimulations and their properties }\label{sec:bisim}
Verifying that a given relation is a bisimulation (\cref{def:bisimulation_of_transition_systems}) requires that we construct a suitable transition structure on the relation.
In this section, we give necessary and sufficient conditions for the existence of such a transition structure.
We also study properties of the bisimilarity relation \(\bisim\), the largest bisimulation~\cite{Rutten:2000:Universal}.

\subparagraph*{Concrete characterisation of bisimulation equivalence. }
There is a beautiful characterisation of bisimulations between Markov chains in~\cite{Jones:1990:Probabilistic}, whose proof makes use of the max-flow min-cut theorem.
Adapting this work to \acro{ProbGKAT} automata produces a useful characterisation of \emph{bisimulation equivalences}, bisimulations that are also equivalence relations.

\begin{restatable}{lemma}{larsenskou}\label{lem:larsen_skou}
    Let \((X, \beta)\) be a \acro{ProbGKAT} automaton and let \(R \subseteq X \times X\) be an equivalence relation. \(R\) is a bisimulation if and only if and only if for all \((x,y) \in R\) and \(\alpha \in \At\),
    \begin{enumerate}
        \item for all \(o \in \2 + \V\), \(\beta(x)_\alpha(o)=\beta(y)_\alpha(o)\), and
        \item for all equivalence classes \(Q \in X / {R}\) and all \(\action{p} \in \Act\), \(\beta(x)_\alpha[\{p\}\times Q] = \beta(y)_\alpha[\{p\}\times Q]\)
    \end{enumerate}
\end{restatable}

This lemma can be seen as an extension of Larsen-Skou bisimilarity~\cite{Larsen:1991:Bisimulation} to systems with outputs.
Intuitively, \(R\) is a bisimulation equivalence if for any atom $\alpha \in \At$ and \((x,y)\in R\), the transitions assign the same probabilities to any output, and the probability of transitioning into any given equivalence class after emitting \(\action{p}\) is the same for both \(x\) and \(y\).

\subparagraph*{Bisimilarity and its properties. }
Given a relation \(R \subseteq X \times Y\), define $R^{-1} = \{ (y, x) \mid x\mathrel{R}y \}$, and given \(A \subseteq X\), write \(R(A) = \{y \in Y \mid x\mathrel{R}y, x \in A\}\).
The bisimilarity relation  \({\bisim_{\beta, \gamma}} \subseteq X \times Y\) between \((X, \beta)\) and \((Y, \gamma)\) is the \emph{greatest fixpoint} of the following operator.

\begin{definition}\label{def:bisimulation_functional}
    Let \((X, \beta)\) and \((Y, \gamma)\) be \acro{ProbGKAT} automata and let \(R \subseteq X \times Y\).
    We define the operator \(\Phi_{\beta, \gamma}: 2^{X \times Y} \to 2^{X \times Y}\) so that \((x,y)\in \Phi_{\beta, \gamma}(R)\) if for any given \(\alpha \in \At\),
    \begin{itemize}
        \item for all \(o \in \2 + \V\), \(\beta(x)_\alpha(o)=\gamma(y)_\alpha(o)\),
        \item for all \(A \subseteq X\) and all \(\action{p} \in \Act\), \(\beta(x)_\alpha[\{\action{p}\}\times A]\leq \gamma(y)_\alpha[\{\action{p}\} \times R(A)]\), and
        \item for all \(B \subseteq Y\) and \(\action{p} \in \Act\), \(\gamma(y)_\alpha[\{\action{p}\}\times B]\leq \beta(x)_\alpha[\{\action{p}\} \times R^{-1}(B)]\).
    \end{itemize}
    From now on, we will omit the subscripts from $\Phi$ when the automata are clear from context.%
\end{definition}

The operator \(\Phi_{\beta,\gamma}\) can also be used to define a behavioural pseudometric.
Let \((X, \beta)\) be a \acro{ProbGKAT} automaton.
A \emph{relation refinement chain} is an indexed family \(\{\bisim^{(i)}\}_{i \in \mathbb N}\) of relations on \(X\) defined as:
$
    \bisim^{(0)}= X \times X$, $\bisim^{(i+1)} = \Phi(\bisim^{(i)})
$.
We can intuitively think of successive elements of this chain as closer approximations of bisimilarity (see also~\cite{Hennessy:80:Observing}).

\begin{restatable}{theorem}{refinementchain}\label{thm:refinement_chain}
    Let \((X, \beta)\) be a \acro{ProbGKAT} automaton.
    For any $x, y \in X$, \(x \bisim y\) if and only if for all \(i \in \mathbb N\), we have \(x \bisim^{(i)} y\).
\end{restatable}

Thus, if \(x,y \in X\) are not bisimilar, then there exists a maximal \(i \in \mathbb N\) such that \(x\bisim^{(i)} y\).
In \cref{sec: completeness}, we use this to define a pseudometric on the states of any \acro{ProbGKAT} automaton.
Informally speaking, this allows us to quantify \emph{how close} to being bisimilar two states are.

Our main goal is to axiomatise bisimilarity of \acro{ProbGKAT} expressions with a set of equational laws and reason about equivalence using \emph{equational logic}.
For such an axiomatisation to exist, bisimilarity needs to be both an equivalence relation and a congruence with respect to the \acro{ProbGKAT} operations.
The greatest bisimulation on any \acro{ProbGKAT} automaton is an equivalence~\cite{Rutten:2000:Universal}, but being congruence requires an inductive argument.

\begin{restatable}{theorem}{greatestcong}%
    \label{thm: greatest_bisim_is_congruence}
    The greatest bisimulation on \((\Exp, \partial)\) is a congruence with respect to \acro{ProbGKAT} operations.
\end{restatable}


\begin{table*}[t]\small\centering\noindent%
\caption{Axiomatisation of \acro{ProbGKAT}. In the figure \(\expr{e}, \expr{f}, \expr{g} \in \Exp\), \( \test{b}, \test{c} \in \Bexp\), \(\var{v} \in \V\), \(\action{p} \in \Act\) and \(\prob{r}, \prob{s} \in [0,1]\). Laws involving division of probabilities apply when the denominator is not zero. To simplify the notation, we write \(\trmt{\expr{e}} = 0\) to denote that for all \(\alpha \in \At\) it holds that \(\trmt{\expr{e}}_\alpha=0\).}
\begin{center}
    \fbox{\noindent \begin{minipage}[t]{.43\textwidth}
        \noindent%
        \textbf{Guarded Choice Axioms}

        \vspace{.3em}
        \begin{tabular}{l@{~} >{$}r<{$}@{~} >{$}c<{$}@{~} >{$}l<{$}}
        \acro{(G1)} &\expr{e} +_\test{b} \expr{e} & \equiv & \expr{e}\\
        \acro{(G2)} &\expr{e} +_\test{b} \expr{f} &\equiv& \test{b}\seq\expr{e} +_\test{b} \expr{f}\\
        \acro{(G3)} &\expr{e} +_\test{{b}} \expr{f} &\equiv& \expr{f} +_\test{\bar{b}} \expr{e}\\
        \acro{(G4)} &(\expr{e} +_\test{b} \expr{f}) +_\test{c} \expr{g} &\equiv& \expr{e} +_\test{bc} (\expr{f} +_\test{c} \expr{g})
        \end{tabular}
        \vspace{.3em}
        \hrule
        \vspace{.3em}
        \noindent%
        \textbf{Distributivity Axiom}

        \vspace{.3em}
        \begin{tabular}{l@{~} >{$}r<{$}@{~} >{$}c<{$}@{~} >{$}l<{$}}
        \acro{(D)} &\expr{e} \oplus_\prob{r} (\expr{f} +_\test{b} \expr{g}) &\equiv& (\expr{e} \oplus_\prob{r} \expr{f}) +_\test{b} (\expr{e} \oplus_\prob{r} \expr{g})
        \end{tabular}
        \vspace{.3em}
        \hrule
        \vspace{.3em}
        \noindent%
        \textbf{Sequencing Axioms}

        \vspace{.3em}
        \begin{tabular}{l@{~} >{$}r<{$}@{~} >{$}c<{$}@{~} >{$}l<{$}}
        \acro{(S1)} &\test{\one} \seq \expr{e} &\equiv& \expr{e}\\
        \acro{(S2)} &\expr{e} \seq \test{\one} &\equiv& \expr{e}\\
        \acro{(S3)} &(\expr{e} \seq \expr{f})\seq\expr{g}&\equiv& \expr{e} \seq (\expr{f}\seq\expr{g})\\
        \acro{(S4)} & \test{\zero}\seq\expr{e} &\equiv& \test{\zero} \\
        \acro{(S5)} & (\expr{e} +_\test{b} \expr{f})\seq \expr{g} &\equiv& \expr{e} \seq \expr{g} +_\test{b} \expr{f}\seq\expr{g}\\
        \acro{(S6)} & (\expr{e} \oplus_\prob{r} \expr{f})\seq \expr{g} &\equiv& \expr{e} \seq \expr{g} \oplus_\prob{r} \expr{f}\seq\expr{g}\\
        \acro{(S7)} & \var{v} \seq \expr{e} &\equiv& \var{v}\\
        \acro{(S8)} & \test{b} \seq \test{c} &\equiv& \test{b}\test{c}
        \end{tabular}
        \vspace{.3em}
        \hrule
        \vspace{.3em}
        \noindent%
        \textbf{Loop Axioms}
        \vspace{.3em}

        \begin{tabular}{l@{~} >{$}r<{$}@{~} >{$}c<{$}@{~} >{$}l<{$}}
            \acro{(L1)} &\expr{e}^{(\test{b})} &\equiv& \expr{e} \seq \expr{e}^{(\test{b})} +_\test{b} \test{\one}\\
            \acro{(L2)} &\expr{e}^{[\prob{r}]} &\equiv& \expr{e} \seq \expr{e}^{[\prob{r}]} \oplus_{\prob{r}} \test{\one}\\
            \acro{(L3)} &(\expr{e} +_\test{c} \test{\one})^{(\test{b})}&\equiv& (\test{c}\seq\expr{e})^{(\test{b})}\\
            \acro{(L4)} & \expr{e}^{(\test{\one})} &\equiv& \expr{e}^{[\prob{1}]} \\
        \end{tabular}
        \begin{tabular}{l@{~} >{$}r<{$}@{~} >{$}c<{$}@{~} >{$}l<{$}}
            \acro{(L5)} & \multicolumn{3}{c}{\multirow{3}{*}{\inferrule{\expr{e} \equiv (\expr{f} \oplus_{\prob{s}} \test{\one}) +_\test{c} \expr{g} \\ \prob{s} > \prob{0}}{\test{c}\seq\expr{e}^{(\test{b})}\equiv \test{c}\seq\left(\expr{f}\seq\expr{e}^{(\test{b})} +_\test{b} \test{\one}\right)}}} \\
            \multicolumn{4}{c}{} \\
            \multicolumn{4}{c}{} \\
            \acro{(L6)} & \multicolumn{3}{c}{\multirow{3}{*}{\inferrule{\expr{e} \equiv (\expr{f} \oplus_{\prob{s}} \test{\one}) +_\test{c} \expr{g}}{\test{c}\seq\expr{e}^{[\prob{r}]}\equiv \test{c}\seq\left(\expr{f}\seq\expr{e}^{[\prob{r}]} \oplus_{\prob{\frac{rs}{1-r(1-s)}}} \test{\one}\right)}}} \\
            \multicolumn{4}{c}{} \\
            \multicolumn{4}{c}{} \\
            \end{tabular}
    \end{minipage}\hspace{.01\linewidth}\vrule\hspace{.01\linewidth}
    \begin{minipage}[t]{.515\textwidth}
        \textbf{Probabilistic Choice Axioms}

        \vspace{.3em}
        \begin{tabular}{l@{~} >{$}r<{$}@{~} >{$}c<{$}@{~} >{$}l<{$}}
        \acro{(P1)} &\expr{e} \oplus_\prob{r} \expr{e} &\equiv& \expr{e}\\
        \acro{(P2)} &\expr{e} \oplus_{\prob{1}} \expr{f} &\equiv& \expr{e}\\
        \acro{(P3)} &\expr{e} \oplus_{\prob{r}} \expr{f} &\equiv& \expr{f} \oplus_{\prob{1-r}}\expr{e} \\
        \acro{(P4)} &(\expr{e} \oplus_{\prob{r}} \expr{f}) \oplus_{\prob{s}} \expr{g} &\equiv& \expr{e} \oplus_\prob{rs} (\expr{f} \oplus_\prob{\frac{(1-r)s}{1-rs}} \expr{g})
        \end{tabular}
        \vspace{.3em}
        \hrule
        \vspace{.3em}
        \noindent%
        \textbf{Fixpoint Rules}

        \vspace{.3em}
        \begin{tabular}{l@{~} >{$}r<{$}@{~} >{$}c<{$}@{~} >{$}l<{$}}
            \acro{(F1)} & \multicolumn{3}{c}{\multirow{3}{*}{\inferrule{\expr{g} \equiv \expr{e} \seq \expr{g} +_\test{b} \expr{f} \\
            \trmt{\expr{e}}=0}{\expr{g}\equiv \expr{e}^{(\test{b})}\seq\expr{f}}}} \\
            \multicolumn{4}{c}{} \\
            \multicolumn{4}{c}{} \\
            \acro{(F2)} & \multicolumn{3}{c}{\multirow{3}{*}{\inferrule{\expr{g} \equiv \expr{e} \seq \expr{g} \oplus_\prob{r} \expr{f} \\ \trmt{\expr{e}}=0}{\expr{g}\equiv \expr{e}^{[\prob{r}]}\seq\expr{f}}}} \\
            \multicolumn{4}{c}{} \\
            \multicolumn{4}{c}{}
        \end{tabular}
        \vspace{.3em}
        \hrule
        \vspace{.3em}
        \noindent%
        \vspace{.3em}
        Define \(\mathsf{E} : \Exp \to \At \to [0,1]\) inductively by
        \begin{align*}
             \trmt{\action{p}}_\alpha &= \trmt{\var{v}}_\alpha = \prob{0} \\
             \trmt{\test{b}}_\alpha &= \begin{cases}
                \prob{1} & \alpha \bleq \test{b} \\
                \prob{0} & \alpha \bleq \bar{\test{b}}
             \end{cases}\\
             \trmt{\expr{e} +_\test{b} \expr{f}}_\alpha &=
             \begin{cases}
                \trmt{\expr{e}}_\alpha & \alpha \bleq \test{b} \\
                \trmt{\expr{f}}_\alpha & \alpha \bleq \bar{\test{b}}
             \end{cases}\\
             \trmt{\expr{e} \oplus_{\prob{r}} \expr{f}}_\alpha &=
                \prob{r} \trmt{\expr{e}}_\alpha + \prob{(1-r)} \trmt{\expr{f}}_\alpha\\
             \trmt{\expr{e} \seq \expr{f}}_\alpha &= \trmt{\expr{e}}_\alpha \trmt{\expr{f}}_\alpha\\
             \trmt{\expr{e}^{(\test{b})}}_\alpha &= \trmt{\bar{\test{b}}}_\alpha\\
             \trmt{\expr{e}^{[\prob{r}]}}_\alpha &= \begin{cases}
                \prob{0} & \prob{r}=\prob{1}\text{ and } \trmt{\expr{e}}_\alpha=\prob{1} \\
                \frac{\prob{1}-\prob{r}}{\prob{1}-\prob{r}\trmt{\expr{e}}_\alpha} & \text{otherwise}
            \end{cases}
        \end{align*}
    \end{minipage}
    }
    \end{center}
    \vspace{1.5pt}
\label{fig:axiomatisation}
\end{table*}

\section{Axiomatisation}\label{sec: axiomatisation}
We turn our attention to \emph{axiomatisation} of bisimilarity of \acro{ProbGKAT} expressions, using an axiom system based on \acro{GKAT} modulo bisimilarity~\cite{Schmid:2021:Guarded}.
First, we give an overview of the axioms, and establish their soundness.
Finally, we show that our axioms are strong enough to decompose every expression into a certain syntactic normal form relating the expressions to their small-step semantics.
Completeness is tackled in the next section.

\subparagraph*{Overview of the axioms. }
\cref{fig:axiomatisation} contains the axioms, which are either \emph{equational} (of the form $\expr{e} \equiv \expr{f}$), or \emph{quasi-equational} (of the form $\expr{e}_1 \equiv \expr{f}_1, \dots \expr{e}_n \equiv \expr{f}_n \implies \expr{e} \equiv \expr{f}$).
It also holds the definition of the function $\trmt{-}$, which is necessary to give a side condition to the fixpoint rules.
We define \({\equiv} \subseteq \Exp \times \Exp\) as the smallest congruence relation satisfying the axioms.

Axioms \acro{G1}--\acro{G4} are inherited from \acro{GKAT} and govern the behaviour of Boolean guarded choice.
\acro{P1}--\acro{P4} can be thought of as their analogues, but for the probabilistic choice. The distributivity axiom \acro{D} states that guarded choice distributes over a probabilistic choice, which reflects the way our operational model resolves both types of branching.

The sequencing axioms \acro{S1}--\acro{S8} are mostly inherited from \acro{GKAT}.
The new axioms include \acro{S6} which talks about right distributivity of sequencing over probabilistic choice and \acro{S7} which captures the intuitive property that any code executed after a \texttt{return} statement is not executed.
\acro{L1} and \acro{L3} come from \acro{GKAT}, while \acro{L2} is a probabilistic loop analogue of \acro{L1}, which captures the semantics of the probabilistic loop in terms of recursive unrolling.
\acro{L4} equates the \texttt{while(true)} and \texttt{while(flip(1))} loops.
\acro{F1} and \acro{F2} are inspired by Salomaa's axioms~\cite{Salomaa:1966:Two} and provide a partial converse to \acro{L1} and \acro{L2} respectively, given the loop body cannot immediately terminate.
The property that a loop body has a zero probability of outputting \(\accept\) is formally writen using the side condition \(\trmt{\expr{e}}=0\), which can be thought of as \emph{empty word property} from Salomaa's axiomatisation~\cite{Salomaa:1966:Two}.

This leaves us with \acro{L5} and \acro{L6}, which describe the behaviour of guarded and probabilistic loops where parts of the loop body may be skipped. These  are quasi-equational, but can be replaced by equivalent equations --- see \cref{rem:axioms}.
\acro{L5} concerns a loop on an expression $\expr{e}$ that has probability $1 - \prob{s}$ of not performing any action, given that $\test{c}$ holds.
The rule says that, if we start the loop on $\expr{e}$ given that $\expr{c}$ holds, then either $\test{b}$ holds and we execute $\expr{f}$, or it does not, and the loop is skipped.
The reason that we can disregard the $\test{\one}$ part of $\expr{e}$ is that if this branch is taken, then $\expr{c}$ still holds on the next iteration of the loop, and so the program will have to choose probabilistically between $\expr{f}$ and $\test{\one}$ once more.
Since $\prob{s} > 0$, it will eventually choose the probabilistic branch $\expr{f}$ with almost sure probability.

The second rule, \acro{L6}, is the analogue of \acro{L5} for probabilistic choice.
In this case, however, a choice for $\test{\one}$ also means another probabilistic experiment to determine whether the loop needs to be executed once more, with probability $\prob{r}$.
The consequence is that if the loop on $\expr{e}$ is started given that $\test{c}$ holds, some more probability mass will shift towards skipping, as a result executing $\test{\one}$ some number of times before halting the loop.

\subparagraph*{Soundness with respect to bisimilarity. }
Using the characterisation from \cref{sec:bisim}, we can show that \(\equiv\) is a bisimulation equivalence on \((\Exp, \partial)\).
The proof is available in \cref{apx:soundness}.
\begin{restatable}{lemma}{soundness}\label{lem: soundness}
\(\equiv\) is a bisimulation equivalence on \((\Exp, \partial)\)
\end{restatable}
We immediately obtain that provable equivalence is contained in bisimilarity.
\begin{theorem}[Soundness]\label{thm: soundness}
For all \(\expr{e}, \expr{f} \in \Exp\), if \(\expr{e} \equiv \expr{f}\) then \(\expr{e} \bisim \expr{f}\)
\end{theorem}

\subparagraph*{Example of equational reasoning. }
Since our axioms are sound, we can reason about \acro{ProbGKAT} expressions equationally, without constructing bisimulations by hand.
Once again, we revisit the algorithm from \cref{fig:3die}. To show correctness, we need to prove the equivalence of expressions \(\expr{e}\) and \(\expr{g}^{(\test{\one})}\) from \cref{ex:3die}, as follows:
\begin{align*}
    \expr{g}^{(\test{\one})} &\equiv \left( \left(\epsdice{1} \oplus_{\prob{\frac{1}{2}}} \epsdice{2}\right) \oplus_{\prob{\frac{1}{2}}} \left(\epsdice{3} \oplus_{\prob{\frac{1}{2}}} \test{\one}\right) \right)^{(\test{1})} \tag{Def. of $\expr{g}$} \\
    &\equiv \left(\left( \left(\epsdice{1} \oplus_{\prob{\frac{1}{2}}} \epsdice{2}\right) \oplus_{\prob{\frac{2}{3}}} \epsdice{3}\right) \oplus_{\prob{\frac{1}{2}}} \test{\one} \right)^{(\test{1})} \tag{See below}\\
    &\equiv \left(\left(\left( \left(\epsdice{1} \oplus_{\prob{\frac{1}{2}}} \epsdice{2}\right) \oplus_{\prob{\frac{2}{3}}} \epsdice{3}\right) \oplus_{\prob{\frac{1}{2}}} \test{\one}\right) +_\test{\one} \test{\zero} \right)^{(\test{1})} \tag{See below}\\
    &\equiv \left( \left(\epsdice{1} \oplus_{\prob{\frac{1}{2}}} \epsdice{2}\right) \oplus_{\prob{\frac{2}{3}}} \epsdice{3}\right)^{(\test{\one})} \tag{Axioms \acro{L5} and \acro{S1}}\\
    &\equiv \left( \epsdice{1} \oplus_{\prob{\frac{1}{3}}} \left( \epsdice{2} \oplus_{\prob{\frac{1}{2}}} \epsdice{3}\right)\right)^{(\test{1})} \tag{Axiom \acro{P4}}\\
    &\equiv \left( \epsdice{1} \oplus_{\prob{\frac{1}{3}}} \left( \epsdice{2} \oplus_{\prob{\frac{1}{2}}} \epsdice{3}\right)\right)\seq\left( \epsdice{1} \oplus_{\prob{\frac{1}{3}}} \left( \epsdice{2} \oplus_{\prob{\frac{1}{2}}} \epsdice{3}\right)\right)^{(\test{\one})} +_\test{\one} \test{\one} \tag{Axiom \acro{L1}}\\
    &\equiv \left( \epsdice{1} \oplus_{\prob{\frac{1}{3}}} \left( \epsdice{2} \oplus_{\prob{\frac{1}{2}}} \epsdice{3}\right)\right)\seq\expr{e}^{(\test{\one})} +_\test{\one} \test{\one} \tag{Def. \(\expr{e}\)}\\
        &\equiv \left( \epsdice{1} \oplus_{\prob{\frac{1}{3}}} \left( \epsdice{2} \oplus_{\prob{\frac{1}{2}}} \epsdice{3}\right)\right)\seq \expr{e}^{(\test{\one})} \tag{See below}
\\
    &\equiv \left( \epsdice{1}\seq e^{(\test{\one})} \oplus_{\prob{\frac{1}{3}}} \left( \epsdice{2}\seq e^{(\test{\one})} \oplus_{\prob{\frac{1}{2}}} \epsdice{3}\seq e^{(\test{\one})}\right)\right)\tag{Axiom \acro{S6}}\\
    &\equiv\left( \epsdice{1} \oplus_{\prob{\frac{1}{3}}} \left( \epsdice{2} \oplus_{\prob{\frac{1}{2}}} \epsdice{3}\right)\right) \tag{Axiom \acro{S7}}\\
    &\equiv \expr{e} \tag{Def. \(\expr{e}\)}
\end{align*}
In the second step, we used that for all \(\expr{e}_1, \expr{e}_2, \expr{e}_3 \in \Exp\) and $\prob{r}, \prob{s} \in [0,1]$ with $(1-\prob{r})(1 - \prob{s}) > 0$, we have \(\expr{e}_1 \oplus_\prob{r} (\expr{e}_2 \oplus_{\prob{s}} \expr{e}_3) \equiv (\expr{e}_1 \oplus_\prob{k} \expr{e}_2) \oplus_\prob{l} \expr{e}_3\), where \(\prob{k}=\prob{\frac{r}{1-(1-r)(1-s)}}\) and \(\prob{l}=\prob{1-(1-r)(1-s)}\). In the third and eighth steps, we used that for all \(\expr{e}, \expr{f} \in \Exp\), we have that \(\expr{e} +_\test{\one} \expr{f} \equiv \expr{e}\). Both those equivalences follow from the other axioms; see \Cref{lem:derivable_facts} in the appendix for details.

\subparagraph*{Fundamental theorem. }
Every expression in the language of \acro{KA} (resp.\ \acro{KAT}, \acro{GKAT}) can be reconstructed from its small-step semantics, up to \(\equiv\).
This property, often referred to as the \emph{fundamental theorem} of (in analogy with the fundamental theorem of calculus and following the terminology of Rutten~\cite{Rutten:2000:Universal}) is useful in many contexts, and we will need it later on.
\begin{restatable}[Fundamental Theorem]{theorem}{fundamentaltheorem}\label{thm: fundamental_theorem}
    For every \(\expr{e} \in \Exp\) it holds that
	\[
		\expr{e} \equiv \bigsum{\alpha \in \At} \left(\bigplus{d \in \supp\left(\partial(\expr{e})_\alpha\right)} {\partial(\expr{e})_\alpha(d)} \cdot {\ex(d)}\right)
	\]
	where \(\ex\) defines a function \(\2 + \V + \Act \times \Exp \to \Exp\) given by
	\begin{equation*}
		\ex(\reject) = \test{\zero} \quad \ex(\accept) = \test{\one} \quad \ex(\var{v}) = \var{v} \quad \ex(\action{a},\expr{f}) = \action{a}\seq \expr{f}
		\quad \quad \text{ \((\var{v} \in \V, \action{a} \in \Act \) and \(\expr{f} \in \Exp)\)}
	\end{equation*}
\end{restatable}
The proof is given in the appendix. We use a generalised type of guarded and probabilistic choice ranging over indexed collections of expressions, which is defined in \cref{apx:generalised}.
\section{Completeness}\label{sec: completeness}
Given the axioms presented in the previous section, a natural question is to ask whether they are \emph{complete} w.r.t.\ bisimilarity --- i.e., whether any bisimilar pair can be proved equivalent using the axioms that make up $\equiv$.
The traditional strategy is to develop the idea of \emph{systems of equations} within the calculus, and show that these systems have unique (least) solutions up to provable equivalence.
If we can characterise the expressions of interest as solutions to a common system of equations (typically derived from the bisimulation that relates them), then uniqueness of solutions guarantees their equivalence.
Unfortunately, the first step of this process, where systems of equations are shown to have unique solutions, does not transfer to \acro{ProbGKAT} (nor \acro{GKAT}).
Indeed, some systems of equations do not have \emph{any} solution~\cite{Kozen:2008:Bohm,Schmid:2021:Guarded}; the lack of a procedure to construct solutions also encumbers a proof of uniqueness.

Instead, we follow the approach from~\cite{Smolka:2020:Guarded} pioneered by Bergstra and Klop~\cite{Bergstra:1985:Verification}, and incorporate uniqueness of solutions into the axiomatisation.
The \emph{uniqueness axiom} (\acro{UA}) that accomplishes this is an \emph{axiom scheme}, which is to say it is a template for infinitely many axioms, one for each number of unknowns.
In the case of a single unknown, one can show that \acro{F1} and \acro{F2} are special cases of \acro{UA}, which moreover give a candidate solution.

With \acro{UA} in hand, the traditional roadmap towards completeness works out.
Before we get there, however, we must expend some energy to properly state this axiom scheme.
Moreover, showing soundness of \acro{UA} requires effort.
Both of these take up the bulk of the development in this section; we derive the desired completeness property at the end.

\subparagraph*{(Salomaa) systems of equations. } First, we define formally the idea of \emph{systems of equations} for \acro{ProbGKAT} automata.
The constraints on each variable will be built using the following two-sorted grammar, where \(X\) is a finite set of indeterminates.
\begin{align*}
    e_1, e_2 \in \Exp(X) &::= p \mid e_1 +_\test{b} e_2 \tag{\(p \in \Pexp(X), \test{b} \in \Bexp\)}\\
    p_1, p_2 \in \Pexp(X) &::= \expr{f} \mid \expr{g}x \mid p_1 \oplus_{\prob{r}} p_2 \tag{\(\expr{f},\expr{g} \in \Exp, x \in X, \prob{r} \in [0,1]\)}
\end{align*}

\begin{definition}
A \emph{system of equations} is a pair \((X, \tau : X \to \Exp(X))\) consisting of a finite set \(X\) of indeterminates and a function \(\tau : X \to \Exp(X)\).
If for all \(x \in X\), in each of \(\tau(x)\) all subterms of the form \(\expr{g}x\) satisfy \(\trmt{\expr{g}}=0\), then such system is called \emph{Salomaa}.\footnote{%
    In process algebra~\cite{Milner:1984:Complete}, Salomaa systems are usually called \emph{guarded}.
    We avoid the latter name to prevent confusion with Boolean guarded choice present in \acro{(Prob)GKAT}.
}
\end{definition}

Every finite state \acro{ProbGKAT} automaton yields a Salomaa system of equations.
\begin{definition}
    Let \((X, \beta)\) be a finite state \acro{ProbGKAT} automaton. A system of equations associated with \((X, \beta)\) is a Salomaa system \((X, \tau)\), with \(\tau : X \to \Exp(X)\) defined by
    \[
    \tau(x) = \bigsum{\alpha \in \At} \left(\bigplus{d \in \supp(\beta(x))_\alpha} \beta(x)_\alpha(d) \cdot \sys(d)\right)
    \]
    where \(\sys : \2 + \V + \Act \times \Exp \to \Pexp(X)\) is given by
    \[
    \sys(\reject) = \test{\zero} \quad \sys(\accept) = \test{\one} \quad \sys(\var{v}) = \var{v} \quad \sys(\action{a},x) = \action{a}x
    \]
\end{definition}
\begin{example}\label{ex:system}
    In the system associated with the automaton from \cref{ex:automaton}, $\tau$ is given by
    \begin{mathpar}
        x_1  \mapsto (\action{q}x_2 \oplus_\prob{\frac{1}{2}} \var{v}) +_\test{\alpha} \left((\action{p} x_1 \oplus_\prob{\frac{1}{2}} \action{q} x_2)+_{{\test{\beta}}}\test{\zero}\right) \and
        x_2  \mapsto \test{\one} +_\test{\alpha} \left( \test{\one} +_{{\test{\beta}}} \test{\zero} \right)
    \end{mathpar}
\end{example}
Given a function \(h : X \to \Exp\) that assigns a value to each indeterminate in $X$, we can derive a \acro{ProbGKAT} expression $h^{\#}(e)$ for each $e \in \Exp(X)$ inductively, as follows: $h^{\#}(\expr{f}) = \expr{f}$, $h^{\#}(p_1 \oplus_\prob{r} p_2) = h^{\#}(p_1) \oplus_\prob{r} h^{\#}(p_2)$, $h^{\#}(\expr{g}x) = \expr{g}\seq h(x)$, $h^{\#}(e_1 +_\test{b} e_2) = h^{\#}(e_1) +_\test{b} h^{\#}(e_2)$. We can now state the notion of a \emph{solution} to the Salomaa system.
Rather than expecting both sides of equations to be strictly equal, we require them to be related by a relation, which we leave as a parameter to instantiate later.
\begin{definition}
    Let $R \subseteq \Exp \times \Exp$.
    A \emph{solution} up to \(R\) to a system \((X, \tau)\) is a map \(h : X \to \Exp\) satisfying for all \(x \in X\) that \(\left( h(x), h^{\#}(\tau(x)) \right) \in R\).
\end{definition}

\begin{example}
A solution up to $\equiv$ to the system from \Cref{ex:system} would satisfy
\begin{mathpar}
    h(x_1) \equiv (\action{q} \seq h(x_2) \oplus_\prob{\frac{1}{2}} \var{v}) +_\test{\alpha} \left((\action{p} \seq h(x_1) \oplus_\prob{\frac{1}{2}} \action{q} \seq h(x_2))+_{{\test{\beta}}}\test{\zero}\right) \and
    h(x_2) \equiv \test{\one} +_\test{\alpha} \left( \test{\one} +_{{\test{\beta}}} \test{\zero} \right)
\end{mathpar}
In this case, choosing \(h(x_1)=(\action{p} +_\test{\beta} \var{v})^{[\prob{\frac{1}{2}}]}\seq\action{q}\) and \(h(x_2)=\test{\one}\) fits these constraints.
\end{example}
\begin{example}
    Let \(\mathsf{r} \in [0,1]\). The recursive specification on the left below describes a program \(\mathsf{randAdd(m,r)}\) which takes an integer \(\mathsf{m}\) and bias \(\mathsf{r}\).
    As long as \(\mathsf{m}\) is strictly below \(10\), this program flips an \(\mathsf{r}\)-biased coin to decide between incrementing \(\mathsf{m}\) followed by a recursive call or termination. That recursive specification can be thought of as a Salomaa system with one unknown; the program on the right is a solution up to \(\equiv\).
    \begin{center}
        \begin{tikzpicture}
            \node(p1){
                \begin{minipage}{60mm}
                    \begin{algorithm}[H]
                        \SetAlgoNoLine
                        \Def{$\mathsf{randAdd(m,r)}$}
                        {\eIf{$\mathsf{m<10}$}
                        {\eIf{$\mathsf{flip(r)}$}
                        {$\mathsf{m}\texttt{++}$\;$\mathsf{randAdd(m,r)}$}
                        {\kwSkip}}
                        {\kwSkip}}
                    \end{algorithm}
                    
                \end{minipage}
            };
            \node[right=6cm of p1.center,anchor=center] (p2) {
                \begin{minipage}{50mm}
                    \begin{algorithm}[H]
                        \SetAlgoNoLine
                        \While{$\mathsf{flip(r)}$}
                        {\eIf{$\mathsf{m<10}$} {$\mathsf{m}\texttt{++}$\;$\mathsf{randAdd(m,r)}$}{\kwSkip}}
                        
                    \end{algorithm}
                \end{minipage}
            };
            \draw[dotted] ($(p1.north west) + (0.2, 0)$) rectangle ($(p1.south east) + (-1.6, -0.05)$);
            \draw[dotted] ($(p2.north west) + (0.2, 0)$) rectangle ($(p2.south east) + (-1.1, -0.05)$);
        \end{tikzpicture}
    \end{center}
\end{example}
Solutions up to \(\equiv\) can be characterised concretely, using \Cref{thm: fundamental_theorem}.
\begin{restatable}{theorem}{solutionshomomorphisms}\label{lem: solutions_are_homomorphisms}
    Let \((X, \beta)\) be a finite state \acro{ProbGKAT} automaton. The map \(h : X \to \Exp\) is a solution up to \(\equiv\) to the system associated with \((X, \beta)\) if and only if \([-]_\equiv \circ h\) is a \acro{ProbGKAT} automata homomorphism from \((X, \beta)\) to \(({\Exp}/{\equiv}, \bar{\partial})\). We write \(\bar{\partial}\) to denote the unique transition function on \({\Exp}/{\equiv}\) which makes the quotient map \([-]_\equiv : \Exp \to {\Exp}/{\equiv}\) a \acro{ProbGKAT} automaton homomorphism from \((\Exp, \partial)\)~\cite[Proposition~5.8]{Rutten:2000:Universal}.
\end{restatable}

\subparagraph*{Uniqueness of Solutions axiom. }
Informally, \acro{UA} extends $\equiv$ by stating that solutions to Salomaa systems, if they exist, are unique.
Formally, we define \({\uaequiv} \subseteq \Exp \times \Exp\) to be the least congruence that contains \(\equiv\), and satisfies the following (quasi-equational) axiom:
\begin{equation}
    \inferrule{(X, \tau) \text{ is a Salomaa system} \quad f,g : X \to \Exp \text{ are solutions of } (X, \tau)\text{ up to } \uaequiv}{f(x)\uaequiv g(x) \text{ for all }x \in X }
    \tag{\acro{UA}}
\end{equation}
\acro{F1} and \acro{F2} are instantiations of \acro{UA} for Salomaa systems with one variable.

\subparagraph*{Behavioural pseudometric. }
We now develop the theory necessary to verify soundness of \acro{UA}.
First, note that for every \acro{ProbGKAT}, we can define a function $d_X: X \times X \to \bR^{+}$:
\[
    d_{X}(x,y) = \begin{cases}
        2^{-n} &n \text{ is maximal such that } x \bisim^{(n)} y \\
        0 & \text{if } x \bisim y
    \end{cases}
\]
The above is well-defined by \cref{thm:refinement_chain}, and is a \emph{pseudometric}, in the following sense.
\begin{definition}
    A \emph{pseudometric space} is a pair \((X, d_X)\), where \(d_X : X \times X \to \mathbb{R}^+\) is a \emph{pseudometric}, which means that for all \(x,y,z \in X\) we have
    \begin{mathpar}
        d_X(x,x)=0
        \and
        d_X(x,y)=d_X(y,x)
        \and
        d_X(x,z)\leq d_X(x,y) + d_X(y,z)
    \end{mathpar}
    Let \(k \in \mathbb{R}^{+}\). A mapping \(f : X \to Y\) between pseudometric spaces \((X, d_X)\) and \((Y, d_Y)\) is called \emph{\(k\)-Lipschitz} if for all \(x,y \in X\) \(d_Y(f(x), f(y))\leq kd_X(x,y)\).
\end{definition}

The behavioural pseudometric satisfies the definition above, in a strong sense.
\begin{restatable}{lemma}{metricwelldefined}
    For every \acro{ProbGKAT} automaton, \((X,\beta)\), \((X, d_X)\) a pseudometric space that is \emph{ultra}, in the sense that for all $x, y, z \in X$ we have \(d_X(x,z) = \max\{d_X(x,y), d_X(y,z)\}\).
\end{restatable}

Let \((X,d_X)\) and \((Y,d_Y)\) be pseudometric spaces.
Their \emph{product} is a pseudometric space \((X \times Y, d_{X \times Y})\) where \(d_{X\times Y}\) is defined by \(d_{X\times Y}((x,y),(x',y'))=\max\{d_X(x,x'), d_Y(y,y')\}\).
It is easy to show that if both pseudometric spaces are ultra, then so is their product.
Going forward, we will omit subscripts when they are clear from context.


\smallskip\noindent\textsf\textbf{\textsf{ Soundness of the Uniqueness Axiom.}}
Every Salomaa system \((X, \tau: X \to \Exp(X))\) with \(X=\{x_1, x_2, \dots, x_n\}\) induces a mapping \(\bar{\tau}: \Exp^{n} \to \Exp^{n}\).
Intuitively, this mapping takes a vector $\vec{\expr{e}} = (\expr{e}_1, \dots, \expr{e}_n)$, and produces a new vector where the $i$-th component is the evaluation of $\tau(x_i)$ when each $x_j$ is substituted by $\expr{e}_j$.
More formally, given this $\vec{\expr{e}}$, we define \(e : X \to \Exp\) by \(e(x_i)=\expr{e}_i\), and set \(\bar{\tau}(\vec{\expr{e}}) = \left((e^{\#} \circ \tau)(x_1), \dots, (e^{\#} \circ \tau)(x_n)\right)\).

To establish soundness of \acro{UA}, we first show that \(\bar{\tau}\) is \(\frac{1}{2}\)-Lipschitz on the pseudometric space \((\Exp^n, d)\), where $d$ is the metric that arises from the $n$-fold product of \((\Exp, \partial)\).
\begin{restatable}{lemma}{contractivess}\label{lem: contractiveness}
    Given a Salomaa system \((X, \tau: X \to \Exp(X))\), the map \(\bar{\tau} : \Exp^n \to \Exp^n\) from the pseudometric space \((\Exp^n, d)\) to itself is \(\frac{1}{2}\)-Lipschitz.
\end{restatable}

Finally, we can prove the following.
\begin{lemma}\label{lem: ua_satisfied_by_bisimilarity}
    \acro{UA} is satisfied by bisimilarity.
\end{lemma}
\begin{proof}
    Let \((X, \tau)\) is a Salomaa system, with \(X = \{x_1, \dots, x_n\}\), and let \(f, g : X \to \Exp(X) \) be solutions up to \(\uaequiv\) to the system.
    Finally, let \(\vec{\expr{f}} = (f(x_1), \dots, f(x_n))\) and \(\vec{\expr{g}} = (g(x_1), \dots, g(x_n))\).
    Assume that the premises are satisfied by the bisimilarity, \(f(x_i) \bisim (f^{\#} \circ \tau)(x_i)\) and \(g(x_i) \bisim (g^{\#} \circ \tau)(x_i)\) for all \(1 \leq i \leq n\).
    In other words, we have that \(d(\bar{\tau}(\vec{\expr{f}}), \vec{\expr{f}}) = 0\) and \(d(\bar{\tau}(\vec{\expr{g}}), \vec{\expr{g}}) = 0\)

    Let \(d(\bar{\tau}(\vec{\expr{f}}), \bar{\tau}(\vec{\expr{g}})) = k\) for some \(k \in \mathbb{R}^{+}\). Then, since $d$ is ultra,
    \[
        d(\vec{\expr{f}}, \vec{\expr{g}}) = \max\{d(\vec{\expr{f}}, \bar{\tau}(\vec{\expr{f}})), d(\bar{\tau}(\vec{\expr{f}}), \vec{\expr{g}})\}=d(\bar{\tau}(\vec{\expr{f}}), \vec{\expr{g}}) = \max\{d(\bar{\tau}(\vec{\expr{f}}), \bar{\tau}(\vec{\expr{g}})), d(\bar\tau(\vec{\expr{g}}), \vec{\expr{g}})\}=k
    \]
    By \cref{lem: contractiveness}, we find
    \(k=d(\bar{\tau}(\vec{\expr{f}}), \bar{\tau}(\vec{\expr{g}}))\leq \frac{1}{2}d(\vec{\expr{f}}, \vec{\expr{g}})=\frac{1}{2}k\)
    which implies that
    \(
        d(\vec{\expr{f}}, \vec{\expr{g}}) = 0
    \). Because of the definition of the product pseudometric on \((\Exp, \partial)\), we have
   \( f(x_i) \bisim g(x_i)\)
    for all \(1 \leq i \leq n\). Therefore, the conclusion of the \acro{UA} is satisfied by bisimilarity.
\end{proof}

Because of the above lemma and \cref{thm: soundness}, both \acro{UA} and the axioms of \(\equiv\) are contained in \(\bisim\), the greatest bisimulation on \((\Exp, \partial)\). Recall that \(\uaequiv\) is the least congruence containing those rules. Since \(\bisim\) on \((\Exp, \partial)\) is a congruence (\cref{thm: greatest_bisim_is_congruence}), we have that \(\uaequiv\) is sound.
\begin{theorem}[Soundness with \acro{UA}]
    For all \(\expr{e}, \expr{f} \in \Exp\) if \(\expr{e} \uaequiv \expr{f}\) then \(\expr{e} \bisim \expr{f}\).
\end{theorem}

\subparagraph*{Completeness. }
After all the hard work is done, the proof of completeness follows via the same line of reasoning as the one for \acro{GKAT}~\cite{Smolka:2020:Guarded,Schmid:2021:Guarded}.
\begin{theorem}[Completeness]
    For all \(\expr{e}, \expr{f} \in \Exp\) if \(\expr{e} \bisim \expr{f}\) then \(\expr{e} \uaequiv \expr{f}\)
\end{theorem}
\begin{proof}
    Let \(R \subseteq \Exp \times \Exp\) be a bisimulation with a transition structure \(\rho : R \times \At \to \D(\2 + \V + \Act \times R) \) relating \acro{ProbGKAT} automata \((\cogen{\expr{e}}{\partial}, \partial)\) and \((\cogen{\expr{f}}{\partial}, \partial)\) such that \((\expr{e}, \expr{f}) \in R\).
    Let \(\pi_1, \pi_2\) be the projection homomorphisms from \((R, \rho)\) to \((\cogen{\expr{e}}{\partial}, \partial)\) and \((\cogen{\expr{f}}{\partial}, \partial)\) respectively.
    Since both \(\cogen{\expr{e}}{\partial}\) and \(\cogen{\expr{f}}{\partial}\) are finite (by \cref{lem:locally_finite}), so is $R$.

    Let \(j, k\) be the inclusion homomorphisms of \((\cogen{\expr{e}}{\partial}, \partial)\) and \((\cogen{\expr{f}}{\partial}, \partial)\) in \((\Exp, \partial)\).
    We can construct two homomorphisms \([-]_{\equiv} \circ j \circ \pi_1\) and \([-]_{\equiv} \circ k \circ \pi_2\) from \((R, \rho)\) to \((\Exp/{\equiv}, \bar{\partial})\).
    By \cref{lem: solutions_are_homomorphisms}, \(j \circ \pi_1\) and \(k \circ \pi_2\) are solutions up to \(\equiv\) to the Salomaa system associated with \((R, \rho)\).
    Since \(\equiv\) is contained in \(\uaequiv\), those are immediately also solutions up to \(\uaequiv\).

    Because of \acro{UA}, we have that  \((j \circ \pi_1)(\expr{g}, \expr{h}) \uaequiv (k \circ \pi_2)(\expr{g}, \expr{h})\) for all \((\expr{g}, \expr{h}) \in R\).
    Thus,
    \[
        \expr{e}
            \uaequiv j(\expr{e})
            \uaequiv (j \circ \pi_1)(\expr{e}, \expr{f})
            \uaequiv (k \circ \pi_2)(\expr{e}, \expr{f})
            \uaequiv k(\expr{f})
            \uaequiv \expr{f}
            \qedhere
    \]
\end{proof}

\section{Decidability and Complexity}\label{sec:decidability}

To decide whether $\expr{e} \uaequiv \expr{f}$, we need to demonstrate the existence of a bisimulation between the states $\expr{e}$ and $\expr{f}$ in $(\Exp, \partial)$.
Since bisimulations need only involve reachable states, it suffices to find this bisimulation within $\cogen{\expr{e, f}}{\partial}$, the smallest subautomaton of $(\Exp, \partial)$ containing $\expr{e}$ and $\expr{f}$, which is also the union of $\cogen{\expr{e}}{\partial}$ and $\cogen{\expr{f}}{\partial}$; this automaton is finite by \Cref{lem:locally_finite}.
We thus focus on the problem of deciding bisimilarity within a single finite \acro{ProbGKAT} automaton.

Our analysis in this section is facilitated by two simplifying assertions.
\begin{enumerate}
    \item
    To avoid having to compare real (infinite-precision) probabilities, we limit \acro{ProbGKAT} expressions to rational probabilities \(r \in [0,1] \cap \mathbb{Q}\) in this section.
    This restriction is compatible with the earlier operators on probabilities, which all preserve rationality.
    \item
    Equivalence of \acro{GKAT} proper is co-NP-hard~\cite{Smolka:2020:Guarded}, simply because Boolean unsatisfiability can trivially be encoded in the language of tests.
    We take a fixed-parameter approach, assuming that $\At$, the set of atoms that can appear on transitions, is fixed beforehand.
\end{enumerate}

\smallskip\noindent\textsf{\textbf{Coalgebraic partition refinement. }}
%
We rely on \emph{partition refinement}~\cite{Kanellakis:1983:CCS,Kanellakis:1990:CCS,Paige:1987:Three}, which effectively computes the \emph{largest bisimulation} on an automaton, by approximating it from above.
In the \emph{coalgebraic} presentation of partition refinement~\cite{Wissmann:2020:Efficient}, which we instantiate to our setting, automata of various types are encoded as abstract graphs.
More specifically, an automaton is encoded in two maps $o: X \to O$ and $\ell: X \to \mathcal{B}(L \times X)$, where
\begin{itemize}
    \item
    $X$ is a set of \emph{nodes} that represent (partial) states of the automaton;
    \item
    $O$ is a set of \emph{observable values} at each node;
    \item
    $L$ is a set representing possible \emph{labels} of edges between nodes;
    \item
    $\mathcal{B}(L \times X)$ is a multiset of pairs representing \emph{edges} between nodes.
\end{itemize}
Subject to a number of coherence conditions on the encoding (omitted here), coalgebraic partition refinement yields an $\mathcal{O}(n \log |X|)$ algorithm to compute the largest bisimulation on an automaton, where $n = \sum_{x \in X} |\ell(x)|$ is the number of \emph{edges} of the automaton.

\subparagraph*{Encoding \acro{ProbGKAT} automata. }
Coalgebraic partition refinement provides suitable encodings for well-known transition types, as well as methods to soundly obtain encodings of composite transition types~\cite{Wissmann:2020:Efficient}.
The details of these techniques are beyond the scope of this paper, but the underlying idea is fairly intuitive: composite transition types are encoded by inserting synthetic nodes that represent partially evaluated states --- not unlike how our drawings contain intermediate nodes that are the target of $\alpha$-labelled arrows.
More precisely, the nodes of an encoded \acro{ProbGKAT} automaton $(Q, t)$ are three-sorted:
\begin{enumerate}
    \item
    every state of the automaton is a node; and
    \item
    every ``intermediate'' state (the small circles in our drawings) is a node; and
    \item
    every probabilistic edge gives rise to another node.
\end{enumerate}
Nodes of the third kind separate the dashed arrows in our drawings (labelled with a probability as well as an action) into two arrows, each of which is labelled by one value.

Formally, we choose $X := Q + I + \Act \times Q$ as our set of nodes, where $I := \{ t(q)_\alpha : q \in Q, \alpha \in \At \}$.
We also set $L := \At + \mathbb{Q} + \Act$.
The map $\ell: X \to \mathcal{B}(L \times X)$ is then defined by:%
\footnote{Here, $\mset{ - \mid - }$ denotes multiset comprehension, where each element occurs at most once.}
\[
    \ell(x) :=
        \begin{cases}
            \mset{ (\alpha, t(q)_\alpha) \mid \alpha \in \At } & x = q \in Q \\
            \mset{ (d(\action{p}, q), (\action{p}, q)) \mid \action{p} \in \Act, q \in Q } & x = d \in I \\
            \mset{ (\action{p}, q) } & x = (\action{p}, q) \in \Act \times Q
        \end{cases}
\]
In other words, $\ell$ labels the edges between nodes of the first and second kind with an atom, the edges between nodes of the second and third kind with a probability, and the edges between nodes of the third and first kind with an action.

Observables represent the probabilities assigned to acceptance, rejection, or a return value by nodes of the second kind.
Formally, $O := 1 + \mathbb{Q}^{2 + \V}$, where $* \in 1$ means ``no observable value'', and values from $\mathbb{Q}^{2 + \V}$ assign a probability to each $\xi \in 2 + \V$.
We can then define $o: X \to O$ by setting $o(d)(\xi) := d(\xi)$ when $d \in I$, and $o(x) := *$ otherwise.%
\footnote{If the coalgebraic approach from~\cite{Wissmann:2020:Efficient} is followed to the letter, the observable map for nodes of the third kind behaves slightly differently; we simplify our encoding here for the sake of presentation.}

\subparagraph*{Deciding bisimilarity. }
We can now leverage the encoding given above to decide bisimilarity.
\begin{restatable}{theorem}{decideequiv}%
\label{thm:decide-equiv}
If all probabilities are rational and $\At$ is fixed, then bisimilarity of states in a \acro{ProbGKAT} automaton $(Q, t)$ is decidable in time $\mathcal{O}(|Q|^2 |\Act| \log (|\Act \times Q|))$.
\end{restatable}
\begin{proof}
    The results from~\cite{Wissmann:2020:Efficient} ensure that our encoding of \acro{ProbGKAT} automata can be equipped with an appropriate interface that allows their algorithm to decide equivalence.

    As for the complexity, we instantiate their abstract complexity result by computing the parameters.
    The number of nodes and edges can be bound from above fairly easily, as follows:
    \begin{align*}
        |X|
            &= |Q| + |I| + |\Act \times Q|
            \leq |Q| + 2 \cdot |\Act| \cdot |Q|
        \\
        n
            &= \sum_{x \in X} |\ell(x)|
            \leq |Q| \cdot |\At| + |Q|^2 \cdot |\Act| + |Q| \cdot |\Act|
    \end{align*}
    Since $\At$ is fixed, the claimed complexity then follows.
    \end{proof}
This allows us to conclude that bisimilarity of \acro{ProbGKAT} expressions is also decidable.
\begin{restatable}{corollary}{decidable}%
\label{cor:decidable}
If all probabilities are rational and $\At$ is fixed, then \acro{ProbGKAT} equivalence of $\expr{e}, \expr{f} \in \Exp$ is decidable in time $\mathcal{O}(n^3\log n)$, where $n = \#(\expr{e}) + \#(\expr{f})$.
\end{restatable}
\begin{proof}
By \cref{lem:locally_finite}, $\cogen{\expr{e}, \expr{f}}{\partial}$ is of size at most $n$, and the number of distinct actions $\expr{e}$ or $\expr{f}$ is fixed from above by $n$ as well.
The claim then follows by \cref{thm:decide-equiv}.
\end{proof}

\section{Related work}\label{sec:related}

Our work builds on \acro{GKAT}, a strictly deterministic fragment~\cite{Kozen:2008:Bohm} of Kleene Algebra with Tests (\acro{KAT}).
\acro{KAT} has been used in several verification tasks, such as cache control~\cite{Cohen:1994:Lazy}, compiler optimisations~\cite{Kozen:2000:Certification}, source-to-source translations~\cite{Angus:2001:Kleene}, and network properties~\cite{Anderson:2014:NetKAT,Foster:2015:Coalgebraic,probnetkat,cantor,mcnetkat,cnetkat}, and was generalised to include fuzzy logics~\cite{Gomes:2019:Generalising}.
\acro{GKAT} admits a Salomaa-style~\cite{Salomaa:1966:Two} axiomatisation of trace equivalence~\cite{Smolka:2020:Guarded} and bisimilarity~\cite{Schmid:2021:Guarded}, both relying on the Uniqueness of Solutions axiom, and completeness without it remains open, though completeness of a fragment of \acro{GKAT} was recently proved by~\cite{Schmid:2023:Complete}.

\acro{GKAT} modulo bisimilarity and Milner's interpretation of regular expressions arise as fragments of the \emph{parametrised processes} framework~\cite{Schmid:2022:Processes}; this is not the case for \acro{ProbGKAT} due to a different treatment of loops.
The uniqueness axiom was originally introduced by Bergstra and Klop under the name Recursive Specification Principle (\acro{RSP})~\cite{Bergstra:1985:Verification} and used in axiomatisations of process calculi~\cite{Bergstra:1988:Process}. The general pattern of their proofs of completeness is similar to ours, although the key challenge is the extension to the probabilistic setting.

Our paper also builds up the vast line of research on probabilistic bisimulation~\cite{Larsen:1991:Bisimulation, Segala:1994:Probabilistic,Desharnais:1999:Labelled} and the coalgebraic approach to systems with probabilistic transitions~\cite{deVink:1999:Bisimulation, Bartels:2003:Hierarchy, Desharnais:1999:Labelled, Sokolova:2011:Probabilistic}.
More concretely, we relied on relation refinement characterisation of bisimilarity~\cite{Staton:2011:Relating}, natural metrics on the final coalgebras for \(\omega\)-accessible endofunctors~\cite{Barr:1993:Terminal, Worrell:2005:Final}, coalgebraic completeness theorems~\cite{Jacobs:2006:Bialgebraic,Silva:2010:Kleene,Schmid:2021:Star} and minimisation algorithms for coalgebras~\cite{Wissmann:2020:Efficient,Deifel:2019:Generic,Wissmann:2022:Quasilinear,Jacobs:2023:Fast}.
Axiomatisations of probabilistic bisimulation were extensively studied in the process algebra community, including a recursion-free process algebra of Bandini and Segala~\cite{Bandini:2001:Axiomatizations} and recursive calculi of Stark and Smolka~\cite{Stark:2000:Complete} and Mislove, Ouaknine and Worrell~\cite{Mislove:2004:Axioms}. Aceto, Ésik and Ingólfsdóttir~\cite{Aceto:2002:Equational} gave an alternative axiomatisation of Stark and Smolka's calculus by extending Iteration Theories~\cite{Bloom:1993:Iteration, Elgot:1975:Monadic} with equational axioms.

Probabilistic Kleene Algebra (\acro{pKA})~\cite{McIver:2008:Using} relaxes the axioms of \acro{KA} to accommodate reasoning about probabilistic predicate transformers; its axioms are complete w.r.t.\ simulation equivalence of NFAs~\cite{McIver:2011:On}.
\acro{pKA} was also extended with a probabilistic choice operator and concurrency primitives~\cite{McIver:2013:Probabilistic}, but completeness this system was not considered.
\acro{ProbNetKAT}~\cite{probnetkat, mcnetkat} is a domain-specific language for reasoning about probabilistic effects in networks based on \acro{KAT}, which features a probabilistic choice operator, however, axiomatisation of the obtained language was not studied.

\section{Conclusion and Future Work}\label{sec:future}
We have presented \acro{ProbGKAT}, a language for reasoning about uninterpreted programs with branching and loops, with both Boolean and probabilistic guards.
We provided an automata-theoretic operational model and characterised bisimilarity for these automata.
We gave a sound and complete axiomatisation of bisimilarity, relying on the Uniqueness of Solutions (\acro{UA}) axiom, and showed bisimilarity  can be efficiently decided in $O(n^3 \log n)$ time.

A first natural direction for future work is the question whether the more traditional language semantics of \acro{GKAT} can be lifted to \acro{ProbGKAT} and axiomatised.
More broadly, we would like to investigate notions of \acro{ProbGKAT} expression equivalence more permissive than bisimilarity, including the notion of \emph{bisimulation distance}~\cite{Baldan:2018:Coalgebraic} and its possible axiomatisations based on \emph{quantitative equational logic}~\cite{Mardare:2016:Quantitative,Bacci:2016:Complete}.

A second direction touches on the problem of completeness without \acro{UA}, which is still open for \acro{(Prob)GKAT}.
In light of recent completeness results for the skip-free fragment of \acro{GKAT}~\cite{Schmid:2023:Complete} modulo bisimilarity and trace equivalence, we are interested to study the skip-free fragment of \acro{ProbGKAT}.
The proofs in~\cite{Schmid:2023:Complete} do not immediately generalise to \acro{ProbGKAT} as probabilities do not obviously embed into (1-free) regular expressions.

Similarly to \acro{GKAT}, \acro{ProbGKAT} is strictly deterministic and thus avoids known complications of combining nondeterminism with probabilistic choice~\cite{Jones:1990:Probabilistic, Varacca:2006:Distributing, Goy:2020:Combining}.
We are interested if the recent work on combining multisets and probabilities via distributive laws~\cite{Jacobs:2021:From, Kozen:2023:Silva} could be applied to extending our developments with nondeterminism.

\acro{ProbGKAT} can express only uninterpreted programs, hence it cannot be used to reason about programs involving mutable state. An example of a probabilistic program with state is P\'{o}lya's urn~\cite{Mahmoud:2008:Polya}. One way of adding mutable state~\cite{Grathwohl:2014:KATB} to \acro{ProbGKAT} is by adding \emph{hypotheses}~\cite{Cohen:1994:Hypotheses}. Unfortunately, adding hypotheses can lead to undecidability or incompleteness~\cite{Kozen:2002:On}, although there are forms of hypotheses that retain completeness~\cite{Kozen:2014:Kleene,Doumane:2019:Kleene,Pous:2022:On} and exploring this is as an interesting direction for future work.





\bibliography{bibliography}

\appendix

\section{Coalgebra}\label{apx:coalgebra}
In the main text of the paper, we avoided using the language of universal coalgebra~\cite{Rutten:2000:Universal,Gumm:2000:Elements} in the presentation, so as not to distract from the main concepts, which can be described concretely.
We have however used coalgebra in our development, and concrete definitions in the main text are instances of abstract notions.
This is helpful in simplifying proofs, so in the appendix, we will present the proofs of the results using coalgebra.

We assume that the reader is familiar with the basic notions of category theory, such as functors, pullbacks and natural transformations.
We first recall the basic notions from universal coalgebra; for a more detailed introduction, we refer to~\cite{Rutten:2000:Universal,Gumm:2000:Elements}.

\begin{definition}\label{def:coalgebra}
    Let \(\B\) a \(\Set\)-endofunctor. A \emph{\(\B\)-coalgebra} is a pair \((X, \beta)\) where \(X\) is a set and \(\beta : X \to \B X\) is a \emph{transition function}.
    A homomorphism between two \(\B\)-coalgebras \((X, \beta)\) and \((Y, \gamma)\) is a function \(h: X \to Y\) satisfying \(\B h \circ \beta = \gamma \circ h\).
    \(\B\)-coalgebras and homomorphisms between them form a category, which we denote \(\coalg{\B}\).
\end{definition}

Recall that the set \(\D(X)\) of finitely supported probability distributions is an endofunctor on \(\Set\).
As alluded to before, \acro{ProbGKAT} automata can modelled as coalgebras for the functor \(\G = \D(\2 + \V + \Act \times \Id)^{\At}\), and we will study them as such going forward.

\begin{definition}\label{def:subcoalgebra}
    Let \((X, \beta)\) and \((Y, \gamma)\) be \(\B\)-coalgebras, with \(Y \subseteq X\). If the inclusion \(i : Y \to X\) is a \(\B\)-coalgebra homomorphism, then \((Y, \gamma)\) is called a \emph{subcoalgebra} of \((X, \beta)\).
\end{definition}

There can be at most one coalgebra structure map \(\gamma : Y \to \B Y\) that makes the inclusion \(i : Y \to X\) a \(B\)-coalgebra homomorphism from \((Y, \gamma)\) to \((X, \beta)\)~\cite{Rutten:2000:Universal}.
Subcoalgebras of any \(\B\)-coalgebra \((X, \beta)\) form a complete lattice~\cite{Rutten:2000:Universal}.
Given \(x \in X\), we will write \(\cogen{x}{\beta}\) for the smallest subcoalgebra of \((X, \beta)\) containing \(x\), the \emph{subcoalgebra generated by \(x\)}.

\begin{definition}\label{def:coalgebraic_bisimulation}
    Let \((X, \beta)\) and \((Y, \gamma)\) be \(\G\)-coalgebras.
    A relation \(R \subseteq X \times Y\) is a \emph{bisimulation} if there exists a function \(\rho: R \to \B R\) such that the projections \(\pi_1, \pi_2 : R \to X\) are \(\B\)-coalgebra homomorphisms from \((R, \rho)\) to \((X, \beta)\) and \((Y, \gamma)\) respectively.
    We say the elements \(x \in X\) and \(y \in Y\) are \emph{bisimilar} if there exists a bisimulation \(R\) such that \((x,y) \in R\)
\end{definition}

Note that \cref{def:homomorphism_of_transition_systems} is an instantiation of the abstract definition of coalgebra homomorphism from \cref{def:coalgebra} to \(\G\).
Similarly, \cref{def:bisimulation_of_transition_systems} is an instantiation of \cref{def:coalgebraic_bisimulation}.

\begin{restatable}{proposition}{bounded}\label{prop:bounded_and_pullbacks}
\begin{enumerate}
    \item The functor \(\G\) is bounded and preserves weak pullbacks
    \item There exists a \(\G\)-coalgebra \((Z, \zeta)\) which is a final object in \(\coalg{\G}\). In other words, for any \(\G\)-coalgebra \((X, \beta)\) there exists a unique homomorphism \(\beh_\beta : X \to Z\)
    \item Let \((X, \beta)\) and \((Y, \gamma)\) be \(\G\)-coalgebras. The elements \(x \in X\) and \(y \in Y\) are bisimilar if and only if \({\beh_\beta(x)}={\beh_\gamma(y)}\)
\end{enumerate}
\end{restatable}
\begin{proof}
(1) It was proved by Moss in~\cite{Moss:1999:Coalgebraic} that \(\D\) preserves weak pullbacks.
Gumm and Schr\"{o}der~\cite{Gumm:2000:Coalgebraic} showed that preservation of weak pullbacks is closed under functor composition, products and coproducts.
Since the \(\Set\) endofunctors \(F(X)=X^{A}\), \(G(X)=B + X\) and \(H(X)=C\) (where \(A, B\) and \(C\) are arbitrary sets) preserve weak pullbacks, so does $\G$.

(2) The functor \(\D\) is known to be bounded by \(\omega\)~\cite[Theorem~4.6]{deVink:1999:Bisimulation}.
Since \(\Act\) and \(\2 + \V\) are finite, they are also \(\omega\)-bounded.
Boundedness is preserved under functor composition, binary products and binary coproducts~\cite[Corollary~4.9]{Gumm:2002:Bounded}.
Since \(\At\) is finite, the exponential functor \(\Id^\At\) is also bounded.
Altogether, we can conclude that $\G$ is bounded.
Existence of the final coalgebra follows from boundedness~\cite[Theorem~10.3]{Rutten:2000:Universal}.

(3) Follows from the fact that \(\G\) preserves weak pullbacks~\cite[Theorem~9.3]{Rutten:2000:Universal}.
\end{proof}
\section{Proofs from \texorpdfstring{\cref{sec:opsem}}{Section~\ref{sec:opsem}}}
\locallyfinite*
\begin{proof}
    We adapt the analogous proof for \acro{GKAT}~\cite{Schmid:2021:Guarded}. For any \(\expr{e} \in \Exp\), let \(|\cogen{\expr{e}}{\partial}|\) be the cardinality of the carrier set of the least subcoalgebra of \((\Exp, \partial)\) containing \(\expr{e}\). We show by induction that for all \(\expr{e} \in \Exp\) it holds that \(|\cogen{\expr{e}}{\partial}|\leq \#(\expr{e})\).

     For the base cases, observe that for \(\var{v} \in \V\) the generated subcoalgebra has exactly one state, which outputs the appropriate value with probability \(1\). Hence, \(\#(\var{v}) = 1\). Similarly, for \(\test{b} \in \Bexp\), we have \(\#(\test{b})=1\)
     For \(\action{p} \in \Act\), we have two states; the initial state, which transitions with probability \(1\) on \(\action{p}\) to the state which outputs \(\accept\) with probability \(1\).

     For the inductive cases, assume that \(|\cogen{\expr{e}}{\partial}| \leq \#(e)\), \(|\cogen{\expr{f}}{\partial}| \leq \#(f)\), \(\test{b} \in \Bexp\) and \(\prob{r} \in [0,1]\).
     \begin{itemize}
         \item
          Every derivative of \(\expr{e} +_\test{b} \expr{f}\) is either a derivative of \(\expr{e}\) or \(\expr{f}\) and hence \(|\cogen{\expr{e} +_\test{b} \expr{f}}{\partial}| \leq |\cogen{\expr{e}}{\partial}| + |\cogen{\expr{f}}{\partial}| = \#(\expr{e}) + \#(\expr{f}) = \#(\expr{e} +_\test{b} \expr{f})\).
          By analogous reasoning, \(|\cogen{\expr{e} \oplus_{\prob{r}} \expr{f}}{\partial}| \leq \#(\expr{e} \oplus_\prob{r} \expr{f}) \).

          \item
          In the case of \(\expr{e}\seq\expr{f}\), every derivative of this expression is either a derivative of \(\expr{f}\) or some derivative of \(\expr{e}\) followed by \(\expr{f}\). Hence, \(|\cogen{\expr{e}\seq\expr{f}}{\partial}| = |\cogen{\expr{e}}{\partial}\times \{\expr{f}\}| + |\cogen{\expr{e}}{\partial}| \leq \#(\expr{e}) + \#(\expr{f}) = \#(\expr{e}\seq\expr{f}) \).

          \item
          For the probabilistic loop case, observe that every derivative of \(\expr{e}^{[\prob{r}]}\) is a derivative of \(\expr{e}\) followed by \(\expr{e}^{[\prob{r}]}\). It can be easily observed, that there is as many derivatives of \(\expr{e}^{[\prob{r}]}\) as derivatives of \(\expr{e}\). Hence, \(|\cogen{\expr{e}^{[\prob{r}]}}{\partial}| \leq |\cogen{\expr{e}}{\partial}| = \#(\expr{e}) = \#(\expr{e}^{[\prob{r}]})\). We omit the case of guarded loop, as reasoning is identical to the case of the probabilistic loop.
          \qedhere
     \end{itemize}
\end{proof}

\section{Proofs from section \texorpdfstring{\cref{sec:bisim}}{Section~\ref{sec:bisim}}}\label{apx:bisim}
A \emph{flow network} is a pair \((G, c)\) where \(G=(V,E)\) is a directed graph, and the edges \(E \subseteq V \times V\) are equipped with a capacity function \(c : E \to \eR\).
An \emph{\((s,t)\)-flow} through a network \((G,c)\) where \(s, t \in V\) is a function \(f : E \to \eR \) satisfying the following:
\begin{enumerate}
	\item For all \(v \in V \setminus \{s,t\}\) \[\sum_{(u,v) \in E} f(u,v) = \sum_{(v,u)\in E} f(v,u)\]
	\item For all \(e \in E\), \(f(e)\leq c(e)\) (Admissibility of the flow)
	\item There exists a constant \(F \in \eR\) called the value of the flow, satisfying \[F = \sum_{(s,u) \in E} f(s,u) = \sum_{(v,t) \in E} f(v,t)\]
\end{enumerate}
Let \(S\subseteq V\), such that \( s \in S\).
An \emph{\((s,t)\)-cut} is the set of edges \(E(S, V \setminus S)=\{ (u, v) \in E \mid u \in S, v \in V \setminus S\}\).
The \emph{capacity} of the cut is given by \[C_S = \sum_{\substack{(u,v) \in E(S, V \setminus S)}} c(u,v)\]
The following result about flows and cuts is well known.

\begin{theorem}[Max-flow min-cut theorem~\cite{Dantzig:1957:maxflow}]\label{thm:maxflow}
    Let \((G,c)\) be a flow network with \(G=(V,E)\) and let \(s,t \in V\) be vertices. The maximum value of an admissible \((s,t)\)-flow equals the minimum capacity of any \((s,t)\)-cut.
\end{theorem}

With this in hand, we can verify our characterisation of bisimulations.

\begin{restatable}{lemma}{maxflowbisim}\label{lem:max_flow_bisimulation}
    Let \((X, \beta)\) and \((Y, \gamma)\) be \acro{ProbGKAT} automata and let \(R \subseteq X \times Y\) be a relation. \(R\) is a bisimulation if and only if for all \((x,y) \in R\) and \(\alpha \in \At\), the following hold.
    \begin{enumerate}
        \item For all \(o \in \2 + \V\), \(\beta(x)_\alpha(o)=\gamma(y)_\alpha(o)\)
        \item For all \(A \subseteq X\) and all \(\action{p} \in \Act\), \(\beta(x)_\alpha[\{\action{p}\}\times A]\leq \gamma(y)_\alpha[\{\action{p}\} \times R(A)]\)
        \item For all \(B \subseteq Y\) and all \(\action{p} \in \Act\), \(\gamma(y)_\alpha[\{\action{p}\}\times B]\leq \beta(x)_\alpha[\{\action{p}\} \times R^{-1}(B)]\)
    \end{enumerate}
\end{restatable}
\begin{proof}
    Let \((X, \beta)\) and \((Y, \gamma)\) be \(\G\)-coalgebras and let \(R \subseteq X \times Y\) be a relation.
    \cref{def:coalgebraic_bisimulation} states that \(R\) is a bisimulation if and only if there exists a \(\G\)-coalgebra structure map \(\rho: R \to \G R \), which makes the canonical projection maps into homomorphsims from \((R, \rho)\) to \((X, \beta)\) and \((Y, \gamma)\) respectively.
    We can work out this definition more concretely.
    Namely, relation \(R \subseteq X \times Y\) is a bisimulation if and only if there exists \(\G\) coalgebra structure map \(\rho: R \to \G R\), such that for all \((x, y) \in R\) and all \(\alpha \in \At\), the following hold:
    \begin{enumerate}
        \item For all \(o \in 2 + \V\), \(\beta(x)_\alpha(o) = \rho(x,y)_\alpha(o) = \gamma(y)_\alpha(o)\)
        \item For all \((\action{p},x') \in \Act \times X\) \[\beta(x)_\alpha(\action{p}, x')=\sum_{(x',y')\in R} \rho(x,y)_\alpha(\action{p},(x',y'))\]
        \item For all \((\action{p},y') \in \Act \times Y\) \[\gamma(y)_\alpha(\action{p}, y')=\sum_{(x',y')\in R} \rho(x,y)_\alpha(\action{p},(x',y'))\]
    \end{enumerate}
    Condition (1) of the above definition is readily equivalent to the condition (1) of ths lemma.

    We now show that existence of the coalgebra structure map \(\rho: R \to \G R\) satisfying the latter conditions is equivalent to conditions of this lemma, by constructing a family of flow networks and employing the max-flow min-cut theorem.
    Assume \((x,y) \in R\).
    For each \(\alpha \in \At\) and \(p \in \Act\), construct a flow network \((G_{\alpha, \action{p}}, c_{\alpha, p})\) where \(G_{\alpha, \action{p}} = (V_{\alpha, \action{p}}, E_{\alpha, \action{p}})\), as follows.
	\begin{align*}
    V_{\alpha, \action{p}} ={}& \{s,t\} \cup \{x' \in X \mid \beta(x)_\alpha(\action{p},x') > 0\} \cup \{y' \in Y \mid \gamma(y)_\alpha(\action{p},y') > 0\} \\
    E_{\alpha, \action{p}} ={}& \{(s,x') \mid x' \in X \cap V_{\alpha, \action{p}}\} \\
                              & {} \cup \{(x',y') \mid x' \in (X \cap V_{\alpha, \action{p}}), y' \in (Y \cap V_{\alpha,\action{p}}), (x', y') \in R\} \\
                              & {} \cup \{(y',t) \mid y' \in Y \cap V_{\alpha, \action{p}}\}
    \end{align*}

    Informally speaking, \(V_{\alpha,\action{p}}\) contains the designated source and target vertices \(s\) and \(t\), as well as the support of \(\beta(x)_\alpha\) and \(\gamma(y)_\alpha\) restricted to tuples having \(\action{p}\) in the first coordinate.

    We add an edge from \(s\) (resp.\ \(t\)) to any element in the support of \(\beta(x)_\alpha\) (resp.\ \(\gamma(y)_\alpha\)) contained in the set of vertices.
    Vertices are connected when they are related by \(R\).

	The capacity function \(c_{\alpha,\action{p}}\) is defined as follows: for all \(x' \in X \cap V_{\alpha,p}\) we set \(c_{\alpha,p}(s,x')=\beta(x)_\alpha(\action{p},x')\) and for all \(y' \in Y \cap V_{\alpha, \action{p}}\) we set \(c_{\alpha,\action{p}}(y',t)=\gamma(y)_\alpha(\action{p},y')\). Otherwise, for the remaining \((x',y') \in E\), we choose \(c_{\alpha, \action{p}}(x', y') = +\infty\).

	First, assume that conditions (1--3) of the lemma hold. Let \(S\subseteq V_{\alpha, \action{p}}\) be the set containing \(s\) such that \(C = E(S, V_{\alpha,\action{p}} \setminus S)\) is the minimal capacity cut.
    This set cannot contain the central edges of infinite capacity, as otherwise the corresponding flow would not be admissible.

	Without loss of generality, \(C = A \cup B \) where \(A = \{(s,x') \mid s \in S, x' \in V_{\alpha, \action{p}} \setminus S\}\) and \(B = \{(y',t) \mid y' \in S, t \in V_{\alpha, \action{p}} \setminus S\}\). The capacity of the minimal cut, and (by \Cref{thm:maxflow}) the value of the maximum admissible flow, is given by:
	\[
		F_{\alpha,\action{p}} = \beta(x)_\alpha[\{\action{p}\}\times \pi_2(A)] + \gamma(y)_\alpha[\{\action{p}\} \times \pi_1(B)]
	\]
	Now, consider the set \(A' = \{(s,x') \in E \mid (s,x') \notin A \}\). Since elements of \(A'\) are not included in the cut, all vertices in \(\pi_2(A')\) appear in \(S\).
    Because of that, all elements related by \(R\) with those elements must also be included in \(S\), to avoid including the central edges of infinite capacity in the minimal cut. Hence \(R(\pi_2(A')) \subseteq S\), which implies that \(R(\pi_2(A')) \subseteq \pi_1(B)\).

    Using these observations, we can bound the value of the maximal flow from below:
	\begin{align*}
		F_{\alpha, \action{p}} &= \beta(x)_\alpha[\{\action{p}\}\times \pi_2(A)] + \gamma(y)_\alpha[\{\action{p}\} \times \pi_1(B)]\\
		&\geq \beta(x)_\alpha [\{\action{p}\}\times \pi_2(A)] + \gamma(y)_\alpha[\{\action{p}\} \times R(\pi_2(A'))]\\
		&\geq \beta(x)_\alpha [\{\action{p}\}\times \pi_2(A)] + \beta(x)_\alpha[\{\action{p}\} \times \pi_2(A')] \tag{2}\\
		&= \beta(x)_\alpha [\{\action{p}\} \times (X \cap V_{\alpha, \action{p}})] \tag{\(A\) and \(A'\) are disjoint}
	\end{align*}
    Because of the admissibility constraint, $F_{\alpha,\action{p}}$ is also bounded from above by
    \[\sum_{x' \in X \cap V_{\alpha, \action{p}}} c_{\alpha,\action{p}}(s,x') = \beta(x)_\alpha [\{\action{p}\} \times (X \cap V_{\alpha, \action{p}})]\]
    which means the lower bound given above is tight.
    By symmetric reasoning, we can derive
    \[F_{\alpha,\action{p}} = \gamma(y)_\alpha [\{\action{p}\} \times (Y \cap V_{\alpha, \action{p}})]\]
    In particular, this means that the flow at every edge connected to the source $s$ must be at capacity, because the sum of $c_{\alpha,\action{p}}(s, x')$ for $x' \in V_{\alpha,\action{p}}$ is at its maximum permitted value.

	Now, we define \(\rho : R \to GR\) in the following way.
	\begin{enumerate}
		\item For all \(\alpha \in \At\) and \(o \in \2 + \V\), set \(\rho(x,y)_\alpha(o)=\beta(x)_\alpha(o)=\gamma(y)_\alpha(o)\).
		\item For all \(\alpha \in \At, \action{p} \in \Act\) and \((x',y') \in R\), set \(\rho(x,y)_\alpha(\action{p},(x',y')) = f_{\alpha, \action{p}}(x',y')\)
	\end{enumerate}
	Condition (1) of the coalgebraic definition of the bisimulation holds immediately. To see (2), observe that for any \(\alpha \in \At, p \in \Act, x' \in X\), we have that
	\begin{align*}
		\beta(x)_\alpha(\action{p},x') &= c_{\alpha, \action{p}}(s, x')\\
        &= f_{\alpha, \action{p}}(s,x') \\
        &= \sum_{(x',y') \in E_{\alpha, p}} f_{\alpha, \action{p}}(x', y')\\
        &= \sum_{(x',y') \in R} \rho(x,y)_\alpha(\action{p},(x',y'))
	\end{align*}
	Condition (3) holds by symmetric argument.

	For the converse, assume that conditions (2) and (3) of the coalgebraic definition hold. We show that (2) holds by contradiction.
	Thus, assume the negation of condition (2) in the lemma statement, namely that for some \(\alpha \in \At\) and \(\action{p} \in \Act, A \subseteq X\) we have that
	\[\beta(x)_\alpha[\{\action{p}\} \times A] > \gamma(y)_\alpha[\{\action{p}\} \times R(A)]\] As before, the capacity of the minimal \((s,t)\)-cut for each flow network \((G_{\alpha, \action{p}}, c_{\alpha,\action{p}})\) is given by
	\begin{align*}
		F_{\alpha,p} &= \beta(x)_\alpha[\{\action{p}\} \times \pi_2(A)] + \gamma(y)_\alpha[\{\action{p}\} \times \pi_1(B)]\\
		&\geq  \beta(x)_\alpha[\{\action{p}\} \times \pi_2(A)] + \gamma(y)_\alpha[\{\action{p}\} \times R(\pi_1(A'))]\\
		&> \beta(x)_\alpha[\{\action{p}\} \times \pi_2(A)] + \beta(x)_\alpha[\{\action{p}\} \times R(\pi_1(A'))]\\
		&= \beta(x)_\alpha[\{\action{p}\} \times (X \cap V_{\alpha, \action{p}})]
	\end{align*}
	Hence, the maximal flow is below capacity.
    Thus there exists \(x' \in X\) for which
	\begin{align*}
		\beta(x)_\alpha(\action{p},x')&=c_{\alpha, \action{p}}(s,x')\\
        &> f_{\alpha, \action{p}}(s,x')\\
        &=\sum_{(x',y')\in E_{\alpha, \action{p}}}f_{\alpha, \action{p}} (x',y')\\
        &=\sum_{(x',y')\in R} \rho(x,y)_\alpha(\action{p}, (x',y'))
	\end{align*}
	which leads to contradiction.
    Condition (3) holds by symmetry.
\end{proof}

\larsenskou*
\begin{proof}
    Assume \(R\) is a bisimulation.
    Condition (1) holds immediately as a corollary of \cref{lem:max_flow_bisimulation}.
    To show (2), take any equivalence class \(Q \in {X}/{R}\).
    Since by assumption \(R\) is an equivalence relation, observe that \(R(Q)=Q=R^{-1}(Q)\).
    Assume \((x,y) \in R\).
    We can use \cref{lem:max_flow_bisimulation} to get the following for all \(\alpha \in At\), and \(\action{p} \in \Act\)
	\[
		\beta(x)_\alpha[\{\action{p}\} \times Q] \leq \gamma(x)_\alpha[\{\action{p}\} \times Q] \leq \beta(y)_\alpha[\{\action{p}\} \times Q]
	\]
	from which we can conclude (2).

	For the converse, observe that condition (1) of \cref{lem:max_flow_bisimulation} holds immediately.
    To see (2) take an arbitrary \(A \subseteq X\).
    Let \(A / {R}\) be the quotient of \(A\) by the relation \(R\) and let \(X / {R}\) be the quotient of \(X\) by \(R\).
    Observe that \(A / {R}\) is a partition of \(A\) and \(X / {R}\) is a partition of \(X\).

    Moreover, for each equivalence class \(P \in A / {R}\), there exists an equivalence class \(Q_P \in X / {R}\), such that \(P \subseteq Q_P = R(P)\).
    Because of monotonicity, we also have that \(\beta(x)_\alpha[{\action{p}} \times P] \leq \beta(x)_\alpha[\{\action{p}\} \times Q_P]\) for all \(\alpha \in \At\), \(\action{p} \in \Act\).
    By \(\sigma\)-additivity we have that
	\begin{align*}
		\beta(x)_\alpha[\{\action{p}\} \times A] &= \beta(x)_\alpha\left[\{\action{p}\} \times \bigcup_{P \in A / {R}} P \right] \\
        &= \sum_{P \in A / {R}} \beta(x)_\alpha[\{\action{p}\} \times P]\\
        &\leq \sum_{P \in A / {R}} \beta(y)_\alpha[\{\action{p}\} \times Q_P]\\
        &= \beta(y)_\alpha\left[\{\action{p}\} \times \bigcup_{P \in A / {R}} Q_P\right]\\
        &= \beta(y)_\alpha\left[\{\action{p}\} \times \bigcup_{P \in A / {R}} R(P)\right]\\
        &= \beta(y)_\alpha[\{\action{p}\} \times R(A)]
	\end{align*}
	which proves (2) for all \(\alpha \in \At, \action{p} \in \Act\). The case for (3) proceeds symmetrically.
\end{proof}

\subsection{Order theoretic characterisation of bisimulations}\label{apx:order}
The bisimulation functional $\Phi$ (\cref{def:bisimulation_functional}) can be shown to be monotone, while bisimulations can be characterised as postfixed points of that operator.
\begin{restatable}{lemma}{bisimfunctionalproperties}\label{lem:properties_of_bisim_functional}
    Let \((X, \beta)\) and \((Y, \gamma)\) be \(\G\)-coalgebras.
    The following hold:
    \begin{enumerate}
        \item \(\Phi_{\beta, \gamma}\) is monotone with respect to inclusion order
        \item \(R\) is a bisimulation between \((X, \beta)\) and \((Y, \gamma)\) if and only if \(R \subseteq \Phi_{\beta, \gamma}(R)\)
    \end{enumerate}
\end{restatable}
\begin{proof}
    For (1), let \(R \subseteq S \subseteq X \times Y\) and assume that \((x,y) \in \Phi_{\beta, \gamma}(R)\).
    Then for all \(\alpha \in \At\) and  \(o \in \2 + \V\) we have that \(\beta(x)_\alpha(o)=\gamma(y)_\alpha(o)\).
    For all \(\alpha \in \At\), \(\action{p} \in \Act\) and \(A \subseteq X\) we have that \(\beta(x)_\alpha[\{\action{p}\} \times A] \leq \gamma(y)_\alpha[\{\action{p}\} \times R(A)]\).
    Since \(R \subseteq S\), we have that \(R(A) \subseteq S(A)\).
    By monotonicity, it holds that \(\gamma(y)_\alpha(\{\action{p}\} \times R(A)) \leq \gamma(y)[\{\action{p}\} \times S(A)]\) and therefore \(\beta(x)_\alpha[\{\action{p}\} \times A] \leq \gamma(y)_\alpha[\{\action{p}\} \times S(A)]\).
    The remaining case is symmetric.
    Therefore \((x,y)\in \Phi_{\beta, \gamma}(S)\), which proves that \(\Phi_{\beta, \gamma}(R) \subseteq \Phi_{\beta, \gamma}(S)\).

    (2) is trivial as it is rephrasing of \cref{lem:max_flow_bisimulation}.
\end{proof}

Recall that the Knaster-Tarski fixpoint theorem states that the greatest fixpoint of a monotone endofunction \(f : X \to X\) on a complete lattice \((X, \sqsubseteq)\) is given by the following
\[
\gfp(f) = \bigsqcup \{ x \in X \mid x \sqsubseteq f(x)\}
\]

In the following, let \(\bisim_{\beta, \gamma}\) denote the greatest bisimulation between \((X, \beta)\) and \((Y, \gamma)\). The subscripts can be omitted when the coalgebras are obvious from the context.
\begin{restatable}{corollary}{greatestfixpoint}\label{cor:greatest_fixpoint}
    Let \((X, \beta)\) and \((Y, \gamma)\) be \(\G\)-coalgebras. The greatest fixpoint \(\gfp(\Phi_{\beta, \gamma})\) of the functional \(\Phi_{\beta, \gamma}\) is the greatest bisimulation between \((X, \beta)\) and \((Y, \gamma)\).
\end{restatable}
\begin{proof}
    Because of Knaster-Tarski fixpoint theorem, \(\gfp(\Phi_{\beta, \gamma})\) is a fixpoint, and in particular a post-fixed point.
    Thus, by \cref{lem:properties_of_bisim_functional} it is a bisimulation.

    Let \(S\) be an arbitrary bisimulation between \(\G\)-coalgebras \((X, \beta)\) and \((Y, \gamma)\).
    Again, by \cref{lem:properties_of_bisim_functional} we have that \(S \subseteq \Phi_{\beta, \gamma}(S)\) and therefore
    \[
    S \subseteq \Phi_{\beta, \gamma}(S) \subseteq \bigcup \{R \subseteq X \times Y \mid R \subseteq \Phi_{\beta, \gamma}(R) \} = \gfp(\Phi_{\beta, \gamma})
    \qedhere
    \]
\end{proof}

We can simplify the characterisation of \(\Phi\) when dealing with equivalence relations (similarly to \cref{lem:larsen_skou}).
In terms of notation, we will use \(\Phi_{\beta}\) to denote \(\Phi_{\beta, \beta}\).
\begin{restatable}{lemma}{eqpreservation}\label{lem:preservation_of_equivalences}
    Let \((X, \beta)\) be \(\G\)-coalgebra and let \(R \subseteq X \times X\) be an equivalence relation.
    We have \((x, y) \in \Phi_\beta(R)\) if and only if for all \(\alpha \in \At\), the following are true.
    \begin{enumerate}
        \item For all \(o \in \2 + \V\), \(\beta(x)_\alpha(o)=\beta(y)_\alpha(o)\)
        \item For all \(Q \in X / {R}\) and for all \(\action{p} \in \Act\), we have \(\beta(x)_\alpha[\{p\}\times Q] = \beta(y)_\alpha[\{p\}\times Q]\).
    \end{enumerate}
    Moreover, \(\Phi_\beta(R)\) is also an equivalence relation.
\end{restatable}
\begin{proof}
    The proof of the first claim is identical to the proof of \cref{lem:larsen_skou} and hence we only focus on showing that that \(\Phi_\beta(R)\) is also an equivalence relation.

    As for the latter claim, the first half of the lemma tells us that we can use the simpler characterisation of \(\Phi_\beta\).
    Reflexivity holds immediately, as for all \(x \in X\) and all \(\alpha \in \At\)
    \begin{enumerate}
        \item For all \(o \in \2 + \V\), \(\beta(x)_\alpha(o) = \beta(x)_\alpha(o)\).
        \item For all \(\action{p} \in \Act\), \(Q \in X / R\), we have \(\beta(x)_\alpha[\{\action{p}\} \times Q] = \beta(x)_\alpha[\{\action{p}\} \times Q]\).
    \end{enumerate}
    and hence \((x,x) \in \Phi_\beta(R)\).

    To see symmetry, assume that \((x,y) \in \Phi_\beta(R)\) for some \(x,y \in X\). Then, because of symmetry of equality we have that all \(\alpha \in \At\)
    \begin{enumerate}
        \item For all \(o \in \2 + \V\), \(\beta(y)_\alpha(o) = \beta(x)_\alpha(o)\).
        \item For all \(\action{p} \in \Act\), \(Q \in X / R\), we have \(\beta(y)_\alpha[\{\action{p}\} \times Q] = \beta(x)_\alpha[\{\action{p}\} \times Q]\).
    \end{enumerate}
    and hence \((y,x) \in \Phi_\beta(R)\).

    Finally, to see transitivity, assume that \((x,y),(y,z) \in \Phi_\beta(R)\) for some arbitrary \(x,y,z \in X\). Then again, by transitivity of equality we have that for all \(\alpha \in \At \)
    \begin{enumerate}
        \item For all \(o \in \2 + \V\), \(\beta(x)_\alpha(o) = \beta(y)_\alpha(o) =  \beta(z)_\alpha(o)\).
        \item For all \(\action{p} \in \Act\), \(Q \in X / R\), we have \(\beta(x)_\alpha[\{\action{p}\} \times Q] = \beta(y)_\alpha[\{\action{p}\} \times Q]= \beta(z)_\alpha[\{\action{p}\} \times Q]\)
    \end{enumerate}
    and hence \((y,z) \in \Phi_\beta(R)\).
\end{proof}
Since \(\Phi_\beta\) preserves equivalence relations, it can also be viewed as an endofunction on the lattice of equivalence relations, ordered by inclusion (refinement).
Because \(\G\) preserves weak pullbacks, the greatest bisimulation is an equivalence relation~\cite[Corollary 5.6]{Rutten:2000:Universal}.
Thus, to characterise the greatest bisimulation, we can restrict ourselves to equivalence relations.

Before continuing, we recall some definitions.
An \emph{\(\omega\)-cochain} is a sequence \(\{x_i\}_{i\in\omega}\) of elements of \(X\), such that for all \(i \in \omega\), we have that \(x_i \sqsupseteq x_{i+1}\).
We call an endofunction \(f : X \to X\) \emph{\(\omega\)-cocontinuous} if it preserves meets of \(\omega\)-cochains, that is \(f\left(\sqcap_{i \in \omega} x_i\right)=\sqcap_{i \in \omega}f(x_i)\).

\begin{restatable}{proposition}{orderproperties}\label{prop:prop_order}
Let \((X, \sqsubseteq)\) be a complete lattice and \(f: X \to X\) an endofunction.
\begin{enumerate}
    \item (Kleene fixpoint theorem) If \(f\) is \(\omega\)-cocontinuous, then it possesses a greatest fixed point given by \(\gfp(f) = \sqcap_{i \in \omega} f^{(i)}(\top)\) where \(f^{(0)}=\id\) and \(f^{(n+1)}=f^{(n)} \circ f\).
    \item If \(f\) is monotone then for every \(\omega\)-cochain \(\{x_i\}_{i \in \omega}\) it holds that \(f\left(\sqcap_{i \in \omega} x_i\right)\sqsubseteq\sqcap_{i \in \omega}f(x_i)\).
\end{enumerate}
\end{restatable}
\begin{proof}
    We prove the second claim.
    Take a descending \(\omega\)-cochain \(\{x_i\}_{i \in \omega}\).
    Because of the monotonicity, applying \(f\) to the elements of that cochain yields another descending \(\omega\)-cochain \(\{f(x_i)\}_{i \in \omega}\).
    Because \(\sqcap_{i \in \omega} x_i \sqsubseteq x_i\) for all \(i \in \omega\), by monotonicity we have that \(f\left(\sqcap_{i \in \omega} x_i  \right) \sqsubseteq f(x_i)\), which makes \(f\left(\sqcap_{i \in \omega} x_i  \right)\) a lower bound of the cochain \(\{f(x_i)\}_{i \in \omega}\) and hence it is below its meet.
    Therefore, it holds that \(f \left(\sqcap_{i \in \omega} x_i\right) \sqsubseteq \sqcap_{i \in \omega} f(x_i)\).
\end{proof}

Given a \(\G\)-coalgebra \((X, \beta)\), we show the \(\omega\)-cocontinuity of \(\Phi_\beta\) on the lattice of equivalence relations on the set \(X\), by readapting the result of Baier~\cite[Lemma~3.7.5]{Baier:1998:Thesis}.

\begin{restatable}{lemma}{bisimfixpoint}\label{lem:bisimilarity_as_fixpoint}
    Let \((X, \beta)\) be \(\G\)-coalgebra.
    Then \(\Phi_{\beta}\) is \(\omega\)-cocontinous on the lattice of equivalence relations on \(X\).
\end{restatable}
\begin{proof}
    Because of \cref{prop:prop_order} and  \cref{lem:properties_of_bisim_functional}, to establish \(\omega\)-cocontinuity it suffices to show that for any descending \(\omega\)-cochain of equivalence relations \(\{R_i\}_{i \in \omega}\) ordered by inclusion,
    \[\bigcap_{i \in \omega} \Phi_{\beta}(R_i) \subseteq \Phi_{\beta,}\left( \bigcap_{i \in \omega} R_i \right)\]

    To this end, first note that for any $Q \in X/R$ and $i \in \omega$ there exists a unique $Q_i \in X/R_i$ such that $Q \subseteq Q_i$, and so $Q \subseteq \bigcap_{i \in \omega} Q_i$.
    Conversely, if $x \in Q_i$ for all $i \in \omega$, then take any $y \in Q$ (guaranteed to exist, since $Q$ is an equivalence class).
    Since $y \in Q_i$ for all $i$, we have that $x \mathrel{R_i} y$ for all $i \in \omega$, and so $x \mathrel{R} y$, meaning $x \in Q$.
    This tells us that $\bigcap_{i \in \omega} Q_i \subseteq Q$, and hence $Q = \bigcap_{i \in \omega} Q_i$.
    Furthermore, it follows that for $i \in \omega$ we have $Q_{i+1} \subseteq Q_i$.

    Now, fix $x \in X$; by monotonicity of probability, we find that
    \[\beta(x)_\alpha[\{\action{p}\} \times Q_0] \geq \beta(x)_\alpha[\{\action{p}\} \times Q_{1}] \geq \cdots \geq \beta(x)_\alpha[\{\action{p}\} \times Q]\]
    However, since $\beta(x)_\alpha$ has finite support, this sequence can take only finitely many different values.
    This means that it stabilises for some $i(x, Q) \in \omega$, meaning that for $k \geq i(x, Q)$ it holds that $\beta(x)_\alpha[\{\action{p}\} \times Q_i] = \beta(x)_\alpha[\{\action{p}\} \times Q_{i+1}] \geq \beta(x)_\alpha[\{\action{p}\} \times Q]$.

    Suppose now that $\beta(x)_\alpha[\{\action{p}\} \times Q] < \beta(x)_\alpha[\{\action{p}\} \times Q_i]$ for all $i$.
    Then in particular there exists some $x' \in X$ with $\beta(x)_\alpha(\action{p}, x') > 0$, such that $x' \in Q_i$ for all $i \in \omega$, but $x' \not\in Q$.
    This would contradict that $Q = \bigcap_{i \in \omega} Q_i$, and so we conclude that $\beta(x)_\alpha[\{\action{p}\} \times Q_i] = \beta(x)_\alpha[\{\action{p}\} \times Q]$ for some $i \in \omega$ --- in particular, this should hold for all $i \geq i(x, Q)$.

    Now, given $(x, y) \in \bigcap_{i \in \omega} \Phi_{\beta}(R_i)$, by \Cref{lem:preservation_of_equivalences} we have that $\beta(x)_\alpha(o) = \beta(y)_\alpha(o)$ for all $o \in 2 + \V$, and for all $i \in \omega$ and $S \in X/R_i$ we have $\beta(x)_\alpha[\{\action{p}\} \times S] = \beta(y)_\alpha[\{\action{p}\} \times S]$.
    Given $Q \in X/R$, if we choose $k = \max(i(x, Q), i(y, Q))$, then the above tells us that
    \[
        \beta(x)_\alpha[\{\action{p}\} \times Q]
            = \beta(x)_\alpha[\{\action{p}\} \times Q_k]
            = \beta(y)_\alpha[\{\action{p}\} \times Q_k]
            = \beta(y)_\alpha[\{\action{p}\} \times Q]
    \]
    Both conditions of \Cref{lem:preservation_of_equivalences} are now satisfied, and so $(x, y) \in \Phi_\beta(R) = \Phi_\beta(\bigcap_{i \in \omega} R_i)$.
\end{proof}

\refinementchain*
\begin{proof}
Follows from the \(\omega\)-cocontinuity of \(\Phi_\beta\) and \cref{prop:prop_order}
\end{proof}

Let \(A \subseteq \Exp\) and \(\expr{f}, \in \Exp\).
We define \(A / \expr{f} = A / \expr{f} = \{\expr{g} \in \Exp \mid \expr{g}\seq\expr{f} \in A\} \).

\begin{lemma}\label{lem:cutting_postfixes}
    If \(R \subseteq \Exp \times \Exp\) is a congruence relation with respect to \acro{ProbGKAT} operators and \((\expr{e}, \expr{f}) \in R\) then \(R(A / {\expr{e}}) \subseteq R(A) / \expr{f} \).
\end{lemma}
\begin{proof}
    If \(\expr{g} \in  R(A / \expr{e})\), then there exists some \(\expr{h} \in A / \expr{e}\) such that \((\expr{h}, \expr{g})\in R\).
    Since \(\expr{h} \in A / \expr{e}\), also \(\expr{h} \seq \expr{e} \in A\).
    Because \(R\) is a congruence relation, we have that \((\expr{h}\seq \expr{e}, \expr{g}\seq \expr{f})\in R\) and hence \(\expr{g}\seq \expr{f} \in R(A)\), which in turn implies that \(\expr{g} \in A / \expr{f}\).
\end{proof}

\begin{lemma}\label{lem:sequencing_cutting_postfixes}
    For all \(\alpha \in \At\), \(\expr{e}, \expr{f} \in \Exp\), \(\action{p} \in \Act\), \(A \subseteq \Exp\) we have that
    \[
    \partial(\expr{e}\seq\expr{f})_\alpha[\{\action{p}\}\times A] = \partial(\expr{e})_\alpha[\{\action{p}\}\times A / \expr{f}] + \partial(\expr{e})_\alpha(\accept)\partial(\expr{f})_\alpha[\{\action{p}\}\times A]
    \]
\end{lemma}
\begin{proof}
    We derive as follows.
    \begin{align*}
        \partial(\expr{e}\seq\expr{f})_\alpha[\{\action{p}\} \times A] &= \sum_{\expr{a} \in A} \partial(\expr{e}\seq\expr{f})_\alpha(\action{p}, \expr{a}) \\
        &= \sum_{\expr{a_0}\seq\expr{f} \in A}\partial(\expr{e})_\alpha(\action{p}, \expr{a_0}) + \sum_{\expr{a}\in A}\partial(\expr{e})_\alpha(\accept)\partial(\expr{f})(\action{p}, \expr{a}) \\
        &= \partial(\expr{e})_\alpha[\{\action{p}\} \times A / \expr{f}] + \partial(\expr{e})_\alpha(\accept)\partial(\expr{f})_\alpha[\{\action{p}\}\times A] \qedhere
    \end{align*}
\end{proof}

When considering the Brzozowski coalgebra \((\Exp, \partial)\), the operator \(\Phi_\partial\) has an another desirable property in addition to preserving equivalences.

\begin{restatable}{lemma}{congruencepreservation}\label{lem:preservation_of_congruences}
    Let \(R\) be a congruence relation with respect to \acro{ProbGKAT} operators. Then \(\Phi_\partial(R)\) is also a congruence relation.
\end{restatable}
\begin{proof}
    Let \(R \subseteq \Exp \times \Exp\) be a congruence relation, and let \(\expr{e}, \expr{f}, \expr{g}, \expr{h} \in \Exp\), such that \((\expr{e}, \expr{g}), (\expr{f}, \expr{h})\in \Phi_\partial(R)\). Let \(\test{b} \in \Bexp\) and \(\prob{r} \in [0,1]\).
    To show that \((\expr{e} \oplus_\prob{r} \expr{f}, \expr{g} \oplus_\prob{r} \expr{h}) \in \Phi_\partial(R)\), observe that for all \(\alpha \in \At\) and all \(o \in \2 + \V\) we have that
    \begin{align*}
            \partial(\expr{e} \oplus_\prob{r} \expr{f})_\alpha(o)&=\prob{r}\partial(\expr{e})_\alpha(o) + \prob{(1-r)}\partial(\expr{f})_\alpha(o) \\
            &=\prob{r}\partial(\expr{g})_\alpha(o) + \prob{(1-r)}\partial(\expr{h})_\alpha(o) \\
            &= \partial(\expr{g} \oplus_\prob{r} \expr{h})_\alpha(o)
    \end{align*}
    Similarly, for all \(\alpha \in \At\), \(\action{p} \in \Act\) and \(Q \in \Exp / R\) we have that
    {
    \small
    \begin{align*}
        \partial(\expr{e} \oplus_\prob{r} \expr{f})_\alpha[\{\action{p}\} \times Q] &= \sum_{\expr{q} \in Q} \partial(\expr{e} \oplus_\prob{r} \expr{f})_\alpha(\action{p},\expr{q}) \\
        &= \sum_{\expr{q} \in Q} \prob{r}\partial(\expr{e})_\alpha(\action{p},\expr{q}) + \prob{(1-r)}\partial(\expr{f})_\alpha(\action{p},
        \expr{q}) \\
        &= \prob{r} \partial(\expr{e})_\alpha[\{\action{p}\} \times Q] + \prob{(1-r)}\partial(\expr{f})_\alpha[\{\action{p}\}\times Q] \\
        &= \prob{r} \partial(\expr{g})_\alpha[\{\action{p}\} \times Q] + \prob{(1-r)}\partial(\expr{h})_\alpha[\{\action{p}\}\times Q] \\
        &= \partial(\expr{g} \oplus_\prob{r} \expr{h})_\alpha[\{\action{p}\}\times Q]
    \end{align*}
    }
    and therefore \((\expr{e} \oplus_\prob{r} \expr{f}, \expr{g} \oplus_\prob{r} \expr{h}) \in \Phi_\partial(R) \)

    To show that \((\expr{e} +_\test{b} \expr{f}, \expr{g} +_\test{b} \expr{h}) \in \Phi_\partial(R)\), consider the case when \(\alpha \bleq \test{b}\).
    For all \(o \in \2 + \V\) we then have that
    \[
        \partial(\expr{e} +_\test{b} \expr{f})_\alpha(o) = \partial(\expr{e})_\alpha(o)
        = \partial(\expr{g})_\alpha(o)
        = \partial(\expr{g} +_\test{b} \expr{h})_\alpha(o)
    \]
    Similarly, for all \(\action{p} \in \Act\) and \(Q \in \Exp / R\) we have that
    \[
        \partial(\expr{e} +_\test{b} \expr{f})_\alpha[\{\action{p}\} \times Q] = \partial(\expr{e})_\alpha[\{\action{p}\} \times Q]
        = \partial(\expr{g})_\alpha[\{\action{p}\} \times Q]
        = \partial(\expr{g} +_\test{b} \expr{h})_\alpha[\{\action{p}\} \times Q]
    \]
    The case when \(\alpha \bleq \bar{\test{b}}\) follows analogously.

    Now, we wish to show that \((\expr{e}\seq \expr{f}, \expr{g}\seq\expr{h} ) \in \Phi_\partial(R)\). Instead of using the simpler characterisation of \(\Phi\) when dealing with equivalence relations from \cref{lem:preservation_of_equivalences}, we will establish the conditions of \cref{def:bisimulation_functional}. For all \(\alpha \in \At\), we have that
    \[
    \partial(\expr{e}\seq\expr{f})_\alpha(\accept) = \partial(\expr{e})_\alpha(\accept)\partial(\expr{f})_\alpha(\accept)
    =\partial(\expr{g})_\alpha(\accept)\partial(\expr{h})_\alpha(\accept)
    =\partial(\expr{e}\seq\expr{f})_\alpha(\accept)
    \]
    For all \(\alpha \in \At\) and \(o \in \{\reject\} + \V\) it holds that
    \begin{align*}
        \partial(\expr{e}\seq\expr{f})_\alpha(o) &= \partial(\expr{e})_\alpha(o) + \partial(\expr{e})_\alpha(\accept)\partial(\expr{f})_\alpha(o) \\
        &=\partial(\expr{g})_\alpha(o) + \partial(\expr{g})_\alpha(\accept)\partial(\expr{h})_\alpha(o) \\
        &=\partial(\expr{g}\seq\expr{h})_\alpha(o)
    \end{align*}
    Now, we use \cref{lem:cutting_postfixes} and \cref{lem:sequencing_cutting_postfixes} to show that for all \(\alpha \in \Act\), \(\action{p} \in \Act\) and \(A \subseteq \Exp\)
    \begin{align*}
        \partial(\expr{e}\seq\expr{f})_\alpha[\{\action{p}\} \times A] &= \partial(\expr{e})_\alpha[\{\action{p}\}\times A / \expr{f}] + \partial(\expr{e})_\alpha(\accept)\partial(\expr{f})_\alpha[\{\action{p}\}\times A]\\
        &\leq \partial(\expr{g})_\alpha[\{\action{p}\}\times R(A / \expr{f})] + \partial(\expr{g})_\alpha(\accept)\partial(\expr{h})_\alpha[\{\action{p}\}\times R(A)]\\
        &\leq \partial(\expr{g})_\alpha[\{\action{p}\}\times R(A) / \expr{g}] + \partial(\expr{g})_\alpha(\accept)\partial(\expr{h})_\alpha[\{\action{p}\}\times R(A)]\\
        &=\partial(\expr{g}\seq\expr{h})_\alpha[\{\action{p}\}\times R(A)]
    \end{align*}
    The proof of \(\partial(\expr{g}\seq \expr{h})_\alpha[\{\action{p}\}\times B] \leq \partial(\expr{g}\seq \expr{h})_\alpha[\{\action{p}\}\times R^{-1}(B)]\) for arbitrary \(B \subseteq \Exp\) is analogous.

    Now, consider the case of showing that \((\expr{e}^{(\test{b})}, \expr{f}^{(\test{b})})\in\Phi_\partial(R)\). First assume that \(\alpha \in \bleq \bar{\test{b}}\).
    In such a case, we have that:
    \[
    \partial(\expr{e}^{(\test{b})})_\alpha(\accept) = \prob{1} =\partial(\expr{g}^{(\test{b})})_\alpha(\accept)
    \]
    Since all probability mass is assigned to the element \(\accept\), all remaining conditions of \(\Phi_\partial\) are immediately satisfied.
    Now, assume that \(\alpha \bleq \test{b}\).
    First, consider the subcase when \(\partial(\expr{e})_\alpha(\accept)=\prob{1}\). Since \((\expr{e}, \expr{g}) \in \Phi_\partial(R)\), it means that also \(\partial(\expr{g})_\alpha(\accept)=\prob{1}\).
    Hence,
    \[
    \partial(\expr{e}^{(\test{b})})_\alpha(\reject) = \prob{1} = \partial(\expr{g}^{(\test{b})})_\alpha(\reject)
    \]
    Therefore, in such a subcase we have that \((\expr{e}^{(\test{b})},\expr{g}^{(\test{b})})\in\Phi_\partial(R)\).
    For the remainder of this case, we can safely assume that \(\partial(\expr{e})_\alpha(\accept)<\prob{1}\) and \(\partial(\expr{g})_\alpha(\accept) < 1\). We have that
    \[
    \partial(\expr{e}^{(\test{b})})_\alpha(\accept) = \prob{0} =  \partial(\expr{g}^{(\test{b})})_\alpha(\accept)
    \]
    For all \(o \in \{\reject\} + \V\) it holds that
    \[
        \partial(\expr{e}^{(\test{b})})_\alpha(o) =\frac{\partial(\expr{e})_\alpha(o)}{\prob{1}-\partial(\expr{e})_\alpha(\accept)}
        =\frac{\partial(\expr{g})_\alpha(o)}{\prob{1}-\partial(\expr{g})_\alpha(\accept)}
        = \partial(\expr{g}^{(\test{b})})_\alpha(o)
    \]
    Again, we use \cref{lem:cutting_postfixes} to show that for arbitrary \(\action{p} \in \Act\) and \(A \subseteq \Exp\), we have
    \begin{align*}
        \partial(\expr{e}^{(\test{b})})_\alpha[\{\action{p}\}\times A] &= \sum_{\expr{a} \in A}  \partial(\expr{e}^{(\test{b})})_\alpha(\action{p}, \expr{a}) \\
        &= \sum_{\expr{a_o}\seq\expr{e}^{(\test{b})} \in A} \frac{\partial(\expr{e})_\alpha(\action{p}, \expr{a_0})}{\prob{1}-\partial(\expr{e})_\alpha(\accept)} \\
        &= \frac{\prob{1}}{\prob{1}-\partial(\expr{g})_\alpha(\accept)} \partial(\expr{e})_\alpha[\{\action{p}\}\times A / {\expr{e}^{(\test{b})}}]\\
        &\leq \frac{\prob{1}}{\prob{1}-\partial(\expr{g})_\alpha(\accept)} \partial(\expr{g})_\alpha[\{\action{p}\}\times R(A / {\expr{e}^{(\test{b})}})]\\
        &\leq \frac{\prob{1}}{\prob{1}-\partial(\expr{e})_\alpha(\accept)} \partial(\expr{g})_\alpha[\{\action{p}\}\times R(A) / {\expr{g}^{(\test{b})}}]\\
        &= \sum_{\expr{a_o}\seq\expr{g}^{(\test{b})} \in R(A)} \frac{\partial(\expr{g})_\alpha(\action{p}, \expr{a_0})}{\prob{1}-\partial(\expr{g})_\alpha(\accept)} \\
        &= \partial(\expr{g}^{(\test{b})})_\alpha[\{\action{p}\} \times R(A)]
    \end{align*}
    The proof that \(\partial(\expr{g}^{(\test{b})})[\{\action{p}\}\times B] \leq \partial(\expr{e}^{(\test{b})})_\alpha[\{\action{p}\} \times R^{-1}(B)]\) for arbitrary \(B \subseteq \Exp\) and \(\action{p} \in \Act\) is symmetric.
    We conclude that \((\expr{e}^{(\test{b})}, \expr{g}^{(\test{b})})\in \Phi_\partial(R)\).

    Finally, to show that \((\expr{e}^{[\prob{r}]}, \expr{g}^{[\prob{r}]})\in \Phi_\partial(R)\) we first consider the subcase when \(\prob{r} = \prob{1}\) and \(\partial(\expr{e})_\alpha(\accept)=\prob{1}\).
    Here, we also have that \(\partial(\expr{g})_\alpha(\accept) = \prob{1}\). Therefore, we have
    \[
    \partial(\expr{e}^{[\prob{r}]})_\alpha(\reject) = \prob{1} =  \partial(\expr{g}^{[\prob{r}]})_\alpha(\reject)
    \]
    which is enough to show that in that subcase indeed \((\expr{e}^{[\prob{r}]}, \expr{g}^{[\prob{r}]})\in \Phi_\partial(R)\).
    For the remainder, we can safely assume that \(\prob{r} \partial(\expr{e})_\alpha(\accept) \neq 1\) and hence so does \(\partial(\expr{g})_\alpha(\accept)\). We have that
    \[
    \partial(\expr{e}^{[\prob{r}]})_\alpha(\accept) = \prob{1-r} = \partial(\expr{g}^{[\expr{r}]})_\alpha(\accept)
    \]
    For all \(o \in \{\reject\} + \V\), we have that
    \[
        \partial(\expr{e}^{[\prob{r}]})_\alpha(o) = \frac{\prob{r}\partial(\expr{e})_\alpha(o)}{\prob{1}-\prob{r}\partial(\expr{e})_\alpha(\accept)}
        =\frac{\prob{r}\partial(\expr{g})_\alpha(o)}{\prob{1}-\prob{r}\partial(\expr{g})_\alpha(\accept)}
        = \partial(\expr{g}^{[\prob{r}]})_\alpha(o)
    \]
    Finally, we use \cref{lem:cutting_postfixes} to show that for arbitrary \(\action{p} \in \Act\) and \(A \subseteq \Exp\) we have
    \begin{align*}
        \partial(\expr{e}^{[\prob{r}]})[\{\action{p}\}\times A]&=\sum_{\expr{a} \in A} \partial(\expr{e}^{[\prob{r}]})_\alpha(\action{p}, \expr{a}) \\
        &=\frac{\prob{r}}{\prob{1}-\prob{r}\partial(\expr{e})_\alpha(\accept)}\sum_{\expr{a_0}\seq\expr{e}^{[\prob{r}]}} \partial(\expr{e})_\alpha(\action{p}, \expr{a_0})\\
        &= \frac{\prob{r}}{\prob{1}-\prob{r}\partial(\expr{e})_\alpha(\accept)} \partial(\expr{e}^{[\prob{r}]})_\alpha[\{\action{p}\}\times A / {\expr{e}^{[\prob{r}]}}] \\
        &\leq \frac{\prob{r}}{\prob{1}-\prob{r}\partial(\expr{e})_\alpha(\accept)} \partial(\expr{e}^{[\prob{r}]})_\alpha[\{\action{p}\}\times R(A / {\expr{e}^{[\prob{r}]}})] \\
        &\leq \frac{\prob{r}}{\prob{1}-\prob{r}\partial(\expr{g})_\alpha(\accept)} \partial(\expr{g}^{[\prob{r}]})_\alpha[\{\action{p}\}\times R(A) / {\expr{g}^{[\prob{r}]}}] \\
        &= \partial(\expr{g}^{[\prob{r}]})_\alpha[\action{p}\times R(A)]
    \end{align*}
    The case of showing that \(\partial(\expr{g}^{[\prob{r}]})_\alpha[\{\action{p}\} \times B] \leq \partial(\expr{e}^{[\prob{r}]})_\alpha[\{\action{p}\} \times R^{-1}(B)] \) for arbitrary \(B \subseteq \Exp\) and \(\action{p} \in \Act \) is symmetric and is omitted. Therefore \((\expr{e}^{[\prob{r}]}, \expr{g}^{[\prob{r}]})\in \Phi_\partial(R)\).
\end{proof}
\greatestcong*
\begin{proof}
Observe that the full relation on \(\Exp\) is a congruence. Since
\(\Phi_\partial\) preserves congruences and congruence relations are also preserved by intersections~\cite[Theorem~5.3]{Burris:1981:Algebra}, by the fixed point characterisation of the greatest fixed point of \(\Phi_\partial\), we have that the greatest bisimulation on \((\Exp, \partial)\) is a congruence.
\end{proof}

\section{Proofs from \texorpdfstring{\cref{sec: axiomatisation}}{Section~\ref{sec: axiomatisation}}}\label{apx:soundness}
\begin{lemma}\label{lem:no_termination_implies_success}
    Let \(\expr{e} \in \Exp\) and \(\alpha \in \At\). It holds that, \( \trmt{\expr{e}}_\alpha = \partial(\expr{e})_\alpha(\accept)\).
\end{lemma}
\begin{proof}
    By induction on the construction of \(\expr{e} \in \Exp\).
    The base cases hold immediately.
    First, consider the case when \(\expr{e} = \expr{f} +_\test{b} \expr{g}\) for some test \(\test{b} \in \Bexp\). Assume that \(\alpha \bleq \test{b}\); then
    \begin{align*}
        \trmt{\expr{f} +_\test{b} \expr{g}}_\alpha &= \trmt{\expr{f}} \tag{\(\alpha \bleq \test{b}\)} \\
        &= \partial(\expr{f})_\alpha(\accept) \tag{Induction hypothesis} \\
        &= \partial(\expr{f} +_\test{b} \expr{g})_\alpha(\accept) \tag{\(\alpha \bleq \test{b}\)}
    \end{align*}
    The case when \(\alpha \bleq \bar{\test{b}}\) is symmetric.

    Now, let \(\expr{e} = \expr{f} \oplus_{\prob{r}} \expr{g}\), where \(\prob{r} \in [0,1]\). We have that
    \begin{align*}
        \trmt{\expr{f} \oplus_{\prob{r}} \expr{g}}_\alpha &= \prob{r} \trmt{\expr{f}}_\alpha + \prob{(1-r)}\trmt{\expr{g}}_\alpha \tag{Def. of \(\mathsf{E}\)}\\
        &= \prob{r} \partial(\expr{f})_\alpha(\accept) + \prob{(1-r)}\partial(\expr{g})_\alpha(\accept) \tag{Induction hypothesis}\\
        &= \partial(\expr{f} \oplus_\prob{r} \expr{g})_\alpha(\accept) \tag{Def. of \(\partial\)}
    \end{align*}
    If \(\expr{e} = \expr{f} \seq \expr{g}\), then we have that
    \begin{align*}
        \trmt{\expr{f} \seq \expr{g}}_\alpha &= \trmt{\expr{f}}_\alpha \trmt{\expr{g}}_\alpha \tag{Def. of \(\mathsf{E}\)} \\
        &=\partial(\expr{f})_\alpha(\accept)\partial(\expr{g})_\alpha(\accept) \tag{Induction hypothesis} \\
        &=\partial(\expr{f} \seq \expr{g})_\alpha(\accept) \tag{Def. of \(\partial\)}
    \end{align*}
    As for the loops, let's consider first \(\expr{e} = \expr{f}^{(\test{b})}\) for some test \(\test{b} \in \Bexp\). First, consider the case when \(\alpha \bleq \bar{\test{b}}\).
    From the definition of \(\partial\) it immediately follows that
    \[
        \trmt{\expr{f}^{(\test{b})}}_\alpha = \prob{1} = \partial(\expr{e}^{(\test{b})})_\alpha(\accept)
    \]
    Again, in the case when \(\alpha \bleq \test{b}\), it follows from the definition of \(\partial\) that
    \[
        \trmt{\expr{f}^{(\test{b})}}_\alpha = \prob{0} = \partial(\expr{f}^{(\test{b})})_\alpha(\accept)
    \]
    Finally, let's consider probabilistic loop \(\expr{e} = \expr{f}^{[\prob{r}]}\) for some \(\prob{r} = \prob{1}\).
    First, consider the case when \(\trmt{f}_\alpha = \prob{1}\) and \(\prob{r} = \prob{1}\).
    By the induction hypothesis we have that \(\partial(\expr{f})_\alpha(\accept) = 1\).
    In such a case, by the definition of \(\partial\) we have that \(\partial(\expr{f}^{[\prob{r}]})_\alpha(\accept) = 0\) and therefore \[\trmt{\expr{f}^{[\prob{r}]}}_\alpha = \prob{0} = \partial(\expr{f}^{[\prob{r}]})_\alpha(\accept)\]
    For the remaining case we have that
    \begin{align*}
        \trmt{\expr{f}^{[\prob{r}]}}_\alpha &= \frac{1-r}{1-r\trmt{\expr{f}}_\alpha} \tag{Def. of \(\mathsf{E}\)} \\
        &= \frac{1-r}{1-r\partial(\expr{f})_\alpha(\accept)} \tag{Induction hypothesis} \\
        &= \partial(\expr{f}^{[\prob{r}]})_\alpha(\accept) \tag*{(Def. of \(\partial\)) \qedhere}
    \end{align*}
\end{proof}
\begin{lemma}\label{lem:swapping_congruent_ends}
    Let \(R \subseteq \Exp \times \Exp\) be a congruence with respect to \acro{ProbGKAT} operators such that \((\expr{e}, \expr{f}) \in R\), and let \(Q \in {\Exp} /{R}\) be an equivalence class of \(R\). It holds that
    \(
    Q/\expr{e} = Q / \expr{f}
    \)
\end{lemma}
\begin{proof}
Let \(\expr{g} \in Q\).
Since \(R\) is a congruence and \((\expr{e}, \expr{f}) \in R\), we have \((\expr{g}\seq\expr{e},\expr{g}\seq\expr{f}) \in R \).
Hence,
\[\expr{g} \in Q / \expr{e} \iff \expr{g} \seq \expr{e} \in Q \iff \expr{g} \seq \expr{f} \in Q \iff \expr{g} \in Q/ \expr{f} \qedhere\]
\end{proof}
\begin{lemma}\label{lem:associativity_of_cutting}
    Let \(\expr{e}, \expr{f} \in \Exp\) and let \(Q \in \Exp / {\equiv}\).
    Now \((Q/\expr{f})/\expr{e} = Q / \expr{e}\seq\expr{f}\)
\end{lemma}
\begin{proof}
    Let \(\expr{g} \in Q\); we derive as follows
    \[
        \expr{g} \in (Q/ \expr{f})/\expr{e}
            \iff \expr{g} \seq \expr{e} \in Q/\expr{f}
            \iff (\expr{g} \seq \expr{e}) \seq{\expr{f}} \in Q
            \iff \expr{g} \seq (\expr{e} \seq \expr{f}) \in Q
            \iff \expr{g} \in Q / {\expr{e} \seq \expr{f}}
    \]
    Here, the second to last step follows by associativity.
\end{proof}

\begin{lemma}\label{lem: equality_enough_for_bisimilarity}
    Let \((X, \beta)\) be an arbitrary \(\G\)-coalgebra, \(x, y \in X\) and \(R \subseteq X \times X\) an equivalence relation on \(X\).  If for all \(\alpha \in \At\), \(\beta(x)_\alpha = \beta(y)_\alpha\) then \((x,y) \in \Phi_\beta(R)\)
\end{lemma}
\begin{proof}
    Observe that if for all \(\alpha \in \At\) the distributions  \(\beta(x)_\alpha\) and \(\beta(y)_\alpha\) are equal, then conditions of \cref{lem:preservation_of_equivalences} immediately hold.
 \end{proof}
 \begin{lemma}\label{lem:asserting_derivative}
    Let \(\expr{e} \in \Exp\) and \(\test{b} \in \Bexp\). If \(\alpha \bleq \test{b}\), then \(\partial(\test{b} \seq \expr{e})_\alpha = \partial(\expr{e})_\alpha\).
\end{lemma}
\begin{proof}
   Recall that \(\partial(\test{b}\seq\expr{e})_\alpha = \partial(\test{b})_\alpha \semseq_\alpha \expr{e}\).  If \(\alpha \bleq \test{b}\), then \(\partial(\test{b})_\alpha(\accept) = \prob{1}\). Since \((- \semseq_\alpha \expr{e})\) is a convex extension of the map \(c_{\alpha, \expr{e}} : \2 + \V + \Act \times \Exp \to \D(\2 + \V + \Act \times \Exp) \), satisfying that \(c_{\alpha, \expr{e}}(\accept) = \partial(\expr{e})_\alpha\), we have that \[\partial(\test{b})_\alpha \semseq_{\alpha} \expr{e} = \prob{1}\partial(\expr{f})_\alpha = \partial(\expr{f})_\alpha \qedhere\]
\end{proof}
\soundness*
\begin{proof}
   Because of \cref{lem:properties_of_bisim_functional}, it suffices to show that \({\equiv} \subseteq \Phi_\partial(\equiv)\).
   We will interchangeably use different characterisations of \(\Phi_\partial(\equiv)\), including \cref{def:bisimulation_functional}, \cref{lem:preservation_of_equivalences} and \cref{lem: equality_enough_for_bisimilarity}.

    We proceed by induction on the length of derivation of \(\equiv\). The equational axioms are base cases, while the quasi-equational axioms are inductive steps.
    We need not verify the congruence rules of \(\equiv\), since by \cref{lem:preservation_of_congruences}, \(\Phi_\partial(\equiv)\) is also a congruence.

    Let \(\expr{e}, \expr{f}, \expr{g} \in \Exp\), \(\test{b}, \test{c} \in \Bexp\), \(\var{v} \in \V\) and \(\prob{r}, \prob{s} \in [0,1]\). The soundness of \emph{guarded choice axioms}, \emph{probabilistic choice axioms} and \emph{distributivity axiom} can be proved using \cref{lem: equality_enough_for_bisimilarity} by just showing that the each of the distributions given by the derivatives are equal.
    \begingroup
    \allowdisplaybreaks%
    \begin{description}
        \item[(\acro{G1})] If \(\alpha \bleq \test{b}\), then \(\partial(\expr{e} +_\test{b} \expr{e})_\alpha = \partial(\expr{e})_\alpha\).
        Similarly, for \(\alpha \bleq \bar{\test{b}}\), \(\partial(\expr{e} +_\test{b} \expr{e})_\alpha = \partial(\expr{e})_\alpha\).
        Therefore, for all \(\alpha \in \At\), \(\partial(\expr{e} +_\test{b} \expr{e})_\alpha = \partial(\expr{e})_\alpha\).

        \item[(\acro{G2})] If \(\alpha \bleq \test{b}\), then \(\partial(\expr{e} +_\test{b} \expr{f})_\alpha = \partial(\expr{e})_\alpha\). By \cref{lem:asserting_derivative} we have that \(\partial(\expr{e})_\alpha = \partial(\test{b}\seq\expr{e})_\alpha\), so \(\partial(\expr{e} +_\test{b} \expr{f})_\alpha = \partial(\test{b}\seq\expr{e} +_\test{b} \expr{f})_\alpha\). If \(\alpha \leq \bar{\test{b}}\), then \(\partial(\expr{e} +_\test{b} \expr{f})_\alpha = \partial(\expr{f})_\alpha = \partial(\test{b}\seq\expr{e} +_\test{b} \expr{f})_\alpha\)

        \item[(\acro{G3})] If \(\alpha \bleq \test{b}\), then \(\partial(\expr{e} +_{\test{b}} \expr{f})_\alpha = \partial(\expr{e})_\alpha = \partial(\expr{f} +_{\bar{\test{b}}} \expr{e})_\alpha\).
        Similarly, if \(\alpha \bleq \bar{\test{b}}\), then \(\partial(\expr{e} +_{\test{b}} \expr{f})_\alpha = \partial(\expr{f})_\alpha = \partial(\expr{f} +_{\bar{\test{b}}} \expr{e})_\alpha\).

        \item[(\acro{G4})] If \(\alpha \bleq \test{b}\) and \(\alpha \bleq \test{c}\), then \(\alpha \bleq \test{bc}\), and so \[\partial((\expr{e} +_\test{b} \expr{f}) +_\test{c} \expr{g})_\alpha = \partial(\expr{e})_\alpha = \partial(\expr{e} +_{\test{bc} }(\expr{f} +_\test{c} \expr{g}))_\alpha\]

        Now, let \(\alpha \bleq \bar{\test{b}}\) and \(\alpha \bleq \test{c}\). In such a case, \(\alpha \bleq \test{\overline{bc}}\) and \(\alpha \bleq \test{c}\), and so
        \[\partial((\expr{e} +_\test{b} \expr{f}) +_\test{c} \expr{g})_\alpha = \partial(\expr{f})_\alpha = \partial(\expr{e} +_{\test{bc} }(\expr{f} +_\test{c} \expr{g}))_\alpha\]

        Finally, when \(\alpha \bleq \bar{\test{c}}\), we have that $\alpha \leq \test{\overline{bc}}$, and so \[\partial((\expr{e} +_\test{b} \expr{f}) +_\test{c} \expr{g})_\alpha = \partial(\expr{g})_\alpha = \partial(\expr{e} +_{\test{bc} }(\expr{f} +_\test{c} \expr{g}))_\alpha\]

        \item[(\acro{P1})] For all \(\alpha \in \At\), \(\partial(\expr{e} \oplus_\prob{r} \expr{e})_\alpha = \prob{r}\partial(\expr{e})_\alpha + \prob{(1-r)}\partial(\expr{e})_\alpha = \partial(\expr{e})_\alpha\)

        \item[(\acro{P2})] For all \(\alpha \in \At\), \(\partial(\expr{e} \oplus_{\prob{1}} \expr{f})_\alpha = \prob{1}\partial(\expr{e})_\alpha + \prob{0} \partial(\expr{f})_\alpha = \partial(\expr{e})_\alpha\)
        \item[(\acro{P3})] For all \(\alpha \in \At\),
        \begin{align*}
            \partial(\expr{e} \oplus_\prob{r} \expr{f})_\alpha &= \prob{r}\partial(\expr{e})_\alpha + \prob{(1-r)}\partial(\expr{f})_\alpha\\
            &=\prob{(1-r)}\partial(\expr{f})_\alpha + \prob{(1-(1-r))} \partial(\expr{e})_\alpha\\
            &= \partial(\expr{f} \oplus_{\prob{1-r}} \expr{e})_\alpha
        \end{align*}
        \item[(\acro{P4})] Recall that by assumption \(\prob{rs} \neq \prob{1}\). For all \(\alpha \in \At\),
        \begin{align*}
            \partial((\expr{e} \oplus_\prob{r} \expr{f}) \oplus_{\prob{s}} \expr{g})_\alpha &= \prob{rs} \partial(\expr{e})_\alpha + \prob{(1-r)s}\partial(\expr{f})_\alpha + \prob{(1-s)}\partial(\expr{g})_\alpha\\
            &= \prob{rs} \partial(\expr{e})_\alpha + \prob{(1-rs)}\left(\prob{\frac{(1-r)s}{1-rs}} \partial(\expr{f})_\alpha + \prob{\frac{1-s}{1-rs}}\partial(\expr{g})_\alpha\right)
        \end{align*}
        Let \(\prob{t} = \prob{\frac{(1-r)s}{1-rs}}\). Hence, the previous is equal to
        \[
        \prob{rs} \partial(\expr{e})_\alpha + \prob{(1-rs)}\left(\prob{t} \partial(\expr{f})_\alpha + \prob{(1-t)}\partial(\expr{g})_\alpha\right) = \partial(\expr{e} \oplus_{\prob{rs}}  (\expr{f} \oplus_{\prob{t}} \expr{g}))_\alpha
        \]
        \item[(\acro{D})] For \(\alpha \bleq \test{b}\), we have \[\partial(\expr{e} \oplus_\prob{r} (\expr{f} +_\test{b} \expr{g}))_\alpha = \prob{r} \partial(\expr{e})_\alpha + \prob{(1-r)} \partial(\expr{f})_\alpha = \partial((\expr{e} +_\test{b} \expr{f}) \oplus_\prob{r} (\expr{e} +_\test{b} \expr{g}))_\alpha\] Similarly, when \(\alpha \bleq \bar{\test{b}}\), we have \[\partial(\expr{e} \oplus_\prob{r} (\expr{f} +_\test{b} \expr{g}))_\alpha = \prob{r} \partial(\expr{e})_\alpha + \prob{(1-r)} \partial(\expr{g})_\alpha = \partial((\expr{e} +_\test{b} \expr{f}) \oplus_\prob{r} (\expr{e} +_\test{b} \expr{g}))_\alpha\]
    \end{description}
    As for the \emph{sequencing axioms}, we will use characterisations of \(\Phi_\partial(\equiv)\) both from \cref{lem:preservation_of_equivalences} and \cref{lem: equality_enough_for_bisimilarity} depending on the case.
    \begin{description}
        \item[(S1)] Since for all \(\alpha \in \At\) we have that \(\alpha \bleq \test{\one}\), we can immediately use \cref{lem:asserting_derivative} to show that \(\partial(\test{\one}\seq\expr{e})_\alpha = \partial(\expr{e})_\alpha\). By \cref{lem: equality_enough_for_bisimilarity} we have that \((\test{\one}\seq\expr{e}, \expr{e}) \in \Phi_\partial(\equiv)\)
        \item[(S2)] We verify the conditions of \cref{lem:preservation_of_equivalences}. Let \(\alpha \in \At\).
        \begin{enumerate}
            \item  For all \(o \in \{\reject\} + \V\), we have that \[\partial(\expr{e}\seq\test{\one})_\alpha(o)= \partial(\expr{e})_\alpha(o) + \partial(\expr{e})_\alpha(\accept)\partial(\test{\one})_\alpha(o) = \partial(\expr{e})_\alpha(o)\]
            In the remaining case, we have that \[\partial(\expr{e}\seq\test{\one})_\alpha(\accept) = \partial(\expr{e})_\alpha(\accept) \partial(\test{\one})_\alpha(\accept) = \partial(\expr{e})_\alpha(\accept)\]
            \item Let \(\action{p} \in \Act\) and \(Q \in \Exp / \equiv \) be an equivalence class of \(\equiv\). By \cref{lem:sequencing_cutting_postfixes} we have that
            \begin{align*}
                \partial(\expr{e} \seq \test{\one})_\alpha[\{\action{p}\} \times Q] &= \partial(\expr{e})_\alpha[\{\action{p}\} \times Q / \test{\one}] + \partial(\expr{e})_\alpha(\accept)\partial(\test{\one})_\alpha[\{p\} \times Q]\\
                &=  \partial(\expr{e})_\alpha [\{\action{p}\} \times Q / \test{\one}]
            \end{align*}
        Observe that since \(Q\) is an equivalence class of the relation \(\equiv\), we have that \[\expr{h} \in Q / \test{\one} \iff \expr{h}\seq\test{\one} \in Q \iff \expr{h} \in Q \] Since \(Q = Q / \test{\one}\), we have that \[\partial(\expr{e}\seq\test{\one})_\alpha[\{\action{p}\} \times Q] = \partial(\expr{e})_\alpha[\{\action{p}\}\times Q]\]
        \end{enumerate}
    \item[(S3)] We verify the conditions of \cref{lem:preservation_of_equivalences}. Let \(\alpha \in \At\).
    \begin{enumerate}
        \item For all \(o \in \{\reject\} + \V\), we have that
        \begin{align*}
            \partial((\expr{e}\seq\expr{f})\seq\expr{g})_\alpha(o) &= \partial(\expr{e}\seq\expr{f})_\alpha(o) + \partial(\expr{e}\seq\expr{f})_\alpha(\accept)\partial(\expr{g})_\alpha(o)\\
            &= \partial(\expr{e})_\alpha(o)+\partial(\expr{e})_\alpha(\accept)\partial(\expr{f})_\alpha(o) + \partial(\expr{e}\seq\expr{f})_\alpha(\accept)\partial(\expr{g})_\alpha(o)\\
            &= \partial(\expr{e})_\alpha(o)+\partial(\expr{e})_\alpha(\accept)\partial(\expr{f})_\alpha(o) + \partial(\expr{e})_\alpha(\accept)\partial(\expr{f})_\alpha(\accept)\partial(\expr{g})_\alpha(o)\\ &= \partial(\expr{e}\seq(\expr{f} \seq \expr{g}))_\alpha(o)
        \end{align*}
        In the remaining case, we have that \[\partial((\expr{e}\seq\expr{f})\seq\expr{g})_\alpha(\accept) = \partial(\expr{e})_\alpha(\accept)\partial(\expr{f})_\alpha(\accept)\partial(\expr{g})_\alpha(\accept) = \partial(\expr{e} \seq (\expr{f} \seq \expr{g}))_\alpha(\accept)\]
        \item Let \(\action{p} \in \Act\) and \(Q \in \Exp / \equiv\). By \cref{lem:sequencing_cutting_postfixes} and \cref{lem:associativity_of_cutting},
        \begin{align*}
            \partial((\expr{e}\seq\expr{f})\seq\expr{g})_\alpha[\{\action{p}\}\times Q] &= \partial(\expr{e}\seq\expr{f})_\alpha[\{\action{p}\}\times Q / \expr{g}] + \partial(\expr{e} \seq \expr{f})_\alpha(\accept)\partial(\expr{g})_\alpha[\{\action{p}\} \times Q]\\
            &= \partial(\expr{e})_\alpha[\{\action{p}\}\times (Q / \expr{f})/\expr{g}] + \partial(\expr{e})_\alpha(\accept)\partial(\expr{f})_\alpha[\{\action{p}\}\times Q / \expr{g}] \\
            &\phantom{=} {}+ \partial(\expr{e})_\alpha(\accept)\partial(\expr{f})_\alpha(\accept)\partial(\expr{g})_\alpha[\{p\} \times Q]\\
            &= \partial(\expr{e})_\alpha[\{\action{p}\}\times Q / \expr{f}\seq\expr{g}] + \partial(\expr{e})_\alpha(\accept)\partial(\expr{f}\seq\expr{g})_\alpha[\{\action{p}\}\times Q]\\
            &= \partial(\expr{e}\seq(\expr{f}\seq\expr{g}))_\alpha[\{\action{p}\}\times Q]
        \end{align*}
    \end{enumerate}
    \item[(\acro{S4})] For all \(\alpha \in \At\), we have that \[\partial(\test{\zero}\seq\expr{e})_\alpha(\reject) = \partial(\test{\zero})_\alpha(\reject) + \partial(\test{\zero})_\alpha(\accept)\partial(\expr{e})_\alpha(\reject) = \partial(\test{\zero})_\alpha(\reject)=\prob{1}\]
    All the remaining elements of \(\2 + \V + \Act  \times \Exp\) have no probability mass assigned by those distributions. Hence, for all \(\alpha \in \At\), \(\partial(\test{\one}\seq\expr{e})_\alpha = \partial(\test{\zero})_\alpha\). Therefore, we can use \cref{lem: equality_enough_for_bisimilarity} to obtain the desired result.

    \item[(\acro{S5})] We use \cref{lem:preservation_of_equivalences}. Let \(\alpha \bleq \test{b}\).
    \begin{enumerate}
        \item Let \(o \in \{\reject\} + \V\). We have that
        \begin{align*}
            \partial((\expr{e} +_\test{b} \expr{f})\seq\expr{g})_\alpha(o) &= \partial(\expr{e} +_\test{b} \expr{f})_\alpha(o) + \partial(\expr{e} +_\test{b} \expr{f})_\alpha(\accept)\partial(\expr{g})_\alpha(o)\\
            &= \partial(\expr{e} )_\alpha(o) + \partial(\expr{e})_\alpha(\accept)\partial(\expr{g})_\alpha(o)\\
            &= \partial(\expr{e}\seq \expr{g} +_\test{b} \expr{f} \seq \expr{g})_\alpha(o)
        \end{align*}
        For the remaining case, we have
        \begin{align*}
            \partial((\expr{e} +_\test{b} \expr{f})\seq\expr{g})_\alpha(\accept) &= \partial(\expr{e} +_\test{b} \expr{f})_\alpha(\accept)\partial(\expr{g})_\alpha(\accept)\\
            &= \partial(\expr{e})_\alpha(\accept)\partial(\expr{g})_\alpha(\accept)\\
            &= \partial(\expr{e} \seq \expr{g} +_\test{b} \expr{f}\seq\expr{g})_\alpha(\accept)
        \end{align*}

        \item Let \(\action{p} \in \Act\) and \(Q \in \Exp / \equiv\). By \cref{lem:sequencing_cutting_postfixes} we have that
        \begin{align*}
                    \partial((\expr{e} +_\test{b} \expr{f})\seq\expr{g})_\alpha[\{\action{p}\}\times Q] &= \partial(\expr{e} +_\test{b} \expr{f})_\alpha[\{\action{p}\}\times Q/\expr{g}] + \partial(\expr{e} +_\test{b} \expr{f})_\alpha(\accept)\partial(g)_\alpha[\{\action{p}\}\times Q]\\
                    &= \partial(\expr{e})_\alpha[\{\action{p}\}\times Q/\expr{g}] + \partial(\expr{e})_\alpha(\accept)\partial(g)_\alpha[\{\action{p}\}\times Q]\\
                    &= \partial(\expr{e}\seq\expr{g})_\alpha[\{\action{p}\} \times Q]\\
                    &= \partial(\expr{e}\seq\expr{g} +_\test{b} \expr{f}\seq\expr{g})_\alpha[\{\action{p}\} \times Q]
        \end{align*}
    \end{enumerate}
    We omit the symmetric cases when \(\alpha \bleq \bar{\test{b}}\).

    \item[(\acro{S6})] Let \(\alpha \in \At\). As in the case of (\acro{S5}), we rely on \cref{lem:preservation_of_equivalences} and \cref{lem:sequencing_cutting_postfixes}.
    \begin{enumerate}
        \item For all \(o \in \{\reject\} + \V\), we have that
        \begin{align*}
            \partial((\expr{e} \oplus_\prob{r} \expr{f})\seq\expr{g})_\alpha(o) &= \partial(\expr{e} \oplus_\prob{r} \expr{f})_\alpha(o) + \partial(\expr{e} \oplus_\prob{r} \expr{f})(\accept)\partial(\expr{g})(o)\\
            &= \prob{r}\partial(\expr{e})_\alpha(o) + \prob{(1-r)}\partial(\expr{f})_\alpha(o) \\
            &\phantom{=} {} + \left(\prob{r}\partial(\expr{e})_\alpha(\accept) + \prob{(1-r)}\partial(\expr{f})_\alpha(\accept)\right)\partial(\expr{g})_\alpha(o)\\
            &= \prob{r} \partial(\expr{e} \seq g)_\alpha(o) + \prob{(1-r)}\partial(\expr{f}\seq\expr{g})_\alpha(o)\\
            &= \partial(\expr{e} \seq \expr{g} \oplus_\prob{r} \expr{f} \seq \expr{g})_\alpha(o)
        \end{align*}
        For the remaining case, we have that
        \begin{align*}
            \partial((\expr{e} \oplus_\prob{r} \expr{f})\seq\expr{g})_\alpha(o) &= \partial(\expr{e} \oplus_\prob{r} \expr{f})_\alpha(\accept)\partial(\expr{g})_\alpha(\accept)\\
            &= \prob{r}\partial(\expr{e})_\alpha(\accept)\partial(\expr{g})_\alpha(\accept) + \prob{(1-r)}\partial(\expr{f})_\alpha(\accept)\partial(\expr{g})_\alpha(\accept)\\
            &= \partial(\expr{e} \seq \expr{g} \oplus_\prob{r} \expr{f} \seq \expr{g})_\alpha(\accept)
        \end{align*}
        \item Let \(\action{p} \in \Act\) and \(Q \in \Exp /{\equiv}\). We have that
        \begin{align*}
            &\partial((\expr{e} \oplus_\prob{r} \expr{f})\seq\expr{g} )_\alpha[\{\action{p}\} \times Q] =\partial(\expr{e} \oplus_\prob{r} \expr{f})_\alpha[\{\action{p}\}\times {Q}/{\expr{g}} ] + \partial(\expr{e} \oplus_\prob{r} \expr{f})_\alpha(\accept)\partial(\expr{g})_\alpha[\{\action{p}\} \times Q]\\
            &=\prob{r}\left(\partial(\expr{e})_\alpha[\{\action{p}\}\times {Q}/\expr{g}] + \partial(\expr{e})_\alpha(\accept)\partial(\expr{g})_\alpha[\{\action{p}\}\times Q]\right)\\
            &\phantom{=}+\prob{(1-r)}\left(\partial(\expr{f})_\alpha[\{\action{p}\}\times {Q}/\expr{g}] + \partial(\expr{f})_\alpha(\accept)\partial(\expr{g})_\alpha[\{\action{p}\}\times Q]\right)\\
            &=\prob{r}\partial(\expr{e}\seq\expr{g})_\alpha[\{\action{p}\}\times Q] + \prob{(1-r)}\partial(\expr{f}\seq\expr{g})_\alpha[\{\action{p}\}\times Q] \\
            &=\partial(\expr{e}\seq\expr{g} \oplus_{\prob{r}}\expr{f}\seq\expr{g})_\alpha[\{\action{p}\}\times Q]
        \end{align*}
    \end{enumerate}

    \item[(\acro{S7})] Identical line of reasoning to the case of (\acro{S4})

    \item[(\acro{S8})] We use \cref{lem: equality_enough_for_bisimilarity} and show the equality of distributions for all \(\alpha \in \At\). First, consider the case when \(\alpha \bleq \test{b}\) and \(\alpha \bleq \test{c}\). Observe, that it is equivalent to \(\alpha \bleq \test{bc}\). In such a case, we have that \[\partial(\test{b}\seq\test{c})_\alpha(\accept) = \partial(\test{b})_\alpha(\accept)\partial(\test{c})_\alpha(\accept)=\prob{1}=\partial(\test{bc})_\alpha(\accept)\] Since both distributions assign all the probability mass to the same element, they are equal. Now, consider the case when \(\alpha \bleq \bar{\test{b}}\). In such a case, this implies that \(\alpha \bleq \bar{\test{bc}}\). We have that
    \[
    \partial(\test{b}\seq\test{c})_\alpha(\reject) = \partial(\test{b})_\alpha(\reject) + \partial(\test{b})_\alpha(\accept)\partial(\test{c})_\alpha(\reject) = \prob{1} = \partial(\test{bc})_\alpha(\reject)
    \]
    which is enough to show that both distributions are equal.
    Finally, consider the case when \(\alpha \bleq \test{b}\) and \(\alpha \bleq \bar{\test{c}}\), which also implies that \(\alpha \bleq \bar{\test{bc}}\). We again have that
    \[
    \partial(\test{b}\seq\test{c})_\alpha(\reject) = \partial(\test{b})_\alpha(\reject) + \partial(\test{b})_\alpha(\accept)\partial(\test{c})_\alpha(\reject) = \prob{1} = \partial(\test{bc})_\alpha(\reject)
    \]
    which proves that both distributions are equal.
\end{description}
Now, let's consider the \emph{loop axioms}. All cases will rely on \cref{lem:preservation_of_equivalences}.
\begin{description}
    \item[(\acro{L1})] First, consider the situation when \(\alpha \bleq \bar{\test{b}}\). Then, \[\partial\left(\expr{e}^{(\test{b})}\right)_\alpha(\accept) = \prob{1} = \partial(\test{\one})_\alpha(\accept) = \partial\left(\expr{e} \seq \expr{e}^{(\test{b})} +_\test{b} \test{\one}\right)_\alpha(\accept) \] Since both distributions assign all the probability mass to \(\accept\), all conditions of \cref{lem:preservation_of_equivalences} are immediately satisfied. For the rest of cases we will assume that \(\alpha \bleq \test{b}\). First consider the subcase when \(\partial(\expr{e})_\alpha(\accept)=\prob{1}\). Then,
    \begin{align*}
         \partial\left(\expr{e} \seq \expr{e}^{(\test{b})} +_\test{b} \test{\one} \right)_\alpha(\reject) &= \partial\left(\expr{e} \seq \expr{e}^{(\test{b})}\right)_\alpha(\reject)\\
         &= \partial(\expr{e})_\alpha(\reject) + \partial(\expr{e})_\alpha(\accept)\partial\left(\expr{e}^{(\test{b})}\right)(\reject)\\
         &= \partial(\expr{e})_\alpha(\accept)\partial\left(\expr{e}^{(\test{b})}\right)(\reject)\\
         &= \prob{1}\\
         &=\partial \left(\expr{e}^{(\test{b})} \right) _\alpha(\reject)
    \end{align*}
Since again both distributions assign all probability mass to \(\reject\), all conditions of \cref{lem:preservation_of_equivalences} are satisfied.
For the last case, we can now assume that \(\partial(\expr{e})_\alpha(\accept)\neq\prob{1}\). We verify all the conditions of \cref{lem:preservation_of_equivalences}.
\begin{enumerate}
    \item For \(o \in \{\reject\} + \V \), we have that
    \begin{align*}
        \partial\left(\expr{e}\seq\expr{e}^{(\test{b})} +_\test{b} \test{\one}\right)_\alpha(o) &=  \partial\left(\expr{e}\seq\expr{e}^{(\test{b})}\right)_\alpha(o) \\
        &=\partial(\expr{e})_\alpha(o) + \partial(\expr{e})_\alpha(\accept) \partial\left(\expr{e}^{(\test{b})}\right)_\alpha(o)\\
        &= \partial(\expr{e})_\alpha(o) + \partial(\expr{e})_\alpha(\accept)  \frac{\partial(\expr{e})_\alpha(o)}{\prob{1}-\partial(\expr{e})_\alpha(\accept)}\\
        &=\frac{\partial(\expr{e})_\alpha(o)(\prob{1} - \partial(\expr{e})_\alpha(\accept)) + \partial(\expr{e})_\alpha(\accept)\partial(\expr{e})_\alpha(o)}{1-\partial(\expr{e})_\alpha(\accept)} \\
        &= \frac{\partial(\expr{e})_\alpha(o)}{\prob{1}-\partial(\expr{e})_\alpha(\accept)} \\
        &= \partial\left(\expr{e}^{(\test{b})}\right)_\alpha(o)
    \end{align*}

    For the remaining case of acceptance, we have that
    \begin{align*}
        \partial\left(\expr{e}\seq\expr{e}^{(\test{b})} +_\test{b} \test{\one}\right)_\alpha(\accept) &= \partial\left(\expr{e}\seq\expr{e}^{(\test{b})}\right)_\alpha(\accept) \\
        &=\partial(\expr{e})_\alpha(\accept)\partial\left(\expr{e}^{(\test{b})}\right)_\alpha(\accept)\\
        &=\prob{0}\\
        &=\partial\left(\expr{e}^{(\test{b})}\right)_\alpha(\accept)
    \end{align*}
    \item Let \(\action{p} \in \Act\) and \(Q \in \Exp / \equiv\) be an equivalence class of \(\equiv\). We have that
    \begin{align*}
        \partial\left(\expr{e}\seq\expr{e}^{(\test{b})} +_\test{b} \test{\one}\right)_\alpha[\{\action{p}\}\times Q] &=  \partial\left(\expr{e}\seq\expr{e}^{(\test{b})}\right)_\alpha[\{\action{p}\}\times Q] \\
        &=\partial(\expr{e})_\alpha[\{\action{p}\}\times Q / \expr{e}^{(\test{b})}] + \partial(\expr{e})_\alpha(\accept) \partial\left(\expr{e}^{(\test{b})}\right)_\alpha[\{\action{p}\}\times Q]\\
        &=\partial(\expr{e})_\alpha[\{\action{p}\}\times Q / \expr{e}^{(\test{b})}] + \partial(\expr{e})_\alpha(\accept) \frac{\partial(\expr{e})_\alpha[\{\action{p}\}\times Q / {\expr{e}^{(\test{b})}}]}{\prob{1}-\partial(\expr{e})_\alpha(\accept)}\\
        &=\frac{\partial(\expr{e})_\alpha[\{\action{p}\}\times Q / {\expr{e}^{(\test{b})}}]}{\prob{1}-\partial(\expr{e})_\alpha(\accept)}\\
        &=\partial\left(\expr{e}^{(\test{b})}\right)_\alpha[\{\action{p}\}\times Q]
    \end{align*}
\end{enumerate}

\item[(L2)] First, consider the case when \(\partial(\expr{e})_\alpha(\accept) = \prob{1} \) and \(\prob{r} = \prob{1}\). Then, for all \(\alpha \in \At\)
\begin{align*}
    \partial\left(\expr{e} \seq \expr{e}^{[\prob{r}]} \oplus_{\prob{r}} \test{\one }\right)_\alpha(\reject) &= \partial \left(\expr{e} \seq \expr{e}^{[\prob{r}]}\right)_\alpha(\reject)\\
    &= \partial(\expr{e})_\alpha(\reject) + \partial(\expr{e})_\alpha(\accept)\partial\left(\expr{e}^{[\prob{r}]}\right)_\alpha(\reject)\\
    &=\partial\left(\expr{e}^{[\prob{r}]}\right)_\alpha(\reject)\\
    &=\prob{1}
\end{align*}
Since both distributions assign all probability mass to \(\reject\), all conditions of \cref{lem:preservation_of_equivalences} are immediately satisfied.
From now on, we can safely assume that \(\prob{r}\partial(\expr{e})_\alpha(\accept)\neq 1\). We now verify all the conditions of \cref{lem:preservation_of_equivalences} for all \(\alpha \in \At\).
\begin{enumerate}
    \item Let \(o \in \{\reject\} + \V\). Then,
    \begin{align*}
        \partial\left(\expr{e}\seq\expr{e}^{[\prob{r}]} \oplus_{\prob{r}} \test{\one}\right)_\alpha(o)  &= \prob{r}\partial \left(\expr{e}\seq\expr{e}^{[\prob{r}]}\right)_\alpha(o) + \prob{(1-r)}\partial(\test{\one})_\alpha(o) \\
        &= \prob{r}\partial(\expr{e})_\alpha(o) + \prob{r}\partial(\expr{e})_\alpha(\accept)\frac{\prob{r}\partial(\expr{e})_\alpha(o)}{\prob{1}-\prob{r}\partial(\expr{e})_\alpha(\accept)}\\
        &=\frac{\prob{r}\partial(\expr{e})_\alpha(o)}{\prob{1}-\prob{r}\partial(\expr{e})_\alpha(\accept)}\\
        &=\partial \left(\expr{e}^{[\prob{r}]}\right)_\alpha(o)
    \end{align*}
    Now, consider the remaining case of successful termination.
    \begin{align*}
        \partial\left(\expr{e}\seq\expr{e}^{[\prob{r}]} \oplus_{\prob{r}} \test{\one}\right)_\alpha(\accept) &= \prob{r} \partial\left(\expr{e} \seq \expr{e}^{[\prob{r}]}\right) + \prob{(1-r)}\\
        &= \prob{r}\partial(\expr{e})_\alpha(\accept)\frac{\prob{1-r}}{\prob{1} - \prob{r}\partial(\expr{e})_\alpha(\accept)} + \prob{(1-r)}\\
        &=\frac{\prob{1-r}}{\prob{1}-\prob{r}\partial(\expr{e})_\alpha(\accept)}\\
        &= \partial\left(\expr{e}^{[\prob{r}]}\right)_\alpha(\accept)
    \end{align*}
    \item Let \(\action{p} \in \Act\) and let \(Q \in \Exp /{\equiv}\) be an equivalence class of \(\equiv\).
    \begin{align*}
        \partial\left(\expr{e}\seq\expr{e}^{[\prob{r}]} \oplus_{\prob{r}} \test{\one}\right)_\alpha[\{\action{p}\}\times Q]  &= \prob{r}\partial \left(\expr{e}\seq\expr{e}^{[\prob{r}]}\right)_\alpha[\{\action{p}\}\times Q] + \prob{(1-r)}\partial(\test{\one})_\alpha[\{\action{p}\}\times Q] \\
        &= \prob{r}\partial(\expr{e})_\alpha[\{\action{p}\}\times Q/{\expr{e}^{[\prob{r}]}}] + \prob{r}\frac{ \prob{r}\partial(\expr{e})_\alpha[\{\action{p}\}\times Q/{\expr{e}^{[\prob{r}]}}]}{\prob{1}-\prob{r}\partial(\expr{e})_\alpha(\accept)}\\
        &=\frac{ \prob{r}\partial(\expr{e})_\alpha[\{\action{p}\}\times Q/{\expr{e}^{[\prob{r}]}}]}{\prob{1}-\prob{r}\partial(\expr{e})_\alpha(\accept)} \\
        &= \partial\left(\expr{e}^{[\prob{r}]}\right)_\alpha[\{\action{p}\}\times Q]
    \end{align*}
\end{enumerate}
\item[(\acro{L3})] First, consider the case when \(\alpha \bleq \bar{\test{b}}\). Then,
\[
\partial\left((\expr{e} +_\test{c} \test{\one})^{(\test{b})}\right)_\alpha(\accept) = \prob{1} = \partial \left( (\test{c} \seq \expr{e})^{(\test{b})}\right)_\alpha(\accept)
\]
Since both distributions assign all the probability mass to \(\accept\), all conditions of \cref{lem:preservation_of_equivalences} are immediately satisfied. Now, consider the case when \(\alpha \bleq \test{b}\) and \(\alpha \bleq \bar{\test{c}}\). Now
\[
\partial(\expr{e} +_\test{c} \test{\one})_\alpha(\accept) = \partial(\test{\one})_\alpha(\accept) = \prob{1}
\]
and therefore
\begin{align*}
    \partial\left((\test{c}\seq\expr{e})^{(\test{b})}\right)_\alpha(\reject) &= \frac{\partial(\test{c}\seq\expr{e})_\alpha(\reject)}{\prob{1}-\partial(\test{c}\seq\expr{e})_\alpha(\accept)}\\
    &=\frac{\partial(\test{c})_\alpha(\reject) + \partial(\test{c})_\alpha(\accept)\partial(\expr{e})_\alpha(\reject)}{\prob{1}-\partial(\test{c})_\alpha(\accept)\partial(\expr{e})_\alpha(\accept)}\\
    &=\prob{1}\\
    &=\partial\left((\expr{e} +_\test{c} \test{\one})^{(\test{b})}\right)_\alpha(\reject)
\end{align*}
Similarly to the case before, both distributions assign all the probability mass to \(\reject\), and hence all conditions of \cref{lem:preservation_of_equivalences} are immediately satisfied. Now, consider the case when \(\alpha \bleq \test{b}\) and \(\alpha \bleq \test{c}\). First, consider the subcase when also \(\partial(\expr{e})_\alpha(\accept)=\prob{1}\). Observe, that in such a case both \(\partial(\expr{e} +_\test{c} \test{\one})_\alpha(\accept) = \prob{1}\) and \(\partial(\test{c}\seq\expr{e})_\alpha(\accept)=\prob{1}\) and therefore we have that
\[
\partial\left((\expr{e} +_\test{c} \test{\one})^{(\test{b})} \right)_\alpha(\reject) = \prob{1} = \partial \left( (\test{c} \seq \expr{e})^{(\test{b})} \right)_\alpha(\reject)
\]
Both distributions assign all probability mass to the same element, which immediately satisfies the requirements of \cref{lem:preservation_of_equivalences}.
For the remainder, we can now safely assume that \(\partial(\expr{e})_\alpha(\accept)\neq\prob{1}\). This time, we have to verify all conditions of \cref{lem:preservation_of_equivalences}.
\begin{enumerate}
    \item For \(o \in \{\reject\} + \V\), we have that
    \begin{align*}
        \partial\left((\expr{e} +_\test{c} \test{\one})^{(\test{b})}\right)_\alpha(o) &= \frac{\partial(\expr{e} +_\test{c} \test{\one})_\alpha(o)}{\prob{1}-\partial(\expr{e} +_\test{c} \test{\one})_\alpha(\accept)} \\
        &= \frac{\partial(\expr{e})_\alpha(o)}{\prob{1}-\partial(\expr{e})_\alpha(\accept)} \\
        &= \frac{\partial(\test{c}\seq\expr{e})_\alpha(o)}{\prob{1}-\partial(\test{c}\seq\expr{e})_\alpha(\accept)} \tag{\cref{lem:asserting_derivative}} \\
        &= \partial\left((\test{c} \seq \expr{e})^{(\test{b})}\right)_\alpha(\reject)
    \end{align*}
    In the remaining case of immediate acceptance we have
    \[
    \partial\left((\expr{e} +_\test{c} \test{\one})^{(\test{b})}\right)_\alpha(\accept) = \prob{0} = \partial\left((\test{c}\seq\expr{e})^{\test{b}}\right)_\alpha(\accept)
    \]
    \item Let \(\action{p} \in \Act\) and \(Q \in \Exp / \equiv\). For all \(\alpha \in \At\) we have that
    \begin{align*}
        \partial\left((\expr{e} +_\test{c} \test{\one})^{(\test{b})}\right)_\alpha[\{\action{p}\}\times Q] &= \frac{\partial(\expr{e} +_\test{c} \test{\one})_\alpha[\{\action{p}\}\times Q/{(\expr{e} +_\test{c} \test{\one})^{(\test{b})}}]}{\prob{1}-\partial(\expr{e} +_\test{c} \test{\one})_\alpha(\accept)}\\
        &= \frac{\partial(\expr{e})_\alpha[\{\action{p}\}\times Q/{(\expr{e} +_\test{c} \test{\one})^{(\test{b})}}]}{\prob{1}-\partial(\expr{e})_\alpha(\accept)}\\
        &= \frac{\partial(\test{c}\seq\expr{e})_\alpha[\{\action{p}\}\times Q/{(\expr{e} +_\test{c} \test{\one})^{(\test{b})}}]}{\prob{1}-\partial(\test{c}\seq\expr{e})_\alpha(\accept)} \tag{\cref{lem:asserting_derivative}}\\
        &= \frac{\partial(\test{c}\seq\expr{e})_\alpha[\{\action{p}\}\times Q/{(\test{c}\seq\expr{e})^{(\test{b})}}]}{\prob{1}-\partial(\test{c}\seq\expr{e})_\alpha(\accept)} \tag{\cref{lem:swapping_congruent_ends}}\\
        &=\partial\left((\test{c}\seq\expr{e})^{(\test{b})}\right)_\alpha[\{\action{p}\}\times Q]
    \end{align*}
    \end{enumerate}
    \item[(L4)] Let \(\alpha \in \At\) be an arbitrary atom. First we will consider the subcase when \(\partial(\expr{e})_\alpha(\accept) = \prob{1}\). For all \(\alpha \in \At\) we have that
    \[
    \partial\left(\expr{e}^{(\test{\one})}\right)_\alpha(\reject) = \prob{1} = \partial\left(\expr{e}^{[\prob{1}]}\right)_\alpha(\reject)
    \]
    which is enough to verify the conditions of \cref{lem:preservation_of_equivalences}.For the remainder of this case, we can safely assume that \(\partial(\expr{e})_\alpha(\accept)\neq\prob{1}\). In such as case, we need to verify both conditions of \cref{lem:preservation_of_equivalences} for arbitrary \(\alpha \in \At\).
    \begin{enumerate}
        \item Let \(o \in \{\reject\} + \V\). We have that
        \[
            \partial\left(\expr{e}^{(\test{\one})} \right)_\alpha(o)
            = \frac{\partial(\expr{e})_\alpha(o)}{\prob{1} - \partial(\expr{e})_\alpha(\accept)}
            = \frac{\prob{1}\partial(\expr{e})_\alpha(o)}{\prob{1} - \prob{1}\partial(\expr{e})_\alpha(\accept)}
            = \partial\left(\expr{e}^{[\prob{1}]}\right)_\alpha(o)
        \]
    As for the case of immediate acceptance, we have that
    \[
    \partial\left(\expr{e}^{(\test{\one})}\right) = \prob{0} = \frac{\prob{0}}{\prob{1}-\prob{1}\partial(\expr{e})_\alpha(\accept)} = \partial \left(\expr{e}^{[\prob{1}]}\right)_\alpha(\accept)
    \]
    \item Let \(\action{p} \in \Act\) and \(Q \in \Exp / \equiv\). We have that
    \begin{align*}
        \partial\left(\expr{e}^{(\test{\one})}\right)_\alpha[\{\action{p}\} \times Q] &= \frac{\partial(\expr{e})_\alpha[\{\action{p}\}\times Q/ {\expr{e}^{(\test{\one})}}]}{\prob{1} - \partial(\expr{e})_\alpha(\accept)} \\
        &= \frac{\partial(\expr{e})_\alpha[\{\action{p}\}\times Q/ {\expr{e}^{[\prob{1}]}}]}{\prob{1} - \partial(\expr{e})_\alpha(\accept)} \tag{\cref{lem:swapping_congruent_ends}}\\
        &= \partial\left(\expr{e}^{[\prob{1}]}\right)_\alpha[\{\action{p}\}\times Q]
    \end{align*}
    \end{enumerate}
    \item[(\acro{F1})] Assume the premises of the rule hold. We have that \((\expr{g}, \expr{e}\seq\expr{g} +_\test{b} \expr{f}) \in \Phi_\partial(\equiv)\) and for all \(\alpha \in \At\) we have that \(\trmt{\expr{g}}_\alpha = \prob{0}\). By \cref{lem:no_termination_implies_success} we have that \(\partial(\expr{g})_\alpha(\accept) = \prob{0}\) for all \(\alpha \in \At\).

     First, consider the case when \(\alpha \bleq \bar{\test{b}}\). We verify the conditions of \cref{lem:preservation_of_equivalences}.
    \begin{enumerate}
        \item Let \(o \in \{\reject\} + \V \). Consider the following
        \begin{align*}
            \partial(\expr{g})_\alpha(o)&=
            \partial(\expr{e}\seq\expr{g} +_\test{b} \expr{f})_\alpha(o) \tag{Induction hypothesis}\\
            &=\partial(\expr{f})_\alpha(o)\\
            &= \partial\left(\expr{e}^{(\test{b})}\right)_\alpha(o) + \partial\left(\expr{e}^{(\test{b})}\right)_\alpha(\accept)\partial(\expr{f})_\alpha(o) \\
            &=\partial\left(\expr{e}^{(\test{b})}\seq\expr{f}\right)_\alpha(o)
        \end{align*}
        For the case of immediate termination consider
        \begin{align*}
            \partial(\expr{g})_\alpha(\accept) &= \partial(\expr{e} \seq \expr{g} +_\test{b} \expr{f})_\alpha(\accept) \tag{Induction hypothesis}\\
            &=\partial(\expr{f})_\alpha(\accept) \\
            &=\partial \left( \expr{e}^{(\test{b})}\right)_\alpha(\accept)\partial(\expr{f})_\alpha(\accept)\\
            &=\partial\left( \expr{e}^{(\test{b})} \seq \expr{f}\right)_\alpha(\accept)
        \end{align*}
        \item Let \(\action{p} \in \Act\) and \(Q \in \Exp / \equiv \). We have that
        \begin{align*}
            \partial(\expr{g})_\alpha[\{\action{p}\} \times Q] &= \partial(\expr{e}\seq\expr{g} +_\test{b} \expr{f})_\alpha[\{\action{p}\}\times Q]\tag{Induction hypothesis}\\
            &=\partial(\expr{f})_\alpha[\{\action{p}\}\times Q]\\
            &=\partial\left(\expr{e}^{(\test{b})}\right)[\{\action{p}\} \times Q / \expr{f}] + \partial\left(\expr{e}^{(\test{b})}\right)(\accept)\partial(\expr{f})_\alpha[\{\action{p}\}\times Q] \\
            &=\partial\left( \expr{e}^{(\test{b})}\seq \expr{f}\right)_\alpha[\{\action{p}\}\times Q] \tag{\cref{lem:sequencing_cutting_postfixes}}
        \end{align*}
    \end{enumerate}
    Now, we consider the case when \(\alpha \bleq \bar{\test{b}}\).
    \begin{enumerate}
        \item For \(o \in \{\reject\} + \V\), we have that
        \begin{align*}
            \partial(\expr{g})_\alpha(o) &= \partial(\expr{e}\seq\expr{g} +_\test{b} \expr{f})_\alpha(o) \tag{Induction hypothesis}\\
            &= \partial(\expr{e}\seq\expr{g})_\alpha(o)\\
            &=\partial(\expr{e})_\alpha(o) + \partial(\expr{e})_\alpha(\accept)\partial(\expr{g})_\alpha(o)\\
            &=\partial(\expr{e})_\alpha(o) \\
            &=\frac{\partial(\expr{e})_\alpha(o)}{\prob{1}-\partial(\expr{e})_\alpha(\accept)} \\
            &=\partial \left( \expr{e}^{(\test{b})}\right)_\alpha(o)\\
            &=\partial \left( \expr{e}^{(\test{b})}\right)_\alpha(o) + \partial\left(\expr{e}^{(\test{b})}\right)_\alpha(\accept)\partial(\expr{f})_\alpha(o)\\
            &=\partial\left(\expr{e}^{(\test{b})}\seq\expr{f} \right)_\alpha(o)
        \end{align*}
        For the remaining case of outputting \(\accept\), consider the following
        \begin{align*}
            \partial(\expr{g})_\alpha(\accept) &= \partial(\expr{e}\seq\expr{g} +_\test{b} \expr{f})_\alpha(\accept) \tag{Induction hypothesis} \\
            &=\partial(\expr{e} \seq \expr{g})_\alpha(\accept) \\
            &=\partial(\expr{e})_\alpha(\accept)\partial(\expr{g})_\alpha(\accept) \\
            &=\prob{0}\\
            &=\partial \left( \expr{e}^{(\test{b})}\right)_\alpha(\accept)\\
            &=\partial \left( \expr{e}^{(\test{b})}\right)_\alpha(\accept)\partial(\expr{f})_\alpha(\accept) \\
            &=\partial\left(\expr{e}^{(\test{b})}\seq\expr{f}\right)_\alpha(\accept)
        \end{align*}
        \item For \(\action{p} \in \Act \) and \(Q \in \Exp / \equiv \) we have that
        \begin{align*}
            \partial(\expr{g})_\alpha[\{\action{p}\}\times Q] &= \partial(\expr{e}\seq\expr{g} +_\test{b} \expr{f})_\alpha[\{\action{p}\}\times Q] \tag{Induction hypothesis}\\
            &=\partial(\expr{e}\seq\expr{g})_\alpha[\{\action{p}\}\times Q]\\
            &=\partial(\expr{e})_\alpha[\{\action{p}\}\times Q / \expr{g}] + \partial(\expr{e})_\alpha(\accept)\partial(\expr{g})_\alpha[\{\action{p}\}\times Q] \\
            &=\partial(\expr{e})_\alpha[\{\action{p}\}\times Q / \expr{g}] \\
            &= \partial(\expr{e})_\alpha[\{\action{p}\} \times Q / \expr{e}^{(\test{b})}\seq\expr{f}] \tag{\cref{lem:swapping_congruent_ends}}\\
            &= \partial(\expr{e})_\alpha[\{\action{p}\} \times (Q /\expr{f})/ \expr{e}^{(\test{b})}] \tag{\cref{lem:associativity_of_cutting}}\\
            &= \frac{\partial(\expr{e})_\alpha[\{\action{p}\} \times (Q /\expr{f})/ \expr{e}^{(\test{b})}]}{\prob{1}-\partial(\expr{e})_\alpha(\accept)} \\
            &=\partial \left( \expr{e}^{(\test{b})}\right)_\alpha[\{\action{p}\} \times Q / \expr{f}] \\
            &=\partial \left( \expr{e}^{(\test{b})}\right)_\alpha[\{\action{p}\} \times Q / \expr{f}] + \partial\left(\expr{e}^{(\test{b})}\right)_\alpha(\accept)\partial(\expr{f})_\alpha[\{\action{p}\}\times Q] \\
            &=\partial\left(\expr{e}^{(\test{b})}\seq\expr{f} \right)_\alpha[\{\action{p}\} \times Q]
        \end{align*}
    \end{enumerate}
    \item[(\acro{F2})] Assume the premises of the rule hold. We have that \((\expr{g}, \expr{e}\seq\expr{g} \oplus_\prob{r} \expr{f}) \in \Phi_\partial(\equiv)\) and for all \(\alpha \in \At\) we have that \(\trmt{\expr{g}}_\alpha = \prob{0}\). By \cref{lem:no_termination_implies_success} we have that \(\partial(\expr{g})_\alpha(\accept) = \prob{0}\) for all \(\alpha \in \At\). We verify the conditions of \cref{lem:preservation_of_equivalences}.
    \begin{enumerate}
        \item For \(o \in \{\reject\} + \V\) it holds that
        \begin{align*}
            \partial(\expr{g})_\alpha(o) &= \partial(\expr{e}\seq\expr{g} \oplus_{\prob{r}} \expr{g})_\alpha(o) \tag{Induction hypothesis}\\
            &=\prob{r}\partial(\expr{e}\seq\expr{g})_\alpha(o) + \prob{(1-r)} \partial(\expr{f})_\alpha(o) \\
            &=\prob{r}\partial(\expr{e})_\alpha(o) + \prob{r}\partial(\expr{e})_\alpha(\accept)\partial(\expr{g})_\alpha(o) + \prob{(1-r)}\partial(\expr{f})_\alpha(o)\\
            &=\frac{\prob{r}\partial(\expr{e})_\alpha(o)}{\prob{1}-\prob{r}\partial(\expr{e})_\alpha(\accept)} + \frac{\prob{1-r}}{\prob{1}-\prob{r}\partial(\expr{e})_\alpha(\accept)}\partial(\expr{f})_\alpha(o)\\
            &=\partial\left(\expr{e}^{[\prob{r}]} \right)_\alpha(o) + \partial\left(\expr{e}^{[\prob{r}]} \right)_\alpha(\accept)\partial(\expr{f})_\alpha(o)\\
            &=\partial \left(\expr{e}^{[\prob{r}]}\seq\expr{f} \right)_\alpha(o)
        \end{align*}
        Now, consider the case of outputting \(\accept\). We have that
        \begin{align*}
            \partial(\expr{g})_\alpha(\accept) &= \partial(\expr{e}\seq\expr{g} \oplus_{\prob{r}} \expr{f})_\alpha(\accept) \tag{Induction hypothesis}\\
            &=\prob{r}\partial(\expr{e}\seq\expr{g})_\alpha(\accept) + \prob{(1-r)}\partial(\expr{f})_\alpha(\accept) \\
            &=\prob{r}\partial(\expr{e})_\alpha(\accept)\partial(\expr{g})_\alpha(\accept) + \prob{(1-r)}\partial(\expr{f})_\alpha(\accept)\\
            &=\prob{(1-r)}\partial(\expr{f})_\alpha(\accept)\\
            &=\frac{\prob{1-r}}{\prob{1}-\prob{r}\partial(\expr{e})_\alpha(\accept)}\partial(\expr{f})_\alpha(\accept)\\
            &=\partial\left(\expr{e}^{[\prob{r}]}\right)_\alpha(\accept) \partial(\expr{f})_\alpha(\accept)\\
            &=\partial\left(\expr{e}^{[\prob{r}]}\seq\expr{f} \right)_\alpha(\accept)
        \end{align*}
        \item Let \(\action{p} \in \Act\) and let \(Q \in \Exp / \equiv\). Consider the following
        \begin{align*}
            \partial(\expr{g})_\alpha[\{\action{p}\}] &= \partial(\expr{e}\seq\expr{g} \oplus_\prob{r} \expr{f})_\alpha[\{\action{p}\} \times Q] \tag{Induction hypothesis}\\
            &=\prob{r}\partial(\expr{e}\seq\expr{g})_\alpha[\{\action{p}\} \times Q] + \prob{(1-r)} \partial(\expr{f})_\alpha[\{\action{p}\} \times Q] \\
            &=\prob{r}\partial(\expr{e})_\alpha[\{\action{p}\} \times Q / \expr{g}] + \prob{r} \partial(\expr{e})_\alpha(\accept)\partial(\expr{g})_\alpha[\{\action{p}\} \times Q] \\
            &\phantom{=} {} + \prob{(1-r)}\partial(\expr{f})_\alpha[\{\action{p}\}\times Q] \\
            &=\prob{r}\partial(\expr{e})_\alpha[\{\action{p}\} \times Q / \expr{g}] + \prob{(1-r)}\partial(\expr{f})_\alpha[\{\action{p}\}\times Q]\\
            &=\prob{r}\partial(\expr{e})_\alpha[\{\action{p}\} \times Q/ \expr{e}^{[\prob{r}]}\seq\expr{f}] + \prob{(1-r)}\partial(\expr{f})_\alpha[\{\action{p}\}\times Q]\tag{\cref{lem:swapping_congruent_ends}} \\
            &=\prob{r}\partial(\expr{e})_\alpha[\{\action{p}\} \times (Q/\expr{f}) / \expr{e}^{[\prob{r}]}] + \prob{(1-r)}\partial(\expr{f})_\alpha[\{\action{p}\}\times Q]\tag{\cref{lem:associativity_of_cutting}} \\
            &=\frac{\prob{r}\partial(\expr{e})_\alpha[\{\action{p}\} \times (Q/\expr{f}) / \expr{e}^{[\prob{r}]}]}{\prob{1}-\prob{r}\partial(\expr{e})_\alpha(\accept)} + \frac{\prob{(1-r)}}{\prob{1}-\prob{r}\partial(\expr{e})_\alpha(\accept)}\partial(\expr{f})_\alpha[\{\action{p}\}\times Q] \\
            &=\partial\left(\expr{e}^{[\prob{r}]} \right)_\alpha[\{\action{p}\}\times Q / \expr{f}] + \partial \left(\expr{e}^{[\prob{r}]}\right)_\alpha(\accept)\partial(\expr{f})_\alpha[\{\action{p}\}\times Q]\\
            &=\partial\left(\expr{e}^{[\prob{r}]}\seq \expr{f} \right)_\alpha[\{\action{p}\}\times Q]
        \end{align*}
    \end{enumerate}
        \item[(\acro{L5})]
        Assume that premises hold and are satisfied by bisimilarity. In particular, we have that \((\expr{e}, (\expr{f} \oplus_\prob{r} \test{\one}) +_\test{c} \expr{g}) \in \Phi_\partial(\equiv)\). By assumption \(\prob{r} > \prob{0}\). First, consider the case when \(\alpha \bleq \bar{\test{c}}\). We have
        \begin{align*}
            \partial\left(\test{c}\seq\expr{e}^{(\test{b})}\right)_\alpha(\reject) &= \partial(\test{c})_\alpha(\reject) + \partial(\test{c})_\alpha(\accept)\partial\left(\expr{e}^{(\test{b})}\right)_\alpha(\reject) \\
            &=\prob{1} \\
            &=\partial\left(\test{c}\seq(\expr{f} \seq \expr{e}^{(\test{b})} +_{\test{b}} \test{\one}) \right)_\alpha(\reject)
        \end{align*}
        which is enough to verify all the conditions of \cref{lem:preservation_of_equivalences} since all probability mass is assigned by both distributions to the same element. Now, consider the case when \(\alpha \bleq \test{c}\) and \(\alpha \bleq \bar{\test{b}}\). We have that
        \begin{align*}
            \partial\left(\test{c}\seq\expr{e}^{(\test{b})}\right)_\alpha(\accept)&=\partial\left(\expr{e}^{(\test{b})}\right)_\alpha(\accept) \tag{\cref{lem:asserting_derivative}} \\
            &=\prob{1} \\
            &=\partial\left(\expr{f}\seq\expr{e}^{(\test{b})} +_\test{b} \test{\one} \right)_\alpha(\accept)\\
            &=\partial\left(\test{c}\seq(\expr{f}\seq\expr{e}^{(\test{b})} +_\test{b} \test{\one})\right)_\alpha(\accept) \tag{\cref{lem:asserting_derivative}}
        \end{align*}
        which is again enough to verify the conditions of \cref{lem:preservation_of_equivalences}.
        For the rest of cases assume that \(\alpha \bleq \test{c}\) and \(\alpha \bleq \test{b}\).
        First, we consider the subcase when \(\partial(\expr{f})_\alpha(\accept)=\prob{1}\). Because of the induction hypothesis, it means that \(\partial(\expr{e})_\alpha(\accept)=\prob{1}\) and hence we have that
        \begin{align*}
            \partial\left(\test{c}\seq\expr{e}^{(\test{b})}\right)_\alpha(\reject) &= \partial \left(\expr{e}^{(\test{b})}\right)_\alpha(\reject)\\
            &=\prob{1} \\
            &=\partial(\expr{f})_\alpha(\reject) + \partial(\expr{f})_\alpha(\accept) \partial\left(\expr{e}^{(\test{b})}\right)_\alpha(\reject) \\
            &= \partial\left(\expr{f}\seq\expr{e}^{(\test{b})} \right)_\alpha(\reject) \\
            &=\partial\left(\expr{f}\seq\expr{e}^{(\test{b})} +_\test{b} \test{\one} \right)_\alpha(\reject) \\
            &=\partial\left(\test{c}\seq(\expr{f}\seq\expr{e}^{(\test{b})} +_\test{b} \test{\one})\right)_\alpha(\reject)
        \end{align*}
        This is enough to satisfy the conditions of \cref{lem:preservation_of_equivalences}. Finally, for the remainder of this case we will assume that \(\partial(\expr{f})_\alpha(\accept)\neq\prob{1}\). In this case we need to verify all conditions of \cref{lem:preservation_of_equivalences}.
        \begin{enumerate}
            \item Let \(o \in \{\reject\} + \V\). Observe that
            \begin{align*}
                \partial\left(\test{c}\seq\expr{e}^{(\test{b})}\right)_\alpha(o) &= \partial \left(\expr{e}^{(\test{b})} \right)_\alpha(o) \tag{\cref{lem:asserting_derivative}} \\
                &=\frac{\partial(\expr{e})_\alpha(o)}{\prob{1}-\partial(\expr{e})_\alpha(\accept)} \\
                &=\frac{\partial((\expr{f} \oplus_\prob{r} \test{\one}) +_\test{c} \expr{g})_\alpha(o)}{\prob{1}-\partial((\expr{f} \oplus_\prob{r} \test{\one}) +_\test{c} \expr{g})_\alpha(\accept)} \tag{Induction hypothesis}\\
                &=\frac{\partial(\expr{f} \oplus_\prob{r} \test{\one} )_\alpha(o)}{\prob{1}-\partial(\expr{f} \oplus_\prob{r} \test{\one})_\alpha(\accept)} \\
                &=\frac{\prob{r}\partial(\expr{f})_\alpha(o)}{\prob{1}-\prob{r}\partial(\expr{f})_\alpha(\accept) - \prob{(1-r)}}\\
                &=\frac{\partial(\expr{f})_\alpha(o)}{\prob{1}-\partial(\expr{f})_\alpha(\accept)}\\
                &=\partial(\expr{f})_\alpha(o)\left(\frac{\prob{1}-\partial(\expr{f})_\alpha(\accept) + \partial(\expr{f})_\alpha(\accept)}{\prob{1}-\partial(\expr{f})_\alpha(\accept)}\right)\\
                &=\partial(\expr{f})_\alpha(o)\left(\prob{1} + \frac{\partial(\expr{f})_\alpha(\accept)}{\prob{1}-\partial(\expr{f})_\alpha(\accept)}\right)\\
                &=\partial(\expr{f})_\alpha(o) + \partial(\expr{f})_\alpha(\accept)\frac{\partial(\expr{f})_\alpha(o)}{\prob{1}-\partial(\expr{f})_\alpha(o)} \\
                &=\partial(\expr{f})_\alpha(o) + \partial(\expr{f})_\alpha(\accept)\frac{\prob{r}\partial(\expr{f})_\alpha(o)}{\prob{1}-\prob{r}\partial(\expr{f})_\alpha(\accept)-\prob{(1-r)}}\\
                &=\partial(\expr{f})_\alpha(o) + \partial(\expr{f})_\alpha(\accept)\frac{\partial(\expr{f} \oplus_\prob{r} \test{\one})_\alpha(o)}{\prob{1}-\partial(\expr{f} \oplus_\prob{r} \test{\one})_\alpha(\accept)}\\
                &=\partial(\expr{f})_\alpha(o) + \partial(\expr{f})_\alpha(\accept)\frac{\partial((\expr{f} \oplus_\prob{r} \test{\one}) +_\test{c} \expr{g})_\alpha(o)}{\prob{1}-\partial((\expr{f} \oplus_\prob{r} \test{\one}) +_\test{c} \expr{g})_\alpha(\accept)}\\
                &=\partial(\expr{f})_\alpha(o) + \partial(\expr{f})_\alpha(\accept)\frac{\partial(\expr{e})_\alpha(o)}{\prob{1}-\partial(\expr{e})_\alpha(\accept)} \tag{Induction hypothesis}\\
                &=\partial(\expr{f})_\alpha(o) + \partial(\expr{f})_\alpha(\accept)\partial\left( \expr{e}^{(\test{b})}\right)_\alpha(o)\\
                &=\partial\left(\expr{f}\seq\expr{e}^{(\test{b})}\right)_\alpha(o)\\
                &=  \partial\left( \expr{f}\seq\expr{e}^{(\test{b})} +_\test{b} \test{\one} \right)_\alpha(o)\\
                &= \partial\left( \test{c}\seq(\expr{f}\seq\expr{e}^{(\test{b})} +_\test{b} \test{\one}) \right)_\alpha(o)\tag{\cref{lem:asserting_derivative}}
            \end{align*}
            Now, we consider the case of outputting \(\accept\). We have that
            \begin{align*}
                \partial\left(\test{c}\seq\expr{e}^{(\test{b})}\right)_\alpha(o) &= \partial\left(\expr{e}^{(\test{b})}\right)_\alpha(\accept) \tag{\cref{lem:asserting_derivative}}\\
                &=\prob{0} \\
                &=\partial(\expr{f})_\alpha(\accept)\partial\left(\expr{e}^{(\test{b})} \right)_\alpha(\accept) \\
                &=\partial\left( \expr{f}\seq\expr{e}^{(\test{b})}\right)_\alpha(\accept)\\
                &=\partial\left( \expr{f}\seq\expr{e}^{(\test{b})}\right)_\alpha(\accept)\\
                &=\partial\left(\expr{f}\seq\expr{e}^{(\test{b})} +_\test{b} \test{\one}\right)_\alpha(\accept)\\
                &=\partial\left(\test{c}\seq(\expr{f}\seq\expr{e}^{(\test{b})} +_\test{b} \test{\one})\right)_\alpha(\accept) \tag{\cref{lem:asserting_derivative}}
            \end{align*}
            \item Let \(\action{p} \in \Act\) and \(Q \in \Exp / \equiv\). We have that
            {
            \small
            \begin{align*}
                \partial\left(\test{c}\seq\expr{e}^{(\test{b})}\right)_\alpha[\{\action{p}\}\times Q] &= \partial\left(\expr{e}^{(\test{b})}\right)_\alpha[\{\action{p}\}\times Q] \tag{\cref{lem:asserting_derivative}} \\
                &=\frac{\partial(\expr{e})_\alpha[\{\action{p}\}\times Q / \expr{e}^{(\test{b})}]}{\prob{1}-\partial(\expr{e})_\alpha(\accept)}\\
                &=\frac{\prob{r}\partial(\expr{f})_\alpha[\{\action{p}\}\times Q / \expr{e}^{(\test{b})}]}{\prob{1}-\prob{r}\partial(\expr{f})_\alpha(\accept)-\prob{(1-r)}} \tag{Induction hypothesis}\\
                &=\frac{\partial(\expr{f})_\alpha[\{\action{p}\}\times Q / \expr{e}^{(\test{b})}]}{\prob{1}-\partial(\expr{f})_\alpha(\accept)}\\
                &=\partial(\expr{f})_\alpha[\{\action{p}\}\times Q] \left( \frac{\prob{1}-\partial(\expr{f})_\alpha(\accept) + \partial(\expr{f})_\alpha(\accept)}{\prob{1}-\partial(\expr{f})_\alpha(\accept)} \right)\\
                &=\partial(\expr{f})_\alpha[\{\action{p}\}\times Q] \left( \prob{1} + \frac{\partial(\expr{f})_\alpha(\accept)}{\prob{1}-\partial(\expr{f})_\alpha(\accept)} \right)\\
                &=\partial(\expr{f})_\alpha[\{\action{p}\}\times Q/\expr{e}^{(\test{b})}] + \partial(\expr{f})_\alpha(\accept)\frac{\partial(\expr{f})_\alpha[\{\action{p}\}\times Q/\expr{e}^{(\test{b})}]}{\prob{1}-\partial(\expr{f})_\alpha(\accept)}\\
                &=\partial(\expr{f})_\alpha[\{\action{p}\}\times Q/\expr{e}^{(\test{b})}] + \partial(\expr{f})_\alpha(\accept)\frac{\prob{r}\partial(\expr{f})_\alpha[\{\action{p}\}\times Q/\expr{e}^{(\test{b})}]}{\prob{1} - \prob{r}\partial(\expr{f})_\alpha(\accept) - \prob{(1-r)}}\\
                &=\partial(\expr{f})_\alpha[\{\action{p}\}\times Q/\expr{e}^{(\test{b})}] + \partial(\expr{f})_\alpha(\accept)\frac{\partial((\expr{f} \oplus_\prob{r} \test{\one}) +_\test{c} \expr{g})_\alpha[\{\action{p}\}\times Q/\expr{e}^{(\test{b})}]}{\prob{1} - \partial((\expr{f} \oplus_\prob{r} \test{\one}) +_\test{c} \expr{g})_\alpha(\accept)}\\
                &=\partial(\expr{f})_\alpha[\{\action{p}\}\times Q/\expr{e}^{(\test{b})}] + \partial(\expr{f})_\alpha(\accept)\frac{\partial(\expr{e})_\alpha[\{\action{p}\}\times Q/\expr{e}^{(\test{b})}]}{\prob{1} - \partial(\expr{e})_\alpha(\accept)}\tag{Induction hypothesis}\\
                &=\partial(\expr{f})_\alpha[\{\action{p}\}\times Q/\expr{e}^{(\test{b})}] + \partial(\expr{f})_\alpha(\accept)\partial\left(\expr{e}^{(\test{b})}\right)_\alpha[\{\action{p}\}\times Q]\\
                &=\partial\left(\expr{f}\seq\expr{e}^{(\test{b})} \right)_\alpha[\{\action{p}\}\times Q] \tag{\cref{lem:sequencing_cutting_postfixes}}\\
                &=\partial\left(\expr{f}\seq\expr{e}^{(\test{b})} +_\test{b} \test{\one} \right)_\alpha[\{\action{p}\}\times Q]\\
                &=\partial\left(\test{c}\seq(\expr{f}\seq\expr{e}^{(\test{b})} +_\test{b} \test{\one}) \right)_\alpha[\{\action{p}\}\times Q] \tag{\cref{lem:asserting_derivative}}
            \end{align*}
            }
        \end{enumerate}
        \item[(\acro{L6})] Assume that premises hold and are satisfied by bisimilarity. In particular, we have that \((\expr{e}, (\expr{f} \oplus_\prob{s} \test{\one}) +_\test{c} \expr{g}) \in \Phi_\partial(\equiv)\).
        By assumption, we have that \(\prob{r(1-s)}\neq\prob{1}\). We will also write \(\prob{t}\) as a shorthand for \(\prob{\frac{rs}{1-r(1-s)}}\).

        First, consider the case \(\alpha \bleq \bar{\test{c}}\).
        We have that
        \begin{align*}
            \partial\left(\test{c}\seq\expr{e}^{[\prob{r}]}\right)_\alpha(\reject) &= \partial(\test{c})_\alpha(\reject) + \partial(\test{c})_\alpha(\accept)\partial\left(\expr{e}^{[\prob{r}]}\right)_\alpha(\reject) \\
            &=\prob{1}\\
            &=\partial(\test{c})_\alpha(\reject) + \partial(\test{c})_\alpha(\accept)\partial\left(\expr{f}\seq\expr{e}^{[\prob{r}]} \oplus_{\prob{t}} \test{\one}\right)_\alpha(\reject) \\
            &=\partial\left(\test{c}\seq(\expr{f}\seq\expr{e}^{[\prob{r}]} \oplus_{\prob{t}} \test{\one})\right)_\alpha(\reject)
        \end{align*}
        Since both distributions assign all probability mass to \(\reject\), the conditions of \cref{lem:preservation_of_equivalences} hold.
        For the rest of cases assume that \(\alpha \bleq \test{c}\). We verify the conditions of \cref{lem:preservation_of_equivalences}.
        \begin{enumerate}
            \item Let \(o \in \{\reject\} + \V\). We have that
            \begin{align*}
                &\partial\left(\test{c}\seq\expr{e}^{[\prob{r}]}\right)_\alpha(o) = \partial \left(\expr{e}^{[\prob{r}]}\right)_\alpha(o) \tag{\cref{lem:asserting_derivative}}\\
                &=\frac{\prob{r}\partial(\expr{e})_\alpha(o)}{\prob{1}-\prob{r}\partial(\expr{e})_\alpha(\accept)}\\
                &=\frac{\prob{r}\partial((\expr{f} \oplus_\prob{s} \test{\one}) +_\test{c} \expr{g})_\alpha(o)}{\prob{1}-\prob{r}\partial((\expr{f} \oplus_\prob{s} \test{\one}) +_\test{c} \expr{g})_\alpha(\accept)} \tag{Induction hypothesis}\\
                &=\frac{\prob{rs}\partial(\expr{f})_\alpha(o)}{\prob{1}-\prob{rs}\partial(\expr{f})_\alpha(\accept) - \prob{r(1-s)}}\\
                &=\prob{\frac{rs}{1-r(1-s)}}\left(\frac{\partial(\expr{f})_\alpha(o)\prob{\left(1-r(1-s)\right)}}{\prob{1}-\prob{rs}\partial(\expr{f})_\alpha(\accept)-\prob{r(1-s)}}\right)\\
                &=\prob{t}\left(\frac{\partial(\expr{f})_\alpha(o)\left(\prob{1}-\prob{r(1-s)}\right) + \prob{rs}\partial(\expr{f})_\alpha(\accept)\partial(\expr{f})_\alpha(o)-\prob{rs}\partial(\expr{f})_\alpha(\accept)\partial(\expr{f})_\alpha(o)}{\prob{1}-\prob{rs}\partial(\expr{f})_\alpha(\accept)-\prob{r(1-s)}}\right)\\
                &=\prob{t}\left(\frac{\partial(\expr{f})_\alpha(o)\left(\prob{1}-\prob{rs}\partial(\expr{f})_\alpha(\accept)-\prob{r(1-s)}\right) + \prob{rs}\partial(\expr{f})_\alpha(\accept)\partial(\expr{f})_\alpha(o)}{\prob{1}-\prob{rs}\partial(\expr{f})_\alpha(\accept)-\prob{r(1-s)}}\right)\\
                &=\prob{t}\left(\partial(\expr{f})_\alpha(o)+\partial(\expr{f})_\alpha(\accept)\frac{\prob{rs}\partial(\expr{f})_\alpha(o)}{\prob{1}-\prob{rs}\partial(\expr{f})_\alpha(\accept)-\prob{r(1-s)}}\right)\\
                &=\prob{t} \left(\partial(\expr{f})_\alpha(o) + \partial(\expr{f})_\alpha(\accept)\frac{\prob{r}\partial((\expr{f} \oplus_{\prob{s}} \test{\one}) +_\test{c} \expr{g})_\alpha(o)}{\prob{1}-\prob{r}\partial((\expr{f} \oplus_{\prob{s}} \test{\one}) +_\test{c} \expr{g})_\alpha(\accept)}\right)\\
                &=\prob{t}\left(\partial(\expr{f})_\alpha(o) + \partial(\expr{f})_\alpha(\accept)\frac{\prob{r}\partial(\expr{e})_\alpha(o)}{\prob{1}-\prob{r}\partial(\expr{e})_\alpha(\accept)}\right)\tag{Induction hypothesis}\\
                &=\prob{t}\left(\partial(\expr{f})_\alpha(o) + \partial(\expr{f})_\alpha(\accept)\partial\left(\expr{e}^{[\prob{r}]}\right)_\alpha(o)\right)\\
                &=\partial\left(\expr{f}\seq{\expr{e}}^{[\prob{r}]} \oplus_{\prob{t}} \test{\one}\right)_\alpha(o)\\
                &=\partial\left(\test{c}\seq(\expr{f}\seq{\expr{e}}^{[\prob{r}]} \oplus_{\prob{t}} \test{\one})\right)_\alpha(o) \tag{\cref{lem:asserting_derivative}}
            \end{align*}
            As for the probability of outputting \(\accept\), we have that
            \begin{align*} &\partial\left(\test{c}\seq\expr{e}^{[\prob{r}]}\right)_\alpha(\accept) =\partial\left(\expr{e}^{[\prob{r}]}\right)_\alpha(\accept) \tag{\cref{lem:asserting_derivative}}\\
            &=\frac{\prob{1-r}}{\prob{1}-\prob{r}\partial(\expr{e})_\alpha(\accept)} \\
            &=\frac{\prob{1-r}}{\prob{1}-\prob{r}\partial((\expr{f} \oplus_{\prob{s}}  \test{\one}) +_\test{c} \expr{g} )_\alpha(\accept)} \tag{Induction hypothesis}\\
            &=\frac{\prob{1-r}}{\prob{1}-\prob{rs}\partial(\expr{f})_\alpha(\accept)-\prob{r(1-s)}}\\
            &=\frac{\prob{(1-r)}\left(\prob{t}\partial(\expr{f})_\alpha(\accept) - \prob{t}\partial(\expr{f})_\alpha(\accept) + \prob{1}\right)}{\prob{1}-\prob{rs}\partial(\expr{f})_\alpha(\accept)-\prob{r(1-s)}} \\
            &=\frac{\prob{(1-r)}\left(\prob{t}\partial(\expr{f})_\alpha(\accept) - \prob{\frac{rs}{1-r(1-s)}}\partial(\expr{f})_\alpha(\accept) + \prob{1}\right)}{\prob{1}-\prob{rs}\partial(\expr{f})_\alpha(\accept)-\prob{r(1-s)}} \\
            &=\frac{\prob{(1-r)}\prob{t}\partial(\expr{f})_\alpha(\accept) - \prob{\frac{1-r}{1-r(1-s)}}\prob{rs}\partial(\expr{f})_\alpha(\accept) + \prob{(1-r(1-s))\frac{1-r}{1-r(1-s)}}}{\prob{1}-\prob{rs}\partial(\expr{f})_\alpha(\accept)-\prob{r(1-s)}} \\
            &=\frac{\prob{t}\partial(\expr{f})_\alpha(\accept)\prob{(1-r)} + \prob{\frac{1-r}{1-r(1-s)}}\left(\prob{1}-\prob{rs}\partial(\expr{f})_\alpha(\accept) - \prob{r(1-s)} \right)}{\prob{1}-\prob{rs}\partial(\expr{f})_\alpha(\accept)-\prob{r(1-s)}} \\
            &=\prob{t}\partial(\expr{f})_\alpha(\accept)\frac{\prob{1-r}}{\prob{1}-\prob{rs}\partial(\expr{f})_\alpha(\accept)-\prob{r(1-s)}} + \prob{\frac{1-r}{1-r(1-s)}}\\
            &=\prob{t}\partial(\expr{f})_\alpha(\accept)\frac{\prob{1-r}}{\prob{1}-\prob{r}\partial((\expr{f} \oplus_\prob{s} \test{\one}) +_\test{c} \expr{g})_\alpha(\accept)} + \prob{(1-t)}\\
            &=\prob{t}\partial(\expr{f})_\alpha(\accept)\frac{\prob{1-r}}{\prob{1}-\prob{r}\partial(\expr{e})_\alpha(\accept)} + \prob{(1-t)}\tag{Induction hypothesis}\\ &=\prob{t}\partial(\expr{f})_\alpha(\accept)\partial\left(\expr{e}^{[\prob{r}]}\right)_\alpha(\accept) + \prob{(1-t)}\\
            &=\partial\left(\expr{f}\seq\expr{e}^{[\prob{r}]} \oplus_\prob{t} \test{\one}\right)_\alpha(\accept) \\
            &=\partial\left(\test{c}\seq(\expr{f}\seq\expr{e}^{[\prob{r}]} \oplus_\prob{t} \test{\one})\right)_\alpha(\accept) \tag{\cref{lem:asserting_derivative}}
            \end{align*}
            \item Let \(\action{p} \in \Act\) and \(Q \in \Exp / \equiv\). Observe that
            {
            \small
            \begin{align*}
                &~~\partial\left(\test{c}\seq\expr{e}^{[\prob{r}]}\right)_\alpha[\{\action{p}\}\times Q]\\
                &= \partial \left(\expr{e}^{[\prob{r}]}\right)_\alpha[\{\action{p}\}\times Q] \tag{\cref{lem:asserting_derivative}}\\
                &=\frac{\prob{r}\partial(\expr{e})_\alpha[\{\action{p}\}\times Q/\expr{e}^{[\prob{r}]}]}{\prob{1}-r\partial(\expr{e})_\alpha(\accept)}\\
                &=\frac{\prob{r}\partial((\expr{f} \oplus_\prob{s} \test{\one}) +_\test{c} \expr{g})_\alpha[\{\action{p}\}\times Q/\expr{e}^{[\prob{r}]}]}{\prob{1}-r\partial((\expr{f} \oplus_\prob{s} \test{\one}) +_\test{c} \expr{g})_\alpha(\accept)} \tag{Induction hypothesis}\\
                &=\frac{\prob{rs}\partial(\expr{f})_\alpha[\{\action{p}\}\times Q/\expr{e}^{[\prob{r}]}]}{\prob{1}-\prob{rs}\partial(\expr{f})_\alpha(\accept) - \prob{r(1-s)}}\\
                &=\prob{\frac{rs}{1-r(1-s)}}\left(\frac{\partial(\expr{f})_\alpha[\{\action{p}\}\times Q/\expr{e}^{[\prob{r}]}]\prob{\left(1-r(1-s)\right)}}{\prob{1}-\prob{rs}\partial(\expr{f})_\alpha(\accept)-\prob{r(1-s)}}\right)\\
                &=\prob{t}\left(\frac{\partial(\expr{f})_\alpha[\{\action{p}\}\times Q/\expr{e}^{[\prob{r}]}]\prob{\left(1-r(1-s)\right)} + \prob{rs}\partial(\expr{f})_\alpha(\accept)\partial(\expr{f})_\alpha[\{\action{p}\}\times Q/\expr{e}^{[\prob{r}]}]}{\prob{1}-\prob{rs}\partial(\expr{f})_\alpha(\accept)-\prob{r(1-s)}}\right)\\
                &\hspace{5em}-t\frac{\prob{rs}\partial(\expr{f})_\alpha(\accept)\partial(\expr{f})_\alpha[\{\action{p}\}\times Q/\expr{e}^{[\prob{r}]}]}{\prob{1}-\prob{rs}\partial(\expr{f})_\alpha(\accept)-\prob{r(1-s)}}\\
                &=\prob{t}\left(\frac{\partial(\expr{f})_\alpha[\{\action{p}\}\times Q/\expr{e}^{[\prob{r}]}]\left(\prob{1}-\prob{rs}\partial(\expr{f})_\alpha(\accept)-\prob{r(1-s)}\right) + \prob{rs}\partial(\expr{f})_\alpha(\accept)\partial(\expr{f})_\alpha[\{\action{p}\}\times Q/\expr{e}^{[\prob{r}]}]}{\prob{1}-\prob{rs}\partial(\expr{f})_\alpha(\accept)-\prob{r(1-s)}}\right)\\
                &=\prob{t}\left(\partial(\expr{f})_\alpha[\{\action{p}\}\times Q/\expr{e}^{[\prob{r}]}]+\partial(\expr{f})_\alpha(\accept)\frac{\prob{rs}\partial(\expr{f})_\alpha[\{\action{p}\}\times Q/\expr{e}^{[\prob{r}]}]}{\prob{1}-\prob{rs}\partial(\expr{f})_\alpha(\accept)-\prob{r(1-s)}}\right)\\
                &=\prob{t} \left(\partial(\expr{f})_\alpha[\{\action{p}\}\times Q/\expr{e}^{[\prob{r}]}] + \partial(\expr{f})_\alpha(\accept)\frac{\prob{r}\partial((\expr{f} \oplus_{\prob{s}} \test{\one}) +_\test{c} \expr{g})_\alpha[\{\action{p}\}\times Q/\expr{e}^{[\prob{r}]}]}{\prob{1}-\prob{r}\partial((\expr{f} \oplus_{\prob{s}} \test{\one}) +_\test{c} \expr{g})_\alpha(\accept)}\right)\\
                &=\prob{t}\left(\partial(\expr{f})_\alpha[\{\action{p}\}\times Q/\expr{e}^{[\prob{r}]}] + \partial(\expr{f})_\alpha(\accept)\frac{\prob{r}\partial(\expr{e})_\alpha[\{\action{p}\}\times Q/\expr{e}^{[\prob{r}]}]}{\prob{1}-\prob{r}\partial(\expr{e})_\alpha(\accept)}\right)\tag{Induction hypothesis}\\
                &=\prob{t}\left(\partial(\expr{f})_\alpha[\{\action{p}\}\times Q/\expr{e}^{[\prob{r}]}] + \partial(\expr{f})_\alpha(\accept)\partial\left(\expr{e}^{[\prob{r}]}\right)_\alpha[\{\action{p}\}\times Q]\right)\\
                &=\partial\left(\expr{f}\seq{\expr{e}}^{[\prob{r}]} \oplus_{\prob{t}} \test{\one}\right)_\alpha[\{\action{p}\}\times Q]\\
                &=\partial\left(\test{c}\seq(\expr{f}\seq{\expr{e}}^{[\prob{r}]} \oplus_{\prob{t}} \test{\one})\right)_\alpha[\{\action{p}\}\times Q] \tag*{(\cref{lem:asserting_derivative}) \qedhere}
            \end{align*}
            }
        \end{enumerate}
\end{description}
\endgroup
\end{proof}
\begin{remark}\label{rem:axioms}
    We can substitute $\expr{e}$ in the consequence of \acro{L5}, to obtain
\[
    \test{c} \seq ((\expr{f} \oplus_s \test{1} +_{\test{c}} \expr{g}))^{(\test{b})} \equiv \test{c} \seq (\expr{f} \seq ((\expr{f} \oplus_s \test{1}) +_{\test{c}} \expr{g})^{(\test{b})} +_{\test{b}} \test{1})
\]
which is equivalent to \acro{L5} by congruence.
A similar translation turns \acro{L6} into an equational axiom.
Alternatively, one could replace \acro{L5} and \acro{L6} with the following quasi-equational axiom, which can be proved to imply both:
\[
	\inferrule{
        \expr{e} \equiv ((\expr{f} \oplus_s \expr{e}) \oplus_{\prob{r}} \test{1}) +_{\test{c}} \expr{g}
        \\
        \prob{r}(1-\prob{s})\neq 1
    }{
        \expr{e} \equiv (\expr{f} \oplus_{\frac{\prob{r}\prob{s}}{1-\prob{r}(1-\prob{s})}} \test{1}) +_{\test{c}} \expr{g}
    }
\]
This rule cannot be replaced by an axiom in the same way that \acro{L5} and \acro{L6} can, on account of the recurrence of $\expr{e}$ in the premise.
\end{remark}
\subsection{Derivable facts}
\begin{lemma}\label{lem:derivable_facts}
    The following equivalences are derivable from \(\equiv\) for all \(\expr{e}, \expr{f}, \expr{g}, \expr{h} \in \Exp\), \(\test{b}, \test{c} \in \Bexp\) and \(\prob{r} \in [0,1]\).
    \begin{description}
    \item[\normalfont\bf (\acro{DF1})] \(\expr{e} +_\test{b} (\expr{f} +_\test{c} \expr{g}) \equiv (\expr{e} +_\test{b} \expr{f}) +_{\test{b} + \test{c}} \expr{g} \)
    \item[\normalfont\bf (\acro{DF2})] \(\expr{e} +_\test{b} \test{\zero} \equiv \test{b}\seq\expr{e}\)
    \item[\normalfont\bf (\acro{DF3})] \(\test{b} \seq (\expr{e} +_\test{b} \expr{f}) \equiv \test{b} \seq \expr{e}\)
    \item[\normalfont\bf (\acro{DF4})] \((\expr{e} +_\test{b} \expr{f}) +_\test{c} \expr{g} \equiv (\expr{e} +_\test{bc} \expr{f}) +_\test{c} \expr{g} \)
    \item[\normalfont\bf (\acro{DF5})] \((\expr{e} +_\test{b} \expr{f}) +_\test{c} (\expr{g} +_\test{b} \expr{h}) \equiv (\expr{e} +_\test{c} \expr{g}) +_\test{b} (\expr{f} +_\test{c} \expr{h}) \)
    \item[\normalfont\bf (\acro{DF6})] \(\test{b}\seq(\expr{e} +_\test{c} \expr{f}) \equiv \test{b}\seq \expr{e} +_\test{c} \test{b}\seq\expr{f}\)
    \item[\normalfont\bf (\acro{DF7})] \(\test{b}\seq (\expr{e} +_\test{c} \expr{f}) \equiv \test{b}\seq (\test{b}\seq\expr{e} +_\test{c} \expr{f})\)
    \item[\normalfont\bf (\acro{DF8})] \((\expr{e} +_\test{b} \expr{f}) \oplus_{\prob{r}} \expr{g} \equiv (\expr{e} \oplus_{\prob{r}} \expr{g}) +_\test{b} (\expr{f} \oplus_{\prob{r}} \expr{g} ) \)
    \item[\normalfont\bf (\acro{DF9})] \((\expr{e} +_\test{b} \expr{f}) \oplus_\prob{r} (\expr{g} +_\test{b} \expr{h}) \equiv (\expr{e} \oplus_\prob{r} \expr{g}) +_\test{b} (\expr{f} \oplus_{\prob{r}} \expr{h})\)
    \item[\normalfont\bf (\acro{DF10})] \(\test{b}\seq(\expr{e} \oplus_\prob{r} \expr{f}) \equiv \test{b}\seq\expr{e} \oplus_\prob{r} \test{b}\seq\expr{f}\)
    \item[\normalfont\bf (\acro{DF11})] \(\expr{e} \oplus_\prob{r} (\expr{f} \oplus_{\prob{s}} \expr{g}) \equiv (\expr{e} \oplus_\prob{k} \expr{f}) \oplus_\prob{z} \expr{g}\) where \(\prob{k}=\prob{\frac{r}{1-(1-r)(1-s)}}\) and \(\prob{l}=\prob{1-(1-r)(1-s)}\).
    \item[\normalfont\bf (\acro{DF12})] \(\expr{e} +_\test{1} \expr{f} \equiv \expr{e}\)
    \end{description}
\end{lemma}
\begin{proof}
We refer to~\cite{Smolka:2020:Guarded} for (\acro{DF1}), (\acro{DF2}), (\acro{DF3}) and (\acro{DF6}).
The other equivalences are proved as follows.
    \begin{description}
    \item[(\acro{DF4})]
    We derive as follows.
    \begin{align*}
        (\expr{e} +_\test{b} \expr{f}) +_\test{c} \expr{g} &\equiv \test{c}\seq(\expr{e} +_\test{b} \expr{f}) +_\test{c} \expr{g} \tag{\acro{G2}}\\
        &\equiv ((\expr{e} +_\test{b} \expr{f}) +_\test{c} \test{\zero}) +_\test{c} \expr{g} \tag{\acro{DF2}}\\
        &\equiv (\expr{e} +_{\test{bc}} (\expr{f} +_\test{c} \test{\zero})) +_\test{c} \expr{g} \tag{\acro{G4}}\\
        &\equiv ((\expr{e} +_\test{bc} \expr{f}) +_{\test{bc} + \test{c}} \test{\zero}) +_\test{c} \expr{g} \tag{\acro{DF1}}\\
        &\equiv ((\expr{e} +_\test{bc} \expr{f}) +_{\test{c}} \test{\zero}) +_\test{c} \expr{g} \tag{Boolean algebra}\\
        &\equiv \test{c}\seq((\expr{e} +_\test{bc} \expr{f}) +_{\test{c}} \test{\zero}) +_\test{c} \expr{g} \tag{\acro{G2}}\\
        &\equiv \test{c}\seq(\expr{e} +_\test{bc} \expr{f}) +_\test{c} \expr{g} \tag{\acro{DF3}}\\
        &\equiv (\expr{e} +_\test{bc} \expr{f}) +_\test{c} \expr{g} \tag{\acro{G2}}
    \end{align*}
    \item[(\acro{DF5})]
    We derive as follows.
    \begin{align*}
        (\expr{e} +_\test{b} \expr{f}) +_\test{c} (\expr{g} +_\test{b} \expr{h}) &\equiv \expr{e} +_\test{bc} (\expr{f} +_\test{c}  (\expr{g} +_\test{b} \expr{h})) \tag{\acro{G4}}\\
        &\equiv \expr{e} +_\test{bc} (  (\expr{g} +_\test{b} \expr{h}) +_{\bar{\test{c}}} \expr{f} ) \tag{\acro{G3}}\\
        &\equiv \expr{e} +_\test{bc} (  \expr{g} +_{\test{b}\bar{\test{c}}} (\expr{h} +_{\bar{\test{c}}} \expr{f}) ) \tag{\acro{G4}}\\
        &\equiv \expr{e} +_\test{bc} (  \expr{g} +_{\test{b}\bar{\test{c}}} (\expr{f} +_{{\test{c}}} \expr{h}) ) \tag{\acro{G3}}\\
        &\equiv (\expr{e} +_\test{bc}   \expr{g}) +_{\test{b}\bar{\test{c}} + \test{b}\test{c}} (\expr{f} +_{{\test{c}}} \expr{h})  \tag{\acro{DF1}}\\
        &\equiv (\expr{e} +_\test{bc}   \expr{g}) +_{\test{b}} (\expr{f} +_{{\test{c}}} \expr{h})  \tag{Boolean algebra}\\
        &\equiv (\expr{e} +_\test{c}   \expr{g}) +_{\test{b}} (\expr{f} +_{{\test{c}}} \expr{h})  \tag{\acro{DF4}}
    \end{align*}
    \item[(\acro{DF7})]
    We derive as follows.
    \begin{align*}
        \test{b}\seq(\expr{e} +_\test{c} \expr{f} )&\equiv \test{b}\seq\expr{e} +_\test{c} \test{b} \seq \expr{f} \tag{\acro{DF6}}\\
        &\equiv \test{bb}\seq\expr{e} +_\test{c} \test{b} \seq \expr{f} \tag{Boolean algebra}\\
        &\equiv \test{b}\seq\test{b}\seq\expr{e} +_\test{c} \test{b} \seq \expr{f} \tag{\acro{S8}}\\
        &\equiv \test{b}\seq (\test{b}\seq\expr{e} +_\test{c} \expr{f}) \tag{\acro{DF6}}
    \end{align*}
    \item[(\acro{DF8})]
    We derive as follows.
    \begin{align*}
        (\expr{e} +_\test{b} \expr{f}) \oplus_{\prob{r}} \expr{g} &\equiv \expr{g} \oplus_{1-\prob{r}} (\expr{e} +_\test{b} \expr{f}) \tag{\acro{P3}}\\
        &\equiv (\expr{g} \oplus_{\prob{1-r}} \expr{e}) +_\test{b} (\expr{g} \oplus_{\prob{1-r}} \expr{f}) \tag{\acro{D}} \\
        &\equiv (\expr{e} \oplus_{\prob{r}} \expr{g}) +_\test{b} (\expr{f} \oplus_{\prob{r}} \expr{g})  \tag{\acro{P3}}
    \end{align*}
        \item[(\acro{DF9})]
        We derive as follows.
        \begin{align*}
            (\expr{e} +_\test{b} \expr{f}) \oplus_\prob{r} (\expr{g} +_\test{b} \expr{h}) &\equiv ( (\expr{e} +_\test{b} \expr{f}) \oplus_\prob{r} \expr{g}) +_\test{b} ((\expr{e} +_\test{b} \expr{f}) \oplus_\prob{r} \expr{h})  \tag{\acro{D}}\\
            &\equiv ( (\expr{e} \oplus_\prob{r} \expr{g})  +_\test{b} (\expr{f} \oplus_\prob{r} \expr{g})) +_\test{b} ((\expr{e} +_\test{b} \expr{f}) \oplus_\prob{r} \expr{h}) \tag{\acro{DF8}}\\
            &\equiv \test{b}\seq( (\expr{e} \oplus_\prob{r} \expr{g})  +_\test{b} (\expr{f} \oplus_\prob{r} \expr{g})) +_\test{b} ((\expr{e} +_\test{b} \expr{f}) \oplus_\prob{r} \expr{h}) \tag{\acro{G2}}\\
            &\equiv \test{b}\seq(\expr{e} \oplus_\prob{r} \expr{g})  +_\test{b} ((\expr{e} +_\test{b} \expr{f}) \oplus_\prob{r} \expr{h}) \tag{\acro{DF3}}\\
            &\equiv (\expr{e} \oplus_\prob{r} \expr{g})  +_\test{b} ((\expr{e} +_\test{b} \expr{f}) \oplus_\prob{r} \expr{h}) \tag{\acro{G2}}\\
            &\equiv ((\expr{e} +_\test{b} \expr{f}) \oplus_\prob{r} \expr{h}) +_{\bar{\test{b}}} (\expr{e} \oplus_\prob{r} \expr{g}) \tag{\acro{G3}}\\
            &\equiv ((\expr{f} +_{\bar{\test{b}}} \expr{e}) \oplus_\prob{r} \expr{h}) +_{\bar{\test{b}}} (\expr{e} \oplus_\prob{r} \expr{g}) \tag{\acro{G3}}\\
            &\equiv ((\expr{f} \oplus_\prob{r} \expr{h}) +_{\bar{\test{b}}} (\expr{e} \oplus_\prob{r} \expr{h}) ) +_{\bar{\test{b}}} (\expr{e} \oplus_\prob{r} \expr{g}) \tag{\acro{DF8}}\\
            &\equiv \bar{\test{b}}\seq((\expr{f} \oplus_\prob{r} \expr{h}) +_{\bar{\test{b}}} (\expr{e} \oplus_\prob{r} \expr{h}) ) +_{\bar{\test{b}}} (\expr{e} \oplus_\prob{r} \expr{g}) \tag{\acro{G2}}\\
            &\equiv \bar{\test{b}}\seq(\expr{f} \oplus_\prob{r} \expr{h}) +_{\bar{\test{b}}} (\expr{e} \oplus_\prob{r} \expr{g}) \tag{\acro{DF3}}\\
            &\equiv (\expr{f} \oplus_\prob{r} \expr{h}) +_{\bar{\test{b}}} (\expr{e} \oplus_\prob{r} \expr{g}) \tag{\acro{G2}}\\
            &\equiv (\expr{e} \oplus_\prob{r} \expr{g}) +_{{\test{b}}}  (\expr{f} \oplus_\prob{r} \expr{h}) \tag{\acro{G3}}
        \end{align*}
        \item[(\acro{DF10})]
        We derive as follows.
        \begin{align*}
            \test{b}\seq(\expr{e} \oplus_\prob{r} \expr{f}) &\equiv (\expr{e} \oplus_\prob{r} \expr{f}) +_\test{b} \test{\zero} \tag{\acro{DF2}}\\
            &\equiv (\expr{e} \oplus_\prob{r} \expr{f}) +_\test{b} (\test{\zero} \oplus_\prob{r} \test{\zero}) \tag{\acro{P1}} \\
            &\equiv (\expr{e} +_\test{b} \test{\zero}) \oplus_{\prob{r}} (\expr{f} +_\test{b} \test{\zero}) \tag{\acro{DF9}}\\
            &\equiv \test{b}\seq\expr{e} \oplus_\prob{r} \test{b}\seq\expr{f}
        \end{align*}
        \item[(\acro{DF11})]
        We derive as follows.
        \begin{align*}
            \expr{e} \oplus_\prob{r} (\expr{f} \oplus_\prob{s} \expr{g}) &\equiv (\expr{f} \oplus_\prob{s} \expr{g}) \oplus_\prob{1-r} \expr{e} \tag{\acro{P3}}\\
            &\equiv (\expr{g} \oplus_\prob{1-s} \expr{f}) \oplus_\prob{1-r} \expr{e} \tag{\acro{P3}}\\
            &\equiv \expr{g} \oplus_{\prob{(1-r)(1-s)}} (\expr{f} \oplus_\prob{\frac{s(1-r)}{1-(1-r)(1-s)}} \expr{e} ) \tag{\acro{P4}}\\
             &\equiv \expr{g} \oplus_{\prob{1-l}} (\expr{f} \oplus_\prob{1-k} \expr{e} ) \tag{\(\prob{k} = \prob{\frac{r}{1-(1-r)(1-s)}}\) and \(\prob{l}=\prob{1-(1-r)(1-s)}\)}\\
             &\equiv (\expr{f} \oplus_\prob{1-k} \expr{e}) \oplus_\prob{l} \expr{g} \tag{\acro{P3}}\\
             &\equiv (\expr{e} \oplus_\prob{k} \expr{f}) \oplus_\prob{l} \expr{g} \tag*{(\acro{P3}) \qedhere}
        \end{align*}
        \item[(\acro{D12})]
        We derive as follows.
        \begin{align*}
            \expr{e} +_{\test{\one}} \expr{f} &\equiv \test{\one}\seq( \expr{e} +_{\test{\one}} \expr{f}) \tag{\acro{S1}}\\
            &\equiv  \expr{e} \tag{\acro{DF3}}
        \end{align*}
    \end{description}
\end{proof}

\subsection{Generalised guarded and convex sums}\label{apx:generalised}

We now rigorously define our generalised sum operators, and verify that they expect as behaved.
First off is the guarded sum.

\begin{definition}
    Consider a subset \(\Phi \subseteq \At\) and a collection \(\{\expr{e}_\alpha\}_{\alpha \in \Phi}\) such that for all \(\alpha \in \Phi\), \(\expr{e}_\alpha \in \Exp\). A \emph{generalised guarded sum} is an expression inductively given by the following
    \[
    \bigsum{\alpha \in \Phi } \expr{e}_\alpha = \expr{e}_\gamma +_{\test{\gamma}} \left( \bigsum{\alpha \in \Phi \setminus \{ \gamma \}} \expr{e}_\alpha\right) \quad \text{if } \gamma \in \Phi
    \quad\quad\quad\quad
    \bigsum{\alpha \in \Phi} \expr{e}_\alpha = \test{\zero} \quad \text{if } \Phi = \emptyset
    \]
\end{definition}
This definition is ambiguous in the choice of expressions from the collection when unrolling the inductive definition.
However, all possible ways of doing so are equivalent.

\begin{lemma}[{\cite[Lemma~B.1]{Smolka:2020:Guarded}}]
    Generalised guarded sums are well defined up to \(\equiv\)
\end{lemma}

We first recall properties of generalised guarded sums from~\cite{Smolka:2020:Guarded}.
\begin{lemma}[{\cite[Lemma~B.2]{Smolka:2020:Guarded}}]\label{lem:generalised_branch_selection}
	Let \(\test{b}, \test{c} \in \Bexp\) and let \(\{\expr{e}_\alpha\}_{\alpha \in \At}\) be an indexed collection such that \(\expr{e}_\alpha \in \Exp\) for all \(\alpha \in \At\). Then,
	\begin{equation*}
		\test{c} \seq \bigsum{\alpha \bleq \test{b}} \expr{e}_\alpha \equiv \bigsum{\alpha \bleq \test{bc}} \expr{e}_\alpha
	\end{equation*}
\end{lemma}
Observe that in the above lemma, similarly to~\cite{Smolka:2020:Guarded} we abuse notation and write \(\alpha \bleq \test{c}\) for the set \(\{\alpha \mid \alpha \bleq \test{c}\}\). Since the set \(\T\) is finite, the Boolean algebra generated by it is atomic and hence every \(\test{c} \in \Bexp\) corresponds to the subset of \(\At\).
\begin{lemma}[{\cite[Lemma~B.4]{Smolka:2020:Guarded}}]\label{lem:generalised_guardedness}
	Let \(\Phi \subseteq \At \) and \(\{\expr{e}_\alpha\}_{\alpha \in \Phi}\) be an indexed collection such that \(\expr{e}_\alpha \in \Exp\) for each \(\alpha \in \Phi\). Then,
	\begin{equation*}
		\bigsum{\alpha \in \Phi} \expr{e}_\alpha \equiv \bigsum{\alpha \in \Phi} \test{\alpha} \seq \expr{e}_\alpha
	\end{equation*}
\end{lemma}
Now, we state further properties of generalised guarded sums, which will be used in the proof of \cref{thm: fundamental_theorem}.
\begin{lemma}\label{lem:generalised_guarded_distributivity}
    Let \(\Phi \subseteq \At\) and \(\{\expr{e}_\alpha\}_{\alpha \in \Phi}\) be an indexed collection such that \(\expr{e}_\alpha \in \Exp\) for each \(\alpha \in \Phi\) and let \(\expr{f} \in \Exp \). Then,
	\begin{equation*}
		\left(\bigsum{\alpha \in \Phi} \expr{e}_\alpha\right)\seq \expr{f} \equiv \bigsum{\alpha \in \Phi} \expr{e}_\alpha \seq \expr{f}
	\end{equation*}
\end{lemma}
\begin{proof}
    By induction on the size of \(\Phi\). For the base case, when \(\Phi = \emptyset\) consider the following.
    \[\left(\bigsum{\alpha \in \Phi} \expr{e}_\alpha \right) \seq \expr{f }\equiv \test{\zero} \seq \expr{f} \equiv \test{\zero} \equiv\bigsum{\alpha \in \Phi} \expr{e}_\alpha \seq \expr{f} \]
    For the inductive step, we have that
    \begin{align*}
        \left(\bigsum{\alpha \in \Phi} \expr{e}_\alpha \right) \seq \expr{f }&\equiv \left(\expr{e}_{\gamma} +_\test{\gamma} \bigsum{\alpha \in \Phi \setminus \{\gamma\}}  \expr{e}_\alpha \right) \seq \expr{f} \tag{\(\Phi\) is nonempty}\\
        &\equiv \left( \expr{e}_\gamma \seq \expr{f} +_\test{\gamma} \left( \bigsum{\alpha \in \Phi \setminus \{\gamma\}}  \expr{e}_\alpha \right) \seq \expr{f}\right) \tag{\acro{S5}}\\
        &\equiv  \expr{e}_\gamma \seq \expr{f} +_\test{\gamma} \left( \bigsum{\alpha \in \Phi \setminus \{\gamma\}}  \expr{e}_\alpha\seq \expr{f} \right) \tag{Induction hypothesis}\\
        &\equiv \bigsum{\alpha \in \Phi} \expr{e}_\alpha \seq \expr{f} \qedhere
    \end{align*}
\end{proof}
\begin{lemma}\label{lem:generalised_exchange}
	Let \(\Phi \subseteq \At\) and let \(\{\expr{e}_\alpha\}_{\alpha \in \Phi}\) and \(\{\expr{f}_\alpha\}_{\alpha \in \Phi}\) be indexed collections such that \(\expr{e}_\alpha, \expr{f}_\alpha \in \Exp\) for each \(\alpha \in \Psi \) and let \(\prob{r} \in [0,1]\). Then,
	\begin{equation*}
		\left( \bigsum{\alpha \in \Phi} \expr{e}_\alpha \right) \oplus_{\prob{r}} \left(\bigsum{\alpha \in \Phi} \expr{f}_\alpha\right) \equiv \bigsum{\alpha \in \Phi}\left( \expr{e}_\alpha \oplus_{\prob{r}} \expr{f}_\alpha\right)
	\end{equation*}
\end{lemma}
\begin{proof}
    By induction on the size of \(\Phi\). For the base case, consider the following.
    \begin{align*}
        \left( \bigsum{\alpha \in \Phi} \expr{e}_\alpha \right) \oplus_\prob{r} \left( \bigsum{\alpha \in \Phi} \expr{f}_\alpha \right) &\equiv \test{\zero} \oplus_\prob{r} \test{\zero}
        \equiv \test{\zero}
        \equiv \bigsum{\alpha \in \Phi}\left( \expr{e}_\alpha \oplus_{\prob{r}} \expr{f}_\alpha\right)
    \end{align*}
    For the inductive step, we have that
    \begin{align*}
          \left( \bigsum{\alpha \in \Phi} \expr{e}_\alpha \right) \oplus_\prob{r} \left( \bigsum{\alpha \in \Phi} \expr{f}_\alpha \right) &\equiv \left( \expr{e}_\gamma +_\test{\gamma} \bigsum{\alpha \in \Phi \setminus \{\gamma\}} \expr{e}_\alpha \right) \oplus_\prob{r} \left( \expr{f}_\gamma +_\test{\gamma} \bigsum{\alpha \in \Phi \setminus \{\gamma\}} \expr{f}_\alpha \right)\\
          &\equiv (\expr{e}_\gamma \oplus_\prob{r} \expr{f}_\gamma) +_\test{b} \left(\left( \bigsum{\alpha \in \Phi \setminus \{\gamma\}} \expr{e}_\alpha  \right) \oplus_\prob{r} \left( \bigsum{\alpha \in \Phi \setminus \{\gamma\}} \expr{f}_\alpha  \right) \right) \tag{\acro{DF9}} \\
          &\equiv (\expr{e}_\gamma \oplus_\prob{r} \expr{f}_\gamma) +_\test{b} \left(\bigsum{\alpha \in \Phi \setminus \{\gamma\}} (\expr{e}_\alpha \oplus_\prob{r} \expr{f}_\alpha) \right) \tag{Induction hypothesis}\\
          &\equiv \bigsum{\alpha \in \Phi} (\expr{e}_\alpha \oplus_\prob{r} \expr{f}_\alpha) \tag*{\qedhere}
    \end{align*}
\end{proof}
\begin{remark}
In terms of notation, given \(\Phi, \Psi \subseteq \At\) such that  \(\Psi \cap \Theta = \emptyset\) and indexed collections \(\{\expr{e}_\alpha\}_{\alpha \in \Psi}\), \(\{\expr{f}_\alpha\}_{\alpha \in \Phi}\) such that \(\expr{e}_\alpha \in \Exp\) for each \(\alpha \in \Phi\) and \(\expr{f}_\alpha\) for each \(\alpha \in \Psi\) we will write
\[
	\bigsum{\alpha \in \Psi} \expr{e}_\alpha \bsum \bigsum{\alpha \in \Theta} \expr{f}_\alpha
\]
to denote the following
\[
	\bigsum{\alpha \in \Phi \cup \Psi} \expr{g}_\alpha
	\quad\quad\quad\text{where}
	\quad\quad \expr{g}_\alpha =
	\begin{cases}
		\expr{e}_\alpha & \text{if } \alpha \in \Phi \\
		\expr{f}_\alpha & \text{if } \alpha \in \Psi \\
	\end{cases}
\]
\end{remark}

Next up is the generalised convex sum.

\begin{definition}
    Given a non-empty finite index set \(I\) and indexed collections \(\{\prob{r}_i\}_{i \in I}\) and \(\{\expr{e}_i\}_{i \in I}\) such that \(\expr{e}_i \in \Exp\) and \(\prob{r}_i \in [0,1]\) for all \(i \in I\) satisfying that \(\sum_{i \in I}\prob{r}_i = \prob{1}\) we define a \emph{generalised convex sum} to be an expression given by the following inductive definition
\[
	\bigplus{i \in I} \prob{r}_i \cdot \expr{e}_i = \expr{e}_j \oplus_{\prob{r}_j} \left( \bigplus{i \in I \setminus \{j\}} {\frac{\prob{r}_i}{\prob{1}-\prob{r}_j}}\cdot \expr{e}_i\right)\quad\text{for } j\in I  \text{ and } \prob{r}_j \neq \prob{1}
 \]
 \[
 	\bigplus{i \in \{j\}} \prob{r}_i \cdot \expr{e}_i = \expr{e}_j \quad \text{otherwise}
 \]
\end{definition}
Again, the above definition is ambiguous in the choice of expressions from the collection, however all possible unrollings of the definition are provably equivalent up to \(\equiv\).
\begin{restatable}{lemma}{convexsums}\label{lem: convex_sums_well_defined}
    Generalised convex sums are well defined up to \(\equiv\)
\end{restatable}
\begin{proof}
    Let \(I\) be a finite and non-empty index set, \(\{\expr{e}_i\}_{i \in I}\) and \(\{\prob{r}_i\}_{i \in I}\) indexed collections such that for all \(i \in I\) we have that \(\expr{e}_i \in \Exp\) and \(\prob{r}_i \in [0,1]\).

    We show that generalised sums are well-defined by induction on the size of the index set \(I\). The case when \(I=\{j\}\) trivial, as in such a case the generalised sum is defined to be just \(\expr{e}_j\).

    When \(I = \{j,k\}\) we have that
    \begin{align*}
        \expr{e}_j \oplus_{\prob{r}_j} \left(\bigplus{i \in \{k\}} \prob{r}_i \cdot \expr{e}_i \right) &\equiv \expr{e}_j \oplus_{\prob{r}_j} \expr{e}_k \equiv \expr{e}_k \oplus_{\prob{1}-\prob{r}_j} \expr{e}_j \tag{\acro{P3}}\\
        &\equiv \expr{e}_k \oplus_{\prob{r}_k} \expr{e}_j \tag{\(\prob{r}_j + \prob{r}_k = \prob{1}\)} \\
        &\equiv \expr{e}_k \oplus_{\prob{r}_k} \left(\bigplus{i \in \{j\}} \prob{r}_i \cdot \expr{e}_i \right)
    \end{align*}

    For the induction step, assume that \(|I|>2\) and that \(j, k \in I\). Consider the following
    \begin{align*}
        &\expr{e}_j \oplus_{\prob{r}_j} \left( \bigplus{i \in I \setminus \{j\}} \frac{\prob{r}_i}{\prob{1}-\prob{r}_j}\cdot \expr{e}_i\right) \equiv \expr{e}_j \oplus_{\prob{r}_j} \left( \expr{e}_k \oplus_{\frac{\prob{r}_k}{\prob{1}-\prob{r}_j}} \left( \bigplus{i \in I \setminus \{j, k\}} \frac{\prob{r}_i}{(\prob{1}-\prob{r}_j)(1-\prob{r}_k)}\cdot \expr{e}_i\right)\right) \\
        &\equiv \left(\expr{e}_j \oplus_{\frac{\prob{r}_j}{\prob{r}_j+\prob{r}_k}} \expr{e}_k \right) \oplus_{\prob{r}_j + \prob{r}_k} \left( \bigplus{i \in I \setminus \{j, k\}} \frac{\prob{r}_i}{(\prob{1}-\prob{r}_j)(1-\prob{r}_k)}\cdot \expr{e}_i\right) \tag{\acro{DF11}} \\
        &\equiv \left(\expr{e}_k \oplus_{\frac{\prob{r}_k}{\prob{r}_j+\prob{r}_k}} \expr{e}_j \right) \oplus_{\prob{r}_j + \prob{r}_k} \left( \bigplus{i \in I \setminus \{j, k\}} \frac{\prob{r}_i}{(\prob{1}-\prob{r}_j)(1-\prob{r}_k)}\cdot \expr{e}_i\right) \tag{\acro{P3}} \\
        &\equiv \expr{e}_k \oplus_{\prob{r}_k} \left( \expr{e}_j \oplus_{\frac{\prob{r}_j}{\prob{1}-\prob{r}_k}} \left( \bigplus{i \in I \setminus \{j, k\}} \frac{\prob{r}_i}{(\prob{1}-\prob{r}_j)(1-\prob{r}_k)}\cdot \expr{e}_i\right)\right) \tag{\acro{P4}}\\
        &\equiv\expr{e}_k \oplus_{\prob{r}_k} \left( \bigplus{i \in I \setminus \{k\}} \frac{\prob{r}_i}{\prob{1}-\prob{r}_k}\cdot \expr{e}_i\right) \qedhere
    \end{align*}
\end{proof}

\begin{remark}
    Notation wise, given non-empty finite index sets \(I\) and \(J\) and collections \(\{\prob{r}_i\}_{i \in I}\) and \(\{\prob{s}_j\}_{j \in J}\) of probabilities as well as collections \(\{\expr{e}_i\}_{i \in I}\) and  \(\{\expr{f}_j\}_{j \in J}\) of expressions satisfying \(\sum_{i \in I} \prob{r}_i + \sum_{j\in J} \prob{s}_j = \prob{1}\) and \(I \cap J = \emptyset \) we will write
\begin{equation*}
	\bigplus{i \in I}\prob{r}_i \cdot \expr{e}_i \boplus \bigplus{j \in I}\prob{s}_i \cdot \expr{f}_i
\end{equation*}
to denote the generalised convex sum
\begin{equation*}
	\bigplus{k \in I \cup J} \prob{t}_k \cdot \expr{g}_k \quad\quad \text{where}\quad \prob{t}_k = \begin{cases}
		\prob{r}_k & \text{if } k \in I\\\prob{s}_k & \text{if } k \in J
	\end{cases}\quad \text{and}\quad \expr{g}_k =  \begin{cases}
		\expr{e}_k & \text{if } k \in I\\\expr{f}_k & \text{if } k \in J \end{cases}
\end{equation*}
	In the case of singleton index sets \(I = \{j\}\) and collections  \(\{\prob{r}_i\}_{i \in I}, \{\expr{e}_i\}_{i \in I}\) we will simply write \(\prob{r}_j \cdot \expr{e}_j\) instead of \(\bigplus{i \in I} \prob{r}_i \cdot \expr{e}_i\)
\end{remark}
\begin{lemma}\label{lem:generalised_probabilistic_distributivity}
	Let \(I\) be a non-empty finite index set, \(\{\prob{r}_i\}_{i \in I}\) and \(\{\expr{e}_i\}_{i\in I}\) indexed collections such that \(\prob{r}_i \in [0,1]\) and \(\expr{e}_i \in \Exp\) for all \(i \in I\) and let \(\expr{f} \in \Exp \). Then,
	\begin{equation*}
		\left( \bigplus{i \in I}\prob{r}_i \cdot \expr{e}_i \right)\seq \expr{f} \equiv \bigplus{i \in I} \prob{r}_i \cdot \expr{e}_i \seq \expr{f}
	\end{equation*}
\end{lemma}
\begin{proof}
    By induction on the size of index set \(I\). For the base case, when \(I = \{j\}\) we have that
    \[
        \left( \bigplus{i \in I}\prob{r}_i \cdot \expr{e}_i \right)\seq \expr{f} \equiv \expr{e_j} \seq \expr{f}
        \equiv  \bigplus{i \in I} \prob{r}_i \cdot \expr{e}_i \seq \expr{f}
    \]
    For the induction step, consider the following
    \begin{align*}
        \left( \bigplus{i \in I}\prob{r}_i \cdot \expr{e}_i \right)\seq \expr{f} &\equiv \left(\expr{e}_j \oplus_{\prob{r}_j} \left( \bigplus{i \in I \setminus \{j\}} \frac{\prob{r}_i}{\prob{1}-\prob{r}_j} \cdot \expr{e}_i\right) \right) \seq \expr{f} \tag{\(j\in I\) }\\
        &\equiv \expr{e}_j \seq \expr{f} \oplus_{\prob{r}_j} \left( \bigplus{i \in I \setminus \{j\}} \frac{\prob{r}_i}{\prob{1}-\prob{r}_j} \cdot \expr{e}_i\right) \seq \expr{f} \tag{\acro{S6}}\\
        &\equiv \expr{e}_j \seq \expr{f} \oplus_{\prob{r}_j} \left( \bigplus{i \in I \setminus \{j\}} \frac{\prob{r}_i}{\prob{1}-\prob{r}_j} \cdot \expr{e}_i \seq \expr{f}\right) \tag{Induction hypothesis}\\
        &\equiv \bigplus{i \in I} \prob{r}_i \cdot \expr{e}_i \seq \expr{f} \tag*{\qedhere}
    \end{align*}
\end{proof}
\begin{lemma}\label{lem:collapsing_sublemma}
    Let \(I\) be a non-empty index set, \(\{\prob{r}_i\}_{i \in I}\) and \(\{\expr{e}_i\}_{i_I}\) indexed collections such that for all \(i \in I\), \(\prob{i} \in [0,1]\) and \(\expr{e}_i \in \Exp\). Let \(\expr{e}\) be an expression, such that \(\expr{e} = \expr{e}_k\) for some \(k \in I\). It holds that
    \[
    \bigplus{i \in I} \prob{r}_i \cdot \expr{e}_i \equiv \bigplus{j \in J} \prob{s}_j \cdot \expr{f}_j
    \]
    where \(J = \{\expr{e}\} \cup \{i \in I \mid \expr{e}_i \neq \expr{e}\}\), \(\{\expr{f}_j\}_{j \in J} = \begin{cases}\expr{e} &\text{if }j=\expr{e} \\ \expr{e}_i &\text{if }j=i\end{cases}\) and \(\{s_j\}_{j \in J} = \begin{cases}\sum_{\expr{e}_i = \expr{e} } \prob{r}_i & \text{if } j=\expr{e} \\ \prob{r}_i &\text{if } j=i\end{cases}\)
\end{lemma}
\begin{proof}
    Without loss of generality we can assume that \(I = \{1, \dots, n\}\) for some \(n \in \bN\) and because of \cref{lem: convex_sums_well_defined} we can safely assume that there exists \(k \in \bN\) such that for all \(i \in I\), such that \(i \leq k\) we have that \(\expr{e}_i = \expr{e}\) and for all \(i \in I\) such that \(i > k\) it holds that \(\expr{e}_i \neq \expr{e}\).

    We proceed by the induction on \(k\). If \(k = 1\), then the desired property holds immediately. For the induction step, consider the following.
    \begin{align*}
        \bigplus{i \in I} \prob{r}_i \cdot \expr{e}_i &\equiv \expr{e}_k \oplus_{\prob{r}_k} \left(\bigplus{i \in I \setminus \{k\}} \frac{\prob{r}_i}{\prob{1}-\prob{r}_k}\cdot \expr{e}_i\right) \tag{\(k \in I\)} \\
        &\equiv \expr{e} \oplus_{\prob{r}_k} \left(\bigplus{i \in I \setminus \{k\}} \frac{\prob{r}_i}{\prob{1}-\prob{r}_k}\cdot \expr{e}_i\right) \tag{\(\expr{e}=\expr{e}_k\)}
    \end{align*}
    In the inductive step, we will use the following shorthand \[\prob{n} = \sum_{\expr{e}_i=\expr{e}}^{i < k} {\prob{r}_i} \quad\quad\prob{m}= \frac{\prob{n}}{1-\prob{r}_k}\]
    Using the induction hypothesis, we can rewrite the above expression as
    \begin{align*}
        &\expr{e} \oplus{\prob{r}_k}  \left(\expr{e} \oplus_\prob{m} \left(\bigplus{i \in I \setminus \{\expr{e}, k\}} \frac{\prob{r}_i}{(\prob{1}-\prob{r}_k)\prob{(1-m)}} \cdot \expr{e}_i\right)\right) \\&\quad\quad\equiv \left(\expr{e} \oplus_{\frac{\prob{r}_k}{\prob{r}_k + \prob{m}}} \expr{e}\right) \oplus_{\prob{r}_k + \prob{n}} \left(\bigplus{i \in I \setminus \{\expr{e}, k\}} \frac{\prob{r}_i}{(\prob{1}-\prob{r}_k)\prob{(1-m)}} \cdot \expr{e}_i\right) \tag{\acro{DF11}}\\
        &\quad\quad\equiv \expr{e}\oplus_{\prob{r}_k + \prob{n}} \left(\bigplus{i \in I \setminus \{\expr{e}, k\}} \frac{\prob{r}_i}{(\prob{1}-\prob{r}_k)\prob{(1-m)}} \cdot \expr{e}_i\right) \tag{\acro{P1}}\\
        &\quad\quad\equiv \expr{e}\oplus_{\prob{r}_k + \prob{n}} \left(\bigplus{i \in I \setminus \{\expr{e}, k\}} \frac{\prob{r}_i}{\prob{1}-(\prob{r}_k + \prob{n})} \cdot \expr{e}_i\right)\\
        &\quad\quad\equiv \expr{e}\oplus_{\prob{r}_k + \prob{n}} \left(\bigplus{j \in J \setminus \{\expr{e}\}} \frac{\prob{s}_j}{\prob{1}-(\prob{r}_k + \prob{n})} \cdot \expr{f}_j\right)
        \equiv \bigplus{j \in J} \prob{s}_j \cdot \expr{f}_j \tag*{\qedhere}
    \end{align*}
\end{proof}
\begin{lemma}\label{lem:generalised_idemponency}
	Let \(I\) be a non-empty finite index set, \(\{\prob{r}_i\}_{i \in I}\) and \(\{\expr{e}_i\}_{i\in I}\) indexed collections such that \(\prob{r}_i \in [0,1]\) and \(\expr{e}_i \in \Exp\) for all \(i \in I\). Let \(E =  \bigcup_{i \in I} \{\expr{e}_i\} \). Then,
	\[
		\bigplus{i \in I} \prob{r}_i \cdot \expr{e}_i \equiv \bigplus{\expr{e} \in E} \left(\sum_{ \expr{e}_i = \expr{e}}\prob{r}_i \right) \cdot \expr{e}
	\]
\end{lemma}
\begin{proof}
    By induction on the size of \(E\). In the base case, when \(E\) is a singleton, the result holds immediately by \cref{lem:collapsing_sublemma}. For the induction step, consider the following. Let \(\expr{g} \in E\) be an arbitrary expression and let \(J = \{\expr{g}\} \cup \{i \in I \mid \expr{e}_i \neq \expr{g}\}\), \[\prob{m} = \sum_{\expr{e}_i = \expr{g} } \prob{r}_i \quad \quad \text{and}\quad\quad\{\expr{f}_j\}_{j \in J} = \begin{cases}\expr{e} &\text{if }j=\expr{g} \\ \expr{e}_i &\text{if }j=i\end{cases} \quad\quad \text{and} \quad\quad \{s_j\}_{j \in J} = \begin{cases}\prob{m} & \text{if } j=\expr{g} \\ \prob{r}_i &\text{if } j=i\end{cases}\]
    We now have that
    \begin{align*}
        \bigplus{i \in I} \prob{r}_i \cdot \expr{e}_i &\equiv \bigplus{j \in J} \prob{s}_j \cdot \expr{f}_j \tag{\cref{lem:collapsing_sublemma}}\\
        &\equiv \expr{g} \oplus_{\prob{m}} \left(\bigplus{j \in J \setminus \{\expr{g}\}} \frac{\prob{s}_j}{\prob{1-m}}\cdot\expr{f}_j\right) \\
        &\equiv \expr{g} \oplus_{\prob{m}} \left(\bigplus{\expr{e} \in E\setminus \{\expr{g}\}} \left(\sum_{ \expr{e}_i = \expr{e}}\frac{\prob{r}_i}{\prob{1-m}} \right) \cdot \expr{e}\right) \tag{Induction hypothesis} \\
        &\equiv \bigplus{\expr{e} \in E} \left(\sum_{ \expr{e}_i = \expr{e}}\prob{r}_i \right) \cdot \expr{e} \tag*{\qedhere}
    \end{align*}
\end{proof}
\begin{lemma}\label{lem:generalised_idemponency_up_to_equiv}
    Let \(I\) be a non-empty finite index set, \(\{\prob{r}_i\}_{i \in I}\) and \(\{\expr{e}_i\}_{i\in I}\) indexed collections such that \(\prob{r}_i \in [0,1]\) and \(\expr{e}_i \in \Exp\) for all \(i \in I\). Let \(E =  \bigcup_{i \in I} \{\expr{e}_i\} \). Then,
	\[
		\bigplus{i \in I} \prob{r}_i \cdot \expr{e}_i \equiv \bigplus{[\expr{e}]_\equiv \in E / \equiv } \left(\sum_{ \expr{e}_i \equiv \expr{e}}\prob{r}_i \right) \cdot \expr{e}
	\]
\end{lemma}
\begin{proof}
   A similar result to \cref{lem:collapsing_sublemma} can be obtained when replacing the strict equality with \(\equiv\). Using it, one can show the desired result by adapting \cref{lem:generalised_idemponency}.
\end{proof}

\subsection{Fundamental Theorem}\label{apx:fundamental}
\fundamentaltheorem*
\begin{proof}
    We proceed by structural induction on the construction of \(\expr{e} \in \Exp\). The base cases are trivial, hence we only show one of them.
    \begingroup
    \allowdisplaybreaks%
    \begin{itemize}
        \item[]\fbox{\(\expr{e}=\action{a}\)}
        \begin{align*}
            \action{a} &\equiv \bigsum{\alpha \in \At} \action{a} \tag{\acro{G1}}\\
            &\equiv \bigsum{\alpha \in \At} \action{a}\seq\test{\one} \tag{\acro{S2}} \\
            &\equiv \bigsum{\alpha \in \At} \left( \bigplus{d \in \supp(\partial(\action{a})_\alpha)} \partial(\action{a})_\alpha(d) \cdot \ex(d)\right) \tag{Def. of \(\partial\)}
        \end{align*}
    \end{itemize}
    Now, we can move on to inductive steps.
    \begin{itemize}
        \item[]\fbox{\(\expr{e} = \expr{f} +_\test{b} \expr{g}\)}
        \begingroup%
        \allowdisplaybreaks%
        {
        \small
        \begin{align*}
            \expr{e} &\equiv \expr{f} +_\test{b} \expr{g} \\
            &\equiv \left( \bigsum{\alpha \in \At} \left(\bigplus{d \in \supp(\partial(\expr{f})_\alpha)}\partial(\expr{f})_\alpha(d) \cdot \ex(d)\right)\right) \\
            &\hspace{5em} +_\test{b} \left( \bigsum{\alpha \in \At} \left(\bigplus{d \in \supp(\partial(\expr{g})_\alpha)}\partial(\expr{g})_\alpha(d) \cdot \ex(d)\right)\right) \tag{Ind. hypothesis}\\
            &\equiv \test{b}\seq\left( \bigsum{\alpha \in \At} \left(\bigplus{d \in \supp(\partial(\expr{f})_\alpha)}\partial(\expr{f})_\alpha(d) \cdot \ex(d)\right)\right)\\&\hspace{5em} +_\test{b} \bar{\test{b}}\seq\left( \bigsum{\alpha \in \At} \left(\bigplus{d \in \supp(\partial(\expr{g})_\alpha)}\partial(\expr{g})_\alpha(d) \cdot \ex(d)\right)\right) \tag{\acro{G3} and \acro{G2}}\\
            &\equiv \test{b}\seq\test{b}\seq\left( \bigsum{\alpha \in \At} \left(\bigplus{d \in \supp(\partial(\expr{f})_\alpha)}\partial(\expr{f})_\alpha(d) \cdot \ex(d)\right)\right)\\&\hspace{5em} +_\test{b} \bar{\test{b}}\seq\bar{\test{b}}\seq\left( \bigsum{\alpha \in \At} \left(\bigplus{d \in \supp(\partial(\expr{g})_\alpha)}\partial(\expr{g})_\alpha(d) \cdot \ex(d)\right)\right) \tag{Boolean algebra and \acro{S8}}\\
            &\equiv \test{b}\seq\left( \bigsum{\alpha \bleq \test{b}} \left(\bigplus{d \in \supp(\partial(\expr{f})_\alpha)}\partial(\expr{f})_\alpha(d) \cdot \ex(d)\right)\right)\\&\hspace{5em} +_\test{b} \bar{\test{b}}\seq\left( \bigsum{\alpha \bleq \bar{\test{b}}} \left(\bigplus{d \in \supp(\partial(\expr{g})_\alpha)}\partial(\expr{g})_\alpha(d) \cdot \ex(d)\right)\right) \tag{\cref{lem:generalised_branch_selection}}\\
            &\equiv \test{b}\seq\left( \bigsum{\alpha \bleq \test{b}} \left(\bigplus{d \in \supp(\partial(\expr{f} +_\test{b} \expr{g})_\alpha)}\partial(\expr{f} +_\test{b} \expr{g})_\alpha(d) \cdot \ex(d)\right)\right) \\&\hspace{5em}+_\test{b}\bar{\test{b}}\seq\left( \bigsum{\alpha \bleq \bar{\test{b}}} \left(\bigplus{d \in \supp(\partial(\expr{f} +_\test{b} \expr{g})_\alpha)}\partial(\expr{f} +_\test{b }\expr{g})_\alpha(d) \cdot \ex(d)\right)\right) \tag{Def. of \(\partial\)}\\
            &\equiv \left( \bigsum{\alpha \in \At} \left(\bigplus{d \in \supp(\partial(\expr{f} +_\test{b} \expr{g})_\alpha)}\partial(\expr{f} +_\test{b} \expr{g})_\alpha(d) \cdot \ex(d)\right)\right) \\&\hspace{5em}+_\test{b} \left( \bigsum{\alpha \in \At} \left(\bigplus{d \in \supp(\partial(\expr{f} +_\test{b} \expr{g})_\alpha)}\partial(\expr{f} +_\test{b}\expr{g})_\alpha(d) \cdot \ex(d)\right)\right) \tag{\cref{lem:generalised_branch_selection}}\\
            &\equiv  \bigsum{\alpha \in \At} \left(\bigplus{d \in \supp(\partial(\expr{f} +_\test{b} \expr{g})_\alpha)}\partial(\expr{f} +_\test{b} \expr{g})_\alpha(d) \cdot \ex(d)\right) \tag{\acro{G4}}
        \end{align*}
        }
        \endgroup%
    \item[]\fbox{\(\expr{e} = \expr{f} \oplus_\prob{r} \expr{g}\)}
    {
    \small
    \begin{align*}
        \expr{e} &\equiv \expr{f} \oplus_\prob{r} \expr{g} \\
         &\equiv \left( \bigsum{\alpha \in \At} \left(\bigplus{d \in \supp(\partial(\expr{f})_\alpha)}\partial(\expr{f})_\alpha(d) \cdot \ex(d)\right)\right)\\&\hspace{5em} \oplus_\prob{r} \left( \bigsum{\alpha \in \At} \left(\bigplus{d \in \supp(\partial(\expr{g})_\alpha)}\partial(\expr{g})_\alpha(d) \cdot \ex(d)\right)\right) \tag{Ind. hypothesis}\\
         &\equiv \left( \bigsum{\alpha \in \At} \left(\bigplus{d \in \supp(\partial(\expr{f})_\alpha)}\partial(\expr{f})_\alpha(d) \cdot \ex(d)\right) \oplus_\prob{r} \left(\bigplus{d \in \supp(\partial(\expr{g})_\alpha)}\partial(\expr{g})_\alpha(d) \cdot \ex(d)\right)\right)\tag{\cref{lem:generalised_exchange}}
    \end{align*}
    }%
    Observe that for each atom \(\alpha \in \At\) the expression inside the generalised guarded sum can be rearranged into a single generalised convex sum, by using the axiom \acro{P4}.
	For each \(\alpha \in \At\) we will define the index set \(I_\alpha = \supp(\partial(\expr{e})_\alpha) +\supp(\partial(\expr{f})_\alpha)\), where we use \(+\) to denote the coproduct. Let
	\[
		\prob{r}_i = \begin{cases}
			\prob{r}\partial(\expr{f})_\alpha(i) & \text{if } i \in \supp(\partial(\expr{f})_\alpha)\\
			\prob{(1-r)}\partial(\expr{g})_\alpha(i) & \text{if } i \in \supp(\partial(\expr{g})_\alpha)
		\end{cases}
	\]
	Which allows us to obtain the following
	\[
		\expr{f} \oplus_\prob{r} \expr{g} \equiv \bigsum{\alpha \in \At} \left( \bigplus{i \in I} \prob{r}_i \cdot \ex(i)\right)
	\]
	We now apply the Lemma~\ref{lem:generalised_idemponency} to remove the duplicate elements in the convex sums by adding up their probabilities. This yields the following
	\begin{equation*}
		\bigsum{\alpha \in \At} \left(\bigplus{d \in \supp(\partial(\expr{f})_\alpha) \cup \supp(\partial(\expr{g})_\alpha)} \left(\prob{r}\partial(\expr{f})_\alpha(d) + \prob{(1-r)}\partial(\expr{g})_\alpha(d)\right) \cdot \ex(d) \right)
	\end{equation*}
	which is precisely the same as
	\begin{equation*}
		\bigsum{\alpha \in \At} \left( \bigplus{d \in \supp(\partial(\expr{f} \oplus_\prob{r} \expr{g})_\alpha)}\partial(\expr{f} \oplus_\prob{r} \expr{g})_\alpha(d) \cdot \ex(d)\right)
	\end{equation*}
    \item[]\fbox{\(\expr{e}=\expr{f} \seq \expr{g}\)}

    {
    \small
    \begin{align*}
        \expr{e} &\equiv \expr{f}\seq \expr{g} \\
        &\equiv \left(\bigsum{\alpha \in \At} \left(\bigplus{d \in \supp(\partial(\expr{f})_\alpha)} \partial(\expr{f})_\alpha(d) \cdot \ex(d)\right)\right) \seq \expr{g} \tag{Induction hypothesis}\\
        &\equiv \left(\bigsum{\alpha \in \At} \left(\bigplus{d \in \supp(\partial(\expr{f})_\alpha)} \partial(\expr{f})_\alpha(d) \cdot \ex(d)\right)\seq \expr{g}\right) \tag{\cref{lem:generalised_guarded_distributivity}}\\
        &\equiv \left(\bigsum{\alpha \in \At} \left(\bigplus{d \in \supp(\partial(\expr{f})_\alpha)} \partial(\expr{f})_\alpha(d) \cdot \ex(d)\seq \expr{g}\right)\right) \tag{\cref{lem:generalised_probabilistic_distributivity}}
    \end{align*}
    }
    We can unroll the generalised convex sum, to obtain the following
	\begin{align*}
		&\bigsum{\alpha \in \At} \left( \partial(\expr{f})_\alpha(\accept) \cdot \test{\one} \seq \expr{g} \boplus \bigplus{d \in \V \cup \{\reject\}} \partial(\expr{f})_\alpha(d)\cdot \ex(d) \seq \expr{g} \right.\\&\hspace{5em}\boplus \left.\bigplus{d \in \supp(\partial(\expr{f})_\alpha) \cap \Act \times \Exp} \partial(\expr{f})_\alpha(d) \cdot \ex(d) \seq \expr{g}\right)
    \end{align*}
    Now, apply \acro{S2}, \acro{S4} and \acro{S7}, as well as \acro{DF7} to obtain
    \begin{align*}
		&\bigsum{\alpha \in \At} \left( \partial(\expr{f})_\alpha(\accept) \cdot \test{\alpha} \seq \expr{g} \boplus \bigplus{d \in \V \cup \{\reject\}} \partial(\expr{f})_\alpha(d)\cdot \ex(d)\right.\\&\hspace{5em}\left. \boplus \bigplus{d \in \supp(\partial(\expr{f})_\alpha) \cap \Act \times \Exp} \partial(\expr{f})_\alpha(d) \cdot \ex(d) \seq \expr{g}\right)
	\end{align*}
    By the induction hypothesis and \acro{DF3}, we can write \(\test{\alpha} \seq \expr{g}\) for each \(\alpha \in \At\) as
	\[
		\bigplus{d \in \supp(\partial(\expr{g})_\alpha)} \partial(\expr{g})_\alpha(d) \cdot \ex(d)
	\]
	Hence, by substituting it to the previous equation and using \acro{P4} and \acro{S3} we can rewrite the whole expression as a single generalised convex sum. For each \(\alpha \in \At\) let \(I_\alpha\) be an index set defined as
	\[
		I_\alpha = \supp(\partial(\expr{g})_\alpha) + \left(\supp(\partial(\expr{f})_\alpha) \cap \left(\{\reject\} \cup \V \right)\right) + \{(\action{a},\expr{f'}\seq \expr{g}) \mid (\action{a},\expr{f'}) \in \supp(\partial(\expr{f})_\alpha) \cap \Act \times \Exp\}
	\]
	For each \(\alpha\), let \(\{\prob{r}_i\}_{i \in I_\alpha}\) be an indexed collection such that \(\prob{r}_i \in [0,1]\) for each \(i \in I_\alpha\) given by the following
	\[
		\prob{r}_i = \begin{cases}
			\partial(\expr{f})_\alpha(\accept) \partial(\expr{g})_\alpha(i) & \text{if } i \in \supp(\partial(\expr{g})_\alpha)\\
			\partial(\expr{f})_\alpha(i) & \text{if } i \in \supp(\partial(\expr{f})_\alpha)\cap \left( \{\reject\} \cup \V\right)\\
			\partial(\expr{f})_\alpha (\action{a},\expr{f'}) & \text{if } 	i=(\action{a},\expr{f'} \seq \expr{g}) \text{ and } (\action{a},\expr{f'}) \in \supp(\partial(\expr{f})_\alpha) \cap \Act \times \Exp
		\end{cases}
	\]
	This allows us to write the following
	\[
		\expr{f} \seq \expr{g} \equiv \bigsum{\alpha \in \At} \left(\bigplus{i \in I_\alpha} \prob{r}_i \cdot \ex(i)\right)
	\]
	Applying the~\ref{lem:generalised_idemponency} and summing probabilities of the duplicated elements yields precisely the desired result
	\[
		\expr{f} \seq \expr{g} \equiv \bigsum{\alpha \in \At} \left( \bigsum{d \in \supp(\partial(\expr{f} \seq \expr{g})_\alpha)} \partial(\expr{f} \seq \expr{g})_\alpha(d) \cdot \ex(d) \right)
	\]
    \item[]\fbox{\(\expr{e} \equiv \expr{f}^{(\test{b})}\)}
        We want to show that
        \[
        \expr{f}^{(\test{b})} \equiv \bigsum{\alpha \in \At} \left(\bigplus{d \in \supp\left(\partial\left(\expr{f}^{{(\test{b})}}\right)_\alpha\right)} \partial\left(\expr{f}^{(\test{b})}\right)_\alpha(d) \cdot \ex(d)\right)
        \]
        Because of \cref{lem:generalised_guardedness} and axiom \acro{G1} we can write the above as
        \[
        \bigsum{\alpha \in \At} \test{\alpha} \seq \expr{f}^{(\test{b})} \equiv \bigsum{\alpha \in \At} \left(\test{\alpha} \seq \left(\bigplus{d \in \supp\left(\partial\left(\expr{f}^{{(\test{b})}}\right)_\alpha\right)} \partial\left(\expr{f}^{(\test{b})}\right)_\alpha(d) \cdot \ex(d)\right)\right)
        \]
        Hence, it suffices to show that for all \(\alpha \in \At\) we have that
        \[
        \test{\alpha}\seq\expr{f}^{(\test{b})} \equiv \test{\alpha} \seq \left(\bigplus{d \in \supp\left(\partial\left(\expr{f}^{{(\test{b})}}\right)_\alpha\right)} \partial\left(\expr{f}^{(\test{b})}\right)_\alpha(d) \cdot \ex(d)\right)
        \]
        First, consider \(\gamma \bleq \bar{\test{b}}\). We have that
        \begin{align*}
            \test{\gamma} \seq \expr{f}^{(\test{b})} &\equiv \test{\gamma} \seq \left( \expr{f}\seq\expr{f}^{(\test{b})} +_\test{b} \test{\one}\right) \tag{\acro{L1}} \\
            &\equiv  \left( \expr{f}\seq\expr{f}^{(\test{b})} +_\test{b} \test{\one}\right) +_\test{\gamma} \test{\one} \tag{\acro{DF2}} \\
            &\equiv \expr{f}\seq\expr{f}^{(\test{b})} +_{\test{\gamma}\test{b}} (\test{\one} +_\test{\gamma} \test{\zero}) \tag{\acro{G4}}\\
            &\equiv  \expr{f}\seq\expr{f}^{(\test{b})} +_{\test{\zero}} (\test{\one} +_\test{\gamma} \test{\zero}) \tag{\(\gamma \bleq \bar{\test{b}}\)}\\
            &\equiv (\test{\one} +_\test{\gamma} \test{\zero}) +_\test{\one} \expr{f}\seq\expr{f}^{(\test{b})} \tag{\acro{G3}}\\
            &\equiv (\test{\one} +_\test{\gamma} \test{\zero}) \tag{\acro{DF12}}\\
            &\equiv \test{\gamma}\seq\test{\one} \tag{\acro{DF2}} \\
            &\equiv \test{\gamma} \seq \left(\bigplus{d \in \supp\left(\partial\left(\expr{f}^{{(\test{b})}}\right)_\gamma\right)} \partial\left(\expr{f}^{(\test{b})}\right)_\gamma(d) \cdot \ex(d)\right) \tag{Def. of \(\partial\)}
        \end{align*}
        Now, consider \(\gamma \bleq \test{c}\) such that \(\partial(\expr{f})_\gamma(\accept)=\prob{1}\). Because of the induction hypothesis we have that
        \begin{align*}
            \expr{f} &\equiv \bigsum{\alpha \in \At} \left(\bigplus{d \in \supp (\partial\left(\expr{f}\right)_\alpha)} \partial(\expr{f})_\alpha(d) \cdot \ex(d) \right) \tag{Ind. hypothesis}\\
            &\equiv \test{\one } +_\test{\gamma} \bigsum{\alpha \in \At\setminus \{\gamma\}} \left(\bigplus{d \in \supp(\partial\left(\expr{f}\right)_\alpha)} \partial(\expr{f})_\alpha(d) \cdot \ex(d) \right)\\
            &\equiv  \bigsum{\alpha \in \At\setminus \{\gamma\}} \left(\bigplus{d \in \supp(\partial\left(\expr{f}\right)_\alpha)} \partial(\expr{f})_\alpha(d) \cdot \ex(d) \right) +_{\bar{\test{\gamma}}} \test{\one} \tag{\acro{G3}}
        \end{align*}
        For the sake of simplicity, we will write the above term as
        \(
        \expr{f} \equiv \expr{m} +_{\bar{\test{\gamma}}} \test{\one}
        \). Now, consider the following
        \begin{align*}
            \test{\gamma}\seq\expr{f}^{(\test{b})} &\equiv \test{\gamma}\seq(\expr{m} +_{\bar{\test{\gamma}}} \test{\one}) \\
            &\equiv \test{\gamma}\seq(\bar{\test{\gamma}}\seq\expr{m})^{(\test{b})} \tag{\acro{L3}} \\
            &\equiv \test{\gamma}\seq((\bar{\test{\gamma}}\seq\expr{m})\seq(\bar{\test{\gamma}}\seq\expr{m})^{(\test{b})} +_\test{b} \test{\one}) \tag{\acro{L1}} \\
            &\equiv \test{\gamma}\seq(\bar{\test{\gamma}}\seq(\expr{m}\seq(\bar{\test{\gamma}}\seq\expr{m})^{(\test{b})}) +_\test{b} \test{\one}) \tag{\acro{S3}} \\
            &\equiv \test{\gamma}\seq(\test{\gamma}\seq\bar{\test{\gamma}}\seq(\expr{m}\seq(\bar{\test{\gamma}}\seq\expr{m})^{(\test{b})}) +_\test{b} \test{\one}) \tag{\acro{DF7}} \\
            &\equiv \test{\gamma}\seq(\test{\zero}\seq(\expr{m}\seq(\bar{\test{\gamma}}\seq\expr{m})^{(\test{b})}) +_\test{b} \test{\one}) \tag{\acro{S8}} \\
            &\equiv \test{\gamma}\seq(\test{\zero} +_\test{b} \test{\one}) \tag{\acro{S4}} \\
            &\equiv (\test{\zero} +_\test{b} \test{\one}) +_\test{\gamma} \test{\zero} \tag{\acro{DF2}} \\
            &\equiv \test{\zero} +_{\test{\gamma}\test{b}} (\test{\one} +_\test{\gamma} \test{\zero}) \tag{\acro{G4}} \\
            &\equiv \test{\zero} +_{\test{\gamma}} (\test{\one} +_\test{\gamma} \test{\zero}) \tag{Boolean algebra} \\
            &\equiv (\test{\one} +_\test{\gamma} \test{\zero}) +_{\bar{\test{\gamma}}} \test{\zero} \tag{\acro{G3}} \\
            &\equiv \bar{\test{\gamma}}\seq (\test{\one} +_\test{\gamma} \test{\zero}) \tag{\acro{DF2}} \\
            &\equiv \bar{\test{\gamma}}\seq (\test{\zero} +_{\bar{\test{\gamma}}} \test{\one}) \tag{\acro{G3}} \\
            &\equiv \bar{\test{\gamma}} \seq \test{\zero} \tag{\acro{DF3}} \\
            &\equiv \test{\zero} +_{\bar{\test{\gamma}}} \test{\zero} \tag{\acro{DF2}} \\
            &\equiv  \test{\zero} +_{{\test{\gamma}}} \test{\zero} \tag{\acro{G3}}\\
            &\equiv \test{\gamma}\seq\test{\zero} \tag{\acro{DF2}} \\
            &\equiv \test{\gamma} \seq \left(\bigplus{d \in \supp\left(\partial\left(\expr{f}^{{(\test{b})}}\right)_\gamma\right)} \partial\left(\expr{f}^{(\test{b})}\right)_\gamma(d) \cdot \ex(d)\right) \tag{Def. of \(\partial\)}
        \end{align*}

        Finally consider \(\gamma \bleq \test{b}\) such that \(\partial(\expr{f})_\alpha(\accept)<\prob{1}\). Because of the induction hypothesis we have that
        \begin{align*}
            \expr{f} &\equiv \bigsum{\alpha \in \At} \left(\bigplus{d \in \supp(\partial \left(\expr{f}\right)_\alpha)} \partial(\expr{f})_\alpha(d) \cdot \ex(d) \right) \tag{Ind. hypothesis}\\
            &\equiv \left(\bigplus{d \in \supp(\partial\left(\expr{f}\right)_\gamma)} \partial(\expr{f})_\alpha(d) \cdot \ex(d) \right) +_\test{\gamma} \bigsum{\alpha \in \At} \left(\bigplus{d \in \supp (\partial\left(\expr{f}\right)_\alpha)} \partial(\expr{f})_\alpha(d) \cdot \ex(d) \right)\\
            &\equiv \left(\test{\one} \oplus_{\partial(\expr{f})_\gamma(\accept)} \bigplus{d \in \supp(\partial\left(\expr{f}\right)_\gamma)} \frac{\partial(\expr{f})_\alpha(d)}{1-\partial(\expr{f})_\gamma(\accept)} \cdot \ex(d) \right) \\
            &\hspace{5em} +_\test{\gamma} \bigsum{\alpha \in \At} \left(\bigplus{d \in \supp(\partial\left(\expr{f}\right)_\alpha)} \partial(\expr{f})_\alpha(d) \cdot \ex(d) \right) \\
            &\equiv \left(\bigplus{d \in \supp(\partial\left(\expr{f}\right)_\gamma)} \frac{\partial(\expr{f})_\alpha(d)}{1-\partial(\expr{f})_\gamma(\accept)} \cdot \ex(d) \oplus_{\prob{1}-\partial(\expr{f})_\gamma(\accept)}\test{\one}  \right) \\&\hspace{5em}+_\test{\gamma} \bigsum{\alpha \in \At} \left(\bigplus{d \in \supp(\partial\left(\expr{f}\right)_\alpha)} \partial(\expr{f})_\alpha(d) \cdot \ex(d) \right) \tag{\acro{P3}}
        \end{align*}
        We can use the axiom \acro{L5} and obtain the following
        \begin{align*}
        \test{\gamma}\seq\expr{f}^{(\test{b})} &\equiv \test{\gamma} \seq \left(\left(\bigplus{d \in \supp(\partial\left(\expr{f}\right)_\gamma)} \frac{\partial(\expr{f})_\alpha(d)}{1-\partial(\expr{f})_\gamma(\accept)} \cdot \ex(d)\right)\seq \expr{f}^{(\test{b})} +_\test{b} \test{\one} \right) \\
        &\equiv \test{\gamma} \seq \left(\left(\bigplus{d \in \supp(\partial\left(\expr{f}\right)_\gamma)} \frac{\partial(\expr{f})_\alpha(d)}{1-\partial(\expr{f})_\gamma(\accept)} \cdot \ex(d)\seq \expr{f}^{(\test{b})}\right) +_\test{b} \test{\one} \right) \tag{\cref{lem:generalised_probabilistic_distributivity}}\\
        &\equiv \left(\left(\bigplus{d \in \supp(\partial \left(\expr{f}\right)_\gamma)} \frac{\partial(\expr{f})_\alpha(d)}{1-\partial(\expr{f})_\gamma(\accept)} \cdot \ex(d)\seq \expr{f}^{(\test{b})}\right) +_\test{b} \test{\one} \right) +_\test{\gamma} \test{\zero} \tag{\acro{DF2}} \\
        &\equiv \left(\bigplus{d \in \supp(\partial \left(\expr{f}\right)_\gamma)} \frac{\partial(\expr{f})_\alpha(d)}{1-\partial(\expr{f})_\gamma(\accept)} \cdot \ex(d)\seq \expr{f}^{(\test{b})}\right) +_{\test{\gamma}\test{b}} (\test{\one} +_\test{\gamma} \test{\zero}) \tag{\acro{G4}} \\
        &\equiv \left(\bigplus{d \in \supp(\partial\left(\expr{f}\right)_\gamma)} \frac{\partial(\expr{f})_\alpha(d)}{1-\partial(\expr{f})_\gamma(\accept)} \cdot \ex(d)\seq \expr{f}^{(\test{b})}\right) +_{\test{\gamma}} (\test{\one} +_\test{\gamma} \test{\zero}) \tag{Boolean algebra} \\
        &\equiv (\test{\zero} +_{\bar{\test{\gamma}}} \test{\one}) +_{\bar{\test{\gamma}}} \left(\bigplus{d \in \supp(\partial\left(\expr{f}\right)_\gamma)} \frac{\partial(\expr{f})_\alpha(d)}{1-\partial(\expr{f})_\gamma(\accept)} \cdot \ex(d)\seq \expr{f}^{(\test{b})}\right) \tag{\acro{G3}} \\
        &\equiv \bar{\test{\gamma}}\seq(\test{\zero} +_{\bar{\test{\gamma}}} \test{\one}) +_{\bar{\test{\gamma}}} \left(\bigplus{d \in \supp(\partial\left(\expr{f}\right)_\gamma)} \frac{\partial(\expr{f})_\alpha(d)}{1-\partial(\expr{f})_\gamma(\accept)} \cdot \ex(d)\seq \expr{f}^{(\test{b})}\right) \tag{\acro{G2}} \\
        &\equiv \bar{\test{\gamma}}\seq\test{\zero} +_{\bar{\test{\gamma}}} \left(\bigplus{d \in \supp(\partial\left(\expr{f}\right)_\gamma)} \frac{\partial(\expr{f})_\alpha(d)}{1-\partial(\expr{f})_\gamma(\accept)} \cdot \ex(d)\seq \expr{f}^{(\test{b})}\right) \tag{\acro{DF3}} \\
        &\equiv (\test{\zero} +_{\bar{\test{\gamma}}} \test{\zero}) +_{\bar{\test{\gamma}}} \left(\bigplus{d \in \supp(\partial\left(\expr{f}\right)_\gamma)} \frac{\partial(\expr{f})_\alpha(d)}{1-\partial(\expr{f})_\gamma(\accept)} \cdot \ex(d)\seq \expr{f}^{(\test{b})}\right) \tag{\acro{DF2}} \\
        &\equiv \test{\zero}  +_{\bar{\test{\gamma}}} \left(\bigplus{d \in \supp(\partial\left(\expr{f}\right)_\gamma)} \frac{\partial(\expr{f})_\alpha(d)}{1-\partial(\expr{f})_\gamma(\accept)} \cdot \ex(d)\seq \expr{f}^{(\test{b})}\right) \tag{\acro{G1}} \\
        &\equiv \left(\bigplus{d \in \supp(\partial\left(\expr{f}\right)_\gamma)} \frac{\partial(\expr{f})_\alpha(d)}{1-\partial(\expr{f})_\gamma(\accept)} \cdot \ex(d)\seq \expr{f}^{(\test{b})}\right) +_{\test{\gamma}} \test{\zero} \tag{\acro{G3}} \\
        &\equiv \test{\gamma} \seq \left(\bigplus{d \in \supp(\partial\left(\expr{f}\right)_\gamma)} \frac{\partial(\expr{f})_\alpha(d)}{1-\partial(\expr{f})_\gamma(\accept)} \cdot \ex(d)\seq \expr{f}^{(\test{b})}\right) \tag{\acro{DF2}}
        \end{align*}
        We can unroll the generalised convex sum to obtain the following
        \begin{align*}
		&\test{\gamma} \seq \left(\bigplus{d \in \V \cup \{\reject\}} \frac{\partial(\expr{f})_\alpha(d)}{{1-\partial(\expr{f})_\gamma(\accept)}}\cdot \ex(d) \seq \expr{f}^{(\test{b})}\right.\\&\hspace{5em} \left. \boplus \bigplus{d \in \supp(\partial(\expr{f})_\alpha) \cap \Act \times \Exp} \frac{\partial(\expr{f})_\alpha(d)}{1-\partial(\expr{f})_\gamma(\accept)} \cdot \ex(d) \seq \expr{f}^{(\test{b})}\right)
	    \end{align*}
        Now, we can use axioms \acro{S4} and \acro{S7} to obtain
        \begin{align*}
		&\test{\gamma} \seq \left(\bigplus{d \in \V \cup \{\reject\}} \frac{\partial(\expr{f})_\alpha(d)}{{1-\partial(\expr{f})_\gamma(\accept)}}\cdot \ex(d)\right. \left. \boplus \bigplus{\substack{d \in \supp(\partial(\expr{f})_\alpha) \\ {} \cap \Act \times \Exp}} \frac{\partial(\expr{f})_\alpha(d)}{1-\partial(\expr{f})_\gamma(\accept)} \cdot \ex(d) \seq \expr{f}^{(\test{b})}\right)
	    \end{align*}
        which is precisely the same as
        \[
            \test{\gamma} \seq \left(\bigplus{d \in \supp\left(\partial\left(\expr{f}^{{(\test{b})}}\right)_\gamma\right)} \partial\left(\expr{f}^{(\test{b})}\right)_\gamma(d) \cdot \ex(d)\right)
        \]

        \item[]\fbox{\(\expr{e}\equiv \expr{f}^{[\prob{r}]}\)}
        First, note that we can safely assume that \(\prob{r}<\prob{1}\). If \(\prob{r}=\prob{1}\), then we can use axiom \acro{L4} to obtain \(\expr{f}^{[\prob{1}]}\equiv \expr{f}^{(\test{\one})}\) and treat it as guarded loop. Similarly to the case above, we show that for all \(\alpha \in \At\) we have
        \[
        \alpha \seq \expr{f}^{[\prob{r}]} \equiv \alpha \seq q\left(\bigplus{d \in \supp\left(\partial\left(\expr{f}^{[\prob{r}]}\right)_\alpha\right)}\partial\left(\expr{f}^{[\prob{r}]}\right)_\alpha(d) \cdot \ex(d)\right)
        \]

        Now, consider \(\gamma \in \At\), such that \(\partial(\expr{f})_\gamma(\accept)=\prob{1}\). In such a case, we have that
        \begin{align*}
            \expr{f} &\equiv \bigsum{\alpha \in \At} \left(\bigplus{d \in \supp(\partial\left(\expr{f}\right)_\alpha)} \partial(\expr{f})_\alpha(d) \cdot \ex(d) \right) \tag{Ind. hypothesis}\\
            &\equiv \test{\one } +_\test{\gamma} \bigsum{\alpha \in \At\setminus \{\gamma\}} \left(\bigplus{d \in \supp(\partial\left(\expr{f}\right)_\alpha)} \partial(\expr{f})_\alpha(d) \cdot \ex(d) \right)\\
            &\equiv (\test{\one } \oplus_\prob{1} \test{\zero}) +_\test{\gamma} \bigsum{\alpha \in \At\setminus \{\gamma\}} \left(\bigplus{d \in \supp(\partial\left(\expr{f}\right)_\alpha)} \partial(\expr{f})_\alpha(d) \cdot \ex(d) \right) \tag{\acro{P2}}\\
            &\equiv (\test{\zero } \oplus_\prob{o} \test{\one}) +_\test{\gamma} \bigsum{\alpha \in \At\setminus \{\gamma\}} \left(\bigplus{d \in \supp(\partial\left(\expr{f}\right)_\alpha)} \partial(\expr{f})_\alpha(d) \cdot \ex(d) \right) \tag{\acro{P3}}
        \end{align*}
        Now, we can apply the axiom \acro{L6} and obtain the following
        \begin{align*}
            \test{\gamma} \seq \expr{f}^{[\prob{r}]} &\equiv \test{\gamma}\seq(\test{\zero}\seq\expr{f}^{[\prob{r}]} \oplus_\prob{0} \test{\one}) \\
            &\equiv \test{\gamma}\seq(\test{\one} \oplus_\prob{1} \test{\zero}\seq\expr{f}^{[\prob{r}]}) \tag{\acro{P3}} \\
            &\equiv \test{\gamma}\seq\test{\one} \tag{\acro{P2}} \\
            &\equiv \test{\gamma}\seq\left(\bigplus{d \in \supp\left(\partial\left(\expr{f}^{[\prob{r}]}\right)_\gamma\right)}\partial\left(\expr{f}^{[\prob{r}]}\right)_\gamma(d) \cdot \ex(d)\right) \tag{Def. of \(\partial\)}
        \end{align*}

        Now, consider the case of \(\gamma \in \At\), such that \(\partial(\expr{f})_\gamma(\accept)<\prob{1}\).
        \begin{align*}
		\expr{f} &\equiv \bigsum{\alpha \in \At} \left( \bigplus{d \in \supp(\partial(\expr{f})_\alpha)} \partial(\expr{f})_\alpha(d) \cdot \ex(d)\right)  \tag{Induction hypothesis}\\
		&\equiv \left( \bigplus{d \in \supp(\partial(\expr{f})_\gamma)} \partial(\expr{f})_\gamma(d) \cdot \ex(d) \right) \\
        &\hspace{5em}+_{\test{\gamma}} \left( \bigsum{\alpha \in \At \setminus \{\gamma\}} \left( \bigplus{d \in \supp(\partial(\expr{f})_\alpha)} \partial(\expr{f})_\alpha(d) \cdot \ex(d) \right)\right)\\
        &\equiv \left(\test{\one}\oplus_{\prob{1}-\partial(\expr{f})_\gamma(\accept)}\left( \bigplus{d \in \supp(\partial(\expr{f})_\gamma) \setminus \{\accept\}} \frac{\partial(\expr{f})_\gamma(d)}{\prob{1}-\partial(\expr{f})_\gamma(\accept)} \cdot \ex(d) \right)  \right)  \\
        &\hspace{5em} +_\test{\gamma} \left( \bigsum{\alpha \in \At \setminus \{\gamma\}} \left( \bigplus{d \in \supp(\partial(\expr{f})_\alpha)} \partial(\expr{f})_\alpha(d) \cdot \ex(d) \right)\right)\\
		&\equiv \left(\left( \bigplus{d \in \supp(\partial(\expr{f})_\gamma) \setminus \{\accept\}} \frac{\partial(\expr{f})_\gamma(d)}{\prob{1}-\partial(\expr{f})_\gamma(\accept)} \cdot \ex(d) \right) \oplus_{\prob{1}-\partial(\expr{f})_\gamma(\accept)} \test{\one} \right) \\
        &\hspace{5em} +_\test{\gamma} \left( \bigsum{\alpha \in \At \setminus \{\gamma\}} \left( \bigplus{d \in \supp(\partial(\expr{f})_\alpha)} \partial(\expr{f})_\alpha(d) \cdot \ex(d) \right)\right) \tag{\acro{P3}}
	    \end{align*}
        We can now apply \acro{L5} and obtain the following
        \begin{align*}
            \test{\gamma} \seq \expr{f}^{[\prob{r}]} &\equiv \test{\gamma}\seq\left(\left( \bigplus{d \in \supp(\partial(\expr{f})_\gamma) \setminus \{\accept\}} \frac{\partial(\expr{f})_\gamma(d)}{\prob{1}-\partial(\expr{f})_\gamma(\accept)} \cdot \ex(d) \right) \seq \expr{f}^{[\prob{r}]} \oplus_{\frac{\prob{r(}\prob{1}-\partial(\expr{f})_\gamma(\accept))}{\prob{1}-\prob{r}\partial(\expr{f})_\gamma(\accept)}} \test{\one}\right)\\ 
            &\equiv \test{\gamma}\seq\left(\left( \bigplus{d \in \supp(\partial(\expr{f})_\gamma) \setminus \{\accept\}} \frac{\partial(\expr{f})_\gamma(d)}{\prob{1}-\partial(\expr{f})_\gamma(\accept)} \cdot \ex(d) \seq \expr{f}^{[\prob{r}]} \right) \oplus_{\frac{\prob{r(}\prob{1}-\partial(\expr{f})_\gamma(\accept))}{\prob{1}-\prob{r}\partial(\expr{f})_\gamma(\accept)}} \test{\one}\right) \tag{\cref{lem:generalised_probabilistic_distributivity}} 
        \end{align*}
	    We can unroll the convex sum and apply the \acro{S3}, \acro{S4} and \acro{S7} axioms to obtain the following
        {
	   \begin{align*}
		 &\test{\gamma}\seq\left(\left( \bigplus{d \in \supp(\partial(\expr{f})_\gamma) \cap \left(\V \cup \{\reject\}\right)} \frac{\partial(\expr{f})_\gamma(d)}{\prob{1}-\partial(\expr{f})_\gamma(\accept) } \cdot \ex(d) \right.\right.\\&\hspace{5em}\left.\left.\boplus \bigplus{d \in \supp(\partial(\expr{f})_\gamma) \cap \Act \times \Exp} \frac{\partial(\expr{f})_\gamma(d)}{\prob{1}-\partial(\expr{f})_\gamma(\accept)} \cdot \ex(d) \seq \expr{f}^{[\prob{r}]} \right) \oplus_{\frac{\prob{r}(\prob{1}-\partial(\expr{f})_\gamma(\accept))}{\prob{1}-\prob{r}\partial(\expr{f})_\gamma(\accept)}} \test{\one}\right)
       \end{align*}
        }
	   We can write massage that expression back into a generalised convex sum. Let \(I_\gamma\) be an index set defined in the following way.
	   \begin{equation*}
		  I_{\gamma} = \{\accept \} + \left(\supp(\partial(\expr{f})_\gamma) \cap (\V \cup \{\reject\})\right) + \{(\action{a},\expr{f'} \seq \expr{f}^{[\prob{r}]}) \mid (\action{a},\expr{f'}) \in \supp(\partial(\expr{f})_\gamma)\}
	   \end{equation*}
	   Now, let \(\{\prob{r}_i\}_{i \in I_\gamma}\) be an indexed collection such that \(\prob{r}_i \in [0,1] \) for each \(i \in I_\gamma\). We define such collections in the following way
	   \begin{equation*}
		\prob{r}_i = \begin{cases}
			\frac{\prob{1-r}}{\prob{1}-\partial(\expr{f})_\gamma(\accept)} & \text{if } i = \accept \\
			\frac{\partial(\expr{f})_\gamma(i)}{\prob{1}-\partial(\expr{f})_\gamma(\accept)} & \text{if } i \in \supp(\partial(\expr{f})_\gamma)\cap (\V \cap \{\reject\}) \\
			\frac{\partial(\expr{f})_\gamma(\action{a},\expr{f'})}{\prob{1}- \partial(\expr{f})_\gamma(\accept)} & \text{if } i = (\action{a},\expr{f'} \seq \expr{f}^{[\prob{r}]}) \text{ and } (\action{a},\expr{f'}) \in \supp(\partial(\expr{f})_\alpha)
		\end{cases}
	\end{equation*}
    Observe that \(I_\gamma\) is precisely \(\supp\left(\partial\left(\expr{f}^{[\prob{r}]}\right)_\gamma\right)\) and probabilities associated with those elements are precisely the same as those assigned by \(\partial \left(\expr{f}^{[\prob{r}]}\right)_\gamma\). Therefore, we have that
    \[
    \test{\gamma}\seq\expr{f}^{[\prob{r}]} \equiv \test{\gamma} \seq \left(\bigplus{d \in \supp \left(\partial\left(\expr{f}^{[\prob{r}]}\right)_\gamma \right)}\partial\left(\expr{f}^{[\prob{r}]}\right)_\gamma(d) \cdot \ex(d)\right)
    \]
\end{itemize}
\endgroup
\end{proof}

\section{Proofs from \texorpdfstring{\cref{sec: completeness}}{Section~\ref{sec: completeness}}}

To prove \Cref{lem: solutions_are_homomorphisms}, we use the following.

\begin{lemma}[{\cite[Proposition~5.8]{Rutten:2000:Universal}}]\label{lem: quotient_coalgebra}
    Let \((X, \beta)\) be a \acro{ProbGKAT} automaton, \(R\subseteq X \times X\) a bisimulation equivalence on \((X, \beta)\) and \([-]_R : X \to X / R\) the canonical quotient map. Then, there exists a unique transition map \(\bar{\beta} : {X} /{R} \times \At \to \D(\2 + \V + \Act \times {X} /{R})\) which makes
\([-]_R\) into a \acro{ProbGKAT} automaton homomorphism from \((X, \beta)\) to \((X / R , \bar{\beta})\).
\end{lemma}

\begin{proof} \emph{(of \cref{lem: solutions_are_homomorphisms}).}
    First, recall the following.
    Let \(\bar{\partial} : \Exp/{\equiv} \to \G(\Exp/{\equiv})\) be the unique \(\G\)-coalgebra structure map from \cref{lem: quotient_coalgebra} making the quotient map \([-]_\equiv : \Exp \to \Exp/{\equiv}\) into a \(\G\)-coalgebra homomorphism. Given a function \(h : X \to \Exp\), the composite map \([-]_\equiv \circ h\) is a \(\G\)-coalgebra homomorphism from \((X, \beta)\) to \((\Exp/{\equiv}, \bar{\partial})\) if and only if the following diagram commutes
    \begin{equation}
        \label{eq:appE solution}
        \begin{tikzcd}
            X && {\Exp/{\equiv}} \\
            \G X && {\G(\Exp/{\equiv})}
            \arrow["\beta"', from=1-1, to=2-1]
            \arrow["{\bar{\partial}}", from=1-3, to=2-3]
            \arrow["{[-]_\equiv \circ h}", from=1-1, to=1-3]
            \arrow["{\G([-]_\equiv \circ h)}"', from=2-1, to=2-3]
        \end{tikzcd}
        \qquad
        \bar{\partial} \circ [-]_\equiv \circ h = \G([-]_\equiv \circ h) \circ \beta
    \end{equation}
    Now, since \([-]_\equiv\) is a coalgebra homomorphism, \(\bar\partial \circ [-]_\equiv = G[-]_\equiv \circ \partial\).
    This tells us that~\eqref{eq:appE solution} is equivalent to
    \[
        \G[-]_\equiv \circ \partial \circ h
        = \G([-]_\equiv \circ h) \circ \beta
    \]
    We can spell out the above equation in the more concrete terms: for all \(x \in X\) and \(\alpha \in \At\),
    \begin{enumerate}
        \item \(\beta(x)_\alpha(o)=\partial(h(x))_\alpha(o)\) for any \(o \in \2 + \V\), and
        \item for any \((\action{a}, [\expr{e'}]) \in \Act \times \Exp/{\equiv}\), \[
            \sum_{h(x') \equiv \expr{e'}}\beta(x)_\alpha(\action{a}, x') = \sum_{\expr{f}\equiv\expr{e'}} \partial(h(x))_\alpha(\action{a}, \expr{f})
        \]
    \end{enumerate}

    Going from right to left, assume that \([-]_\equiv \circ h\) is a homomorphism.
    Let \((X, \tau)\) be the Salomaa system associated with \((X, \beta)\), and let \(h^\#\) be the unique extension of \(h\) to expressions in \(\Exp(X)\).
    We want to show that \(h(x) \equiv (h^{\#} \circ \tau)(x)\) for all \(x \in X\).
    By \cref{thm: fundamental_theorem} we have that for all \(x \in X\),
    \[
        h(x) \equiv \bigsum{\alpha \in \At} \left(\bigplus{d \in \supp(\partial(h(x))_\alpha)} \partial(h)_\alpha(d) \cdot \ex(d)\right)
    \]
    We can unroll the convex sum and use \cref{lem:generalised_idemponency_up_to_equiv} to group together elements of the support which are equivalent up to \(\equiv\). We obtain the following for all \(x \in X\)
    \begin{equation}
        \label{eqn: solutions_lhs}
        h(x) \equiv \bigsum{\alpha \in \At} \left(\bigplus{d \in \2+\V} \partial(h(x))_\alpha(d) \cdot d \boplus \bigplus{(a,[\expr{e'}]_\equiv) \in E_\alpha} \left(\sum_{\expr{f}\equiv\expr{e'}}\partial(h(x))_\alpha(\action{a},\expr{f})\right)\cdot\action{a}\seq\expr{e'} \right)
    \end{equation}
    where for each \(\alpha \in \At\) we have that
    \[
        E_\alpha = \{(\action{a}, Q) \in \Act \times \Exp/{\equiv} \mid \expr{e'}\in Q, \partial(h(x))_\alpha(\action{a},\expr{e'})>\prob{0}\}
    \]
    Similarly, the right hand side of the equation can written as the following for all \(x \in X\)
    \[
        (h^{\#} \circ \tau)(x) \equiv \bigsum{\alpha \in \At} \left(\bigplus{d \in \2+\V} \beta(x)_\alpha(d) \cdot d \boplus \bigplus{(\action{a},\expr{f})\in \supp(\beta(x)_\alpha)} \beta(x)_\alpha(\action{a},{f}) \cdot \action{a}\seq h({f})\right)
    \]
    We use \cref{lem:generalised_idemponency_up_to_equiv} to group together elements of the support which are equivalent up to \(\equiv\) and sum their associated probabilities
    \begin{equation}
        \label{eqn: solutions_rhs}
        (h^{\#} \circ \tau)(x)
        \equiv
        \bigsum{\alpha \in \At} \left(
            \bigplus
            {d \in \2+\V}
                \beta(x)_\alpha(d) \cdot d
            \boplus
            \bigplus
            {(\action{a}, [\expr{e'}]_\equiv) \in F_\alpha}
            \left(\sum_{h(x')\equiv\expr{e'}}\beta(x)_\alpha(\action{a},x')\right)\cdot\action{a}\seq\expr{e'}
        \right)
    \end{equation}
    where for each \(\alpha \in \At\) we define
    \[
        F_\alpha = \{(\action{a}, Q) \in \Act \times \Exp/{\equiv} \mid h(x') \in Q, \beta(x)_\alpha(\action{a},x')>\prob{0}\}
    \]
    We show that for each \(x \in X\), the expressions from~\eqref{eqn: solutions_lhs} and~\eqref{eqn: solutions_rhs} are precisely the same.
    We do this by arguing that for each atom \(\alpha \in \At\), the index sets of generalised convex sums are the same and so are the probabilities associated with equivalence classes of elements of both index sets.

    Fix an arbitrary \(\alpha \in \At\).
    Given \(d \in \2 + \V\), if \(d \in \supp(\partial(h(x)))_\alpha\), then because of condition (1) of being a homomorphism, \(d \in \supp(\beta(x)_\alpha)\).
    The converse holds via a symmetric argument.
    Moreover, for all such \(d\), we have that \(\partial(h(x))_\alpha(d) = \beta(x)_\alpha(d)\).

    Now consider \((\action{a}, Q) \in \Act \times \Exp/{\equiv}\).
    If \((a,Q) \in E_\alpha\), then there exists some \(\expr{e'} \in \Exp\) such that \(\partial(h(x))_\alpha(\action{a},\expr{e'})>\prob{0}\).
    Because of condition (2) of being a homomorphism, there exists some \(x' \in X\) satisfying \(h(x')\equiv\expr{e'}\) and \(\beta(x)_\alpha(\action{a}, x') > \prob{0}\), hence \((\action{a}, Q) \in F_\alpha\).
    The converse follows via a symmetric argument and therefore \(E_\alpha = F_\alpha\).
    Observe that because of condition (2) the probabilities associated with each \((\action{a}, Q) \in \Act \times \Exp/{\equiv}\) are the same.

    Since all generalised convex sums in both expressions assign the same probabilities to the same elements, we obtain \(h(x) \equiv (h^{\#} \circ \tau)(x)\) for all \(x \in X\) as desired.

    Left to right, assume that \(h\) is a solution up to \(\equiv\).
    We first show that \((\partial \circ h^{\#} \circ \tau)(x) = (\G(h) \circ \beta)(x)\) for all \(x \in X\).
    Recall that for all \(x \in X\),
    \[
        (h^{\#} \circ \tau)(x)
        \equiv \bigsum{\alpha \in \At} \left(\bigplus{d \in \2+\V} \beta(x)_\alpha(d) \cdot d \boplus \bigplus{(\action{a},{f})\in \supp(\beta(x)_\alpha)} \beta(x)_\alpha(\action{a},{f}) \cdot \action{a}\seq h({f})\right)
    \]
    Applying the transition map to that expression can be split into two cases for all \(\alpha \in \At\).
    \begin{enumerate}
        \item For \(o \in \2 + \V\), \((\partial \circ h^{\#} \circ \tau)(x)(o) = \beta(x)_\alpha(o) = (\G(h) \circ \beta)(x)_\alpha(o)\)
        \item For all \((\action{a}, \expr{e'}) \in \Act \times \Exp\),
        \[
            (\partial \circ h^{\#} \circ \tau)(x)(\action{a}, \expr{e'}) = \sum_{h(x')=\expr{e'}} \beta(x)_\alpha(\action{a},x') = (\G(h) \circ \beta)(x)_\alpha(\action{a}, \expr{e'})
        \]
    \end{enumerate}
    which proves that \((\partial \circ h^{\#} \circ \tau) = (\G(h) \circ \beta)\).
    Now postcompose both sides with \(G[-]_\equiv\) and consider the following
    \begin{align*}
    \G[-]_\equiv \circ \partial \circ h^{\#} \circ \tau
    &= \G[-]_\equiv \circ \G(h) \circ \beta\\
    \G[-]_\equiv \circ \partial \circ h^{\#} \circ \tau
    &= \G([-]_\equiv \circ h) \circ \beta \tag{\(\G\) is a functor}\\
    \bar{\partial} \circ [-]_\equiv \circ h^{\#} \circ \tau
    &= \G([-]_\equiv \circ h) \circ \beta \tag{\([-]_\equiv\) is a homomorphism}\\
    \bar{\partial} \circ [-]_\equiv \circ h
    &= \G([-]_\equiv \circ h) \circ \beta \tag{\(h \equiv h^{\#} \circ \tau\)}
    \end{align*}
    And therefore \([-]_\equiv \circ h\) is a \(\G\)-coalgebra homomorphism from \((X, \beta)\) to \((\Exp/{\equiv}, \bar{\partial})\) as desired.
\end{proof}
\metricwelldefined*
\begin{proof}
    \(d(x,x) = 0\) holds immediately since the greatest bisimulation on \((X, \beta)\) is an equivalence relation and therefore for all \(x \in X\) we have that \(x \bisim x\).

    To show symmetry, assume \(d(x,y)=k\) for some \(k \in \mathbb{R}^{+}\).
    If \(k = 0\) then \(x \bisim_{\beta} y\), but since the greatest bisimulation on \((X, \beta)\) is an equivalence relation (by preservation of weak pullbacks for \(G\)), we have that \(y \bisim_\beta x\), which implies that \(d(y,x) = 0\).
    For the remaining case, we can assume without loss of generality that \(k  = 2^{-n}\) for some \(n \in \mathbb N\).
    In such a case, \(x\bisim_\beta^{(n)} y\), but again by \cref{lem:preservation_of_equivalences}, each of the stratified bisimilarity relations are equivalence relations.
    Therefore, \(y\in \bisim_\beta^{(n)} x\) and \(n\) is the largest number such that \(y \in \bisim_\beta^{(n)} x\).
    It follows that \(d(y,x)=2^{-n}=k=d(x,y)\) as desired.

	For the strengthened triangle inequality, first consider the case where \(d(x,y)=2^{-k}\) and \(d(y,z)=2^{-l}\).
    These imply \(x\in \bisim_{\beta}^{(k)} y\) and \(y\in \bisim_{\beta}^{(l)} z\) respectively.
    We have that \((x,y), (y,z) \in {\bisim_{\beta}^{(\min(k,l))}}\).
    Since all the stratified bisimilarity relations are equivalence relations, it follows from transitivity that \(x \in {\bisim_{\beta}^{(\min(k,l))}} z\) and \(n = \min(k,l)\) is the largest natural number for which \(x \in \bisim_\beta^{(n)} z\).
    Therefore, \(d(x,z) = 2^{-\min(k,l)} = \max(2^{-k}, 2^{-l}) = \max(d(x,y), d(y,z))\).

    In case \(d(x,y) = d(y,z) = 0\), we have \(x \bisim_\beta y\) and \(y \bisim_\beta z\).
    Since the greatest bisimulation is an equivalence relation, we have that \(x \bisim_\beta z\), which implies that \(d(x,y) = 0 = \max(d(x,y), d(y,z))\).

    For the next case, let \(d(x,y)=0\) and \(d(y,z)=2^{-k}\).
    Then \(x \bisim_\beta y\) and \((y,z) \in \bisim_\beta^{(k)}\).
    In particular, we also have that \((x,y) \in \bisim_\beta^{(k)}\).
    By transitivity, \(x \in \bisim_\beta^{(k)} z\), and \(k\) is the largest number such that \(x \in \bisim_\beta^{(k)} z\).
    Therefore, \(d(x,z)=2^{-k}=\max(d(x,y), d(y,z))\)

	The only remaining case is symmetric.
\end{proof}

To prove \Cref{lem: contractiveness}, we first prove that substitution in Salomaa systems is well-behaved w.r.t.\ the operator used in the definition of the relation refinement chain on \((\Exp, \partial)\).

\begin{restatable}{lemma}{salomaasubstitution}\label{lem: salomaa_well_behaved}
    Let \((X, \tau: X \to \Exp(X))\) be a Salomaa system and let \(R \subseteq \Exp \times \Exp\) be a congruence.
    Consider two maps \(m, n : X \to \Exp\) satisfying \((m(x), n(x)) \in R\) for all \(x \in X\). Then,
    \(
        ((m^{\#} \circ \tau)(x), (n^{\#} \circ \tau)(x)) \in \Phi_\partial(R)
    \)
    for all \(x \in X\).
\end{restatable}

\begin{proof}
    Without loss of generality we can assume that for all \(x \in X\),
    \[\tau(x) = \bigsum{\alpha \in \At} \left( \bigplus{i \in I_{\alpha}^{x}} \prob{p}_i \cdot \expr{f}_i \boplus \bigplus{j \in J_\alpha^{x}} \prob{r}_j \cdot \expr{g}_j x_j \right)\]
    For any \(x \in X\) and \(\alpha \in \At\), \(I_\alpha^{x}\) and \(J_\alpha^{x}\) are disjoint index sets such that
    \begin{itemize}
        \item \(\{\prob{p}_i\}_{i \in I_\alpha^{x}}\) and \(\{\prob{r}_j\}_{j \in J_\alpha^{x}}\) are indexed collections of probabilities satisfying
            \[
            \sum_{i \in I_\alpha^{x}} \prob{p}_i + \sum_{j \in J_\alpha^{x}} \prob{r}_j = \prob{1}
            \]
        \item \(\{\expr{f}_i\}_{i \in I_\alpha^{x}}\) and \(\{\expr{g}_j\}_{j \in J_\alpha^{x}}\) are indexed collections of elements of \(\Exp\). Moreover, since \((X, \tau)\) is Salomaa, we have that \(\trmt{\expr{g}_{j}}=0\) for all \(j \in J_\alpha^{x}\). Because of \cref{lem:no_termination_implies_success} we have that \(\partial(\expr{g}_j)_\alpha(\accept)=\prob{0}\) for all for all \(j \in J_\alpha^{x}\).
        \item \(\{x_j\}_{j \in J_\alpha^x}\) is an indexed collection of indeterminates (elements of \(X\)).
    \end{itemize}
    Observe that for all \(x \in X\),
    \begin{align*}
        (m^{\#} \circ \tau)(x) &= \bigsum{\alpha \in \At} \left( \bigplus{i \in I_{\alpha}^{x}} \prob{p}_i \cdot \expr{f}_i \boplus \bigplus{j \in J_\alpha^{x}} \prob{r}_j \cdot \expr{g}_j\seq m(x_j) \right)
        \\
        (n^{\#} \circ \tau)(x) &= \bigsum{\alpha \in \At} \left( \bigplus{i \in I_{\alpha}^{x}} \prob{p}_i \cdot \expr{f}_i \boplus \bigplus{j \in J_\alpha^{x}} \prob{r}_j \cdot \expr{g}_j\seq n(x_j) \right)
    \end{align*}
    For the rest of the argument, we use \cref{lem:preservation_of_equivalences}.
    Pick an arbitrary \(x \in X\) and \(\alpha \in \At\).
    \begin{enumerate}
        \item For \(o \in \{\reject\} + \V\), we have that
        \begin{align*}
            \partial((f^{\#} \circ m)(x))_\alpha(o) &= \sum_{i \in I_\alpha^{x}} \prob{p}_i \partial(\expr{f}_i)_\alpha(o) + \sum_{j \in J_\alpha^{x}} \prob{r}_j \partial(\expr{g_j}\seq m(x_j))_\alpha(o)\\
            &=\sum_{i \in I_\alpha^{x}} \prob{p}_i \partial(\expr{f}_i)_\alpha(o) + \sum_{j \in J_\alpha^{x}} \prob{r}_j (\partial(\expr{g_j})_\alpha(o) + \partial(\expr{g}_j)_\alpha(\accept)\partial(m(x_j))_\alpha(o))\\
            &=\sum_{i \in I_\alpha^{x}} \prob{p}_i \partial(\expr{f}_i)_\alpha(o) + \sum_{j \in J_\alpha^{x}} \prob{r}_j \partial(\expr{g_j})_\alpha(o) \tag{\(\trmt{\expr{g}_j}=0\)}\\
            &=\sum_{i \in I_\alpha^{x}} \prob{p}_i \partial(\expr{f}_i)_\alpha(o) + \sum_{j \in J_\alpha^{x}} \prob{r}_j(\partial(\expr{g_j})_\alpha(o) + \partial(\expr{g}_j)_\alpha(\accept)\partial(n(x_j))_\alpha(o)) \\
            &= \partial((f^{\#} \circ n)(x))_\alpha(o)
        \end{align*}
        Now, consider the case when \(o = \accept\). We have the following
        \begin{align*}
            \partial((f^{\#} \circ m)(x))_\alpha(\accept) &= \sum_{i \in I_\alpha^{x}} \prob{p}_i \partial(\expr{f}_i)_\alpha(\accept) + \sum_{j \in J_\alpha^{x}} \prob{r}_j \partial(\expr{g_j}\seq m(x_j))_\alpha(\accept) \\
            &=\sum_{i \in I_\alpha^{x}} \prob{p}_i \partial(\expr{f}_i)_\alpha(o) + \sum_{j \in J_\alpha^{x}} \prob{r}_j \partial(\expr{g}_j)_\alpha(\accept)\partial(m(x_j))_\alpha(\accept) \\
            &=\sum_{i \in I_\alpha^{x}} \prob{p}_i \partial(\expr{f}_i)_\alpha(o) \tag{\(\trmt{\expr{g}_j}=0\)}\\
            &=\sum_{i \in I_\alpha^{x}} \prob{p}_i \partial(\expr{f}_i)_\alpha(o) + \sum_{j \in J_\alpha^{x}} \prob{r}_j \partial(\expr{g}_j)_\alpha(\accept)\partial(n(x_j))_\alpha(\accept) \\
            &=\partial((f^{\#} \circ n)(x))_\alpha(\accept)
        \end{align*}
        \item Let \((\action{a}, Q) \in \Act \times {\Exp}/{R}\). We have that
        \begingroup%
        \allowdisplaybreaks%
        \begin{align*}
            &\quad\partial((f^{\#} \circ m)(x))_\alpha[\{\action{a}\}\times Q] \\&= \sum_{i \in I_\alpha^{x}} \prob{p}_i \partial(\expr{f}_i)_\alpha[\{\action{a}\}\times Q] + \sum_{j \in J_\alpha^{x}} \prob{r}_j \partial(\expr{g_j}\seq m(x_j))_\alpha[\{\action{a}\}\times Q]\\
            &= \sum_{i \in I_\alpha^{x}} \prob{p}_i \partial(\expr{f}_i)_\alpha[\{\action{a}\}\times Q]
                \\&\hspace{5em}+ \sum_{j \in J_\alpha^{x}} \prob{r}_j \left(\partial(\expr{g_j})_\alpha[\{\action{a}\}\times Q/{m(x_j)}] + \partial(\expr{g}_j)_\alpha(\accept)\partial(m(x_j))_\alpha[\{\action{a}\}\times Q]\right)\\
            &= \sum_{i \in I_\alpha^{x}} \prob{p}_i \partial(\expr{f}_i)_\alpha[\{\action{a}\}\times Q] + \sum_{j \in J_\alpha^{x}} \prob{r}_j\partial(\expr{g_j})_\alpha[\{\action{a}\}\times Q/{m(x_j)}] \tag{\(\trmt{\expr{g}_j}=0\)}\\
            &= \sum_{i \in I_\alpha^{x}} \prob{p}_i \partial(\expr{f}_i)_\alpha[\{\action{a}\}\times Q] + \sum_{j \in J_\alpha^{x}} \prob{r}_j\partial(\expr{g_j})_\alpha[\{\action{a}\}\times Q/{n(x_j)}] \tag{\cref{lem:swapping_congruent_ends}}\\
            &= \sum_{i \in I_\alpha^{x}} \prob{p}_i \partial(\expr{f}_i)_\alpha[\{\action{a}\}\times Q]
                \\&\hspace{5em}+ \sum_{j \in J_\alpha^{x}} \prob{r}_j \left(\partial(\expr{g_j})_\alpha[\{\action{a}\}\times Q/{m(x_j)}]
                + \partial(\expr{g}_j)_\alpha(\accept)\partial(n(x_j))_\alpha[\{\action{a}\}\times Q]\right)\\
            &=\partial((f^{\#} \circ n)(x))_\alpha[\{\action{a}\}\times Q]\qedhere
        \end{align*}
        \endgroup%
    \end{enumerate}

\end{proof}

\contractivess*

\begin{proof}
    Observe that all elements of the relation refinement chain are congruences.
    This allows us to use \cref{lem: salomaa_well_behaved}.
    Recall that given \(e,f \in \Exp^n\) such that \(e = (\expr{e}_1, \dots, \expr{e}_n)\) and \(f = (\expr{f}_1, \dots, \expr{f}_n)\), we have
    \[
        d((\expr{e}_1, \dots, \expr{e}_n),(\expr{f}_1, \dots, \expr{f}_n))=\max\{d(\expr{e}_1, \expr{f}_1), \dots, d(\expr{e}_n, \expr{f}_n)\}
    \]
	As before, we slightly abuse notation and also write \(e,f : X \to \Exp\) for functions \(e(x_k)=\expr{e}_k\) and \(f(x_k)=\expr{f}_k\) for \(1\leq k \leq n\).

	Assume that \(d((\expr{e}_1, \dots, \expr{e}_n), (\expr{f}_1, \dots, \expr{f}_n)) \le 2^{-m}\) for some \(m \in \mathbb N\).
	By definition of the product metric, \(2^{-m} \ge \max\{d(\expr{e}_1, \expr{f}_1), \dots, d(\expr{e}_n, \expr{f}_n)\}\), so \(\expr{e}_i \bisim_\partial^{(m)} \expr{f}_i\) for \(1 \leq i \leq n \).
    By \cref{lem: salomaa_well_behaved}, \(((e^{\#} \circ \tau)(x_i), (f^{\#} \circ \tau)(x_i))\in \Phi_\partial\left( \bisim_\partial^{(m)}\right) = {\bisim_\partial^{(m+1)}}\) for all \(1 \leq i \leq n\).
    Therefore, \(d((e \circ \tau)(x_i), (f \circ \tau)(x_i)) \le 2^{-(m+1)}\) for all \(1 \leq i \leq n\).
    It follows that \[
        d(\bar{\tau}(\expr{e}_1, \dots, \expr{e}_n), \bar{\tau}(\expr{f}_1, \dots, \expr{f}_n))
        \le 2^{-(m+1)}
        = \frac{1}{2} \left(\frac1{2^m}\right)
    \]
    It follows that \(\bar \tau\) is \(\frac12\)-Lipschitz.
\end{proof}


\end{document}